\documentclass[12pt]{article}

\usepackage{amssymb,amsmath,amsthm}
\usepackage[colorlinks,citecolor=blue,urlcolor=blue]{hyperref}

\usepackage{cases}
\usepackage{graphicx}
\usepackage{epstopdf}
\usepackage{algorithm}
\usepackage{algpseudocode}
\usepackage[round]{natbib}
\usepackage{color}

\usepackage[toc,page,title,titletoc,header]{appendix}
\usepackage{multirow}
\usepackage{booktabs}
\usepackage{setspace}
\usepackage{url}
\usepackage{arydshln}
\usepackage[overload]{empheq}
\usepackage{cleveref}

\usepackage[normalem]{ulem}

\DeclareMathOperator*{\argmin}{arg\,min} 

\usepackage{caption}
\usepackage{subcaption}

\usepackage{amsfonts}
\usepackage{bibentry}  
\nobibliography*

\newtheorem{theorem}{Theorem}
\newtheorem{lemma}{Lemma}

\newtheorem{definition}{Definition}[section]
\newtheorem{example}{Example}[section]
\newtheorem{remark}{Remark}
\newtheorem{corollary }{Corollary }

\graphicspath{ {figures/} }

\newcommand{\rd}{\,\mathrm{d}}
\newcommand{\EP}{\,\mathbb{P}}

\newcommand{\rhob}{B}

\newcommand{\nd}{A}
\newcommand{\vc}{\mathcal{V}}
\newcommand{\dd}{D}
\newcommand{\kk}{K}
\newcommand{\cc}{\chi_{0}}

\newcommand{\ev}{A}
\newcommand{\unv}{\zeta}

\newcommand{\betadiff}{\xi}
\newcommand{\mtbb}{\varrho}
\newcommand{\etaa}{\tau}

\DeclareMathAlphabet{\mathpzc}{OT1}{pzc}{m}{it}

\newcommand{\lbb}{\mathfrak{L}}
\newcommand{\lbbd}{\mathfrak{L}'}

\newcommand{\rhotwoinf}{\kappa_{2,\infty}}
\newcommand{\rhotwo}{\kappa_{2}}
\newcommand{\deltay}{\mathpzc{Y}}

\DeclareMathOperator*{\sign}{sign}

\usepackage[top=1in, bottom=1in, left=1in, right=1in]{geometry}

\title{Skewed Pivot-Blend Modeling with Applications to Semicontinuous Outcomes
}
\date{ }
\author{Yiyuan She, Xiaoqiang Wu, Lizhu Tao, and Debajyoti Sinha
}

\begin{document}

\maketitle

\begin{abstract}

Skewness is a common occurrence in statistical applications. In recent years, various distribution families have been proposed to model skewed data by introducing unequal scales based on the median or mode. However, we argue that the point at which unbalanced scales occur may be at any quantile and cannot be reparametrized as  an ordinary shift parameter in the presence of skewness.  In this paper, we introduce a novel skewed pivot-blend technique  to create a skewed density family based on any continuous density, even those that are asymmetric and nonunimodal. Our framework enables the simultaneous estimation of scales, the pivotal point, and other location parameters, along with various extensions. We also introduce a skewed two-part model tailored for semicontinuous outcomes,   which identifies relevant variables across the entire population and mitigates the additional skewness induced by commonly used transformations. Our theoretical analysis   reveals the influence of skewness without assuming asymptotic conditions. Experiments on synthetic and real-life data demonstrate the excellent performance of the proposed method.

\end{abstract}

\textbf{Keywords}: semicontinuous outcomes; skewed data; two-piece densities; two-part models;  variable selection; composite models.

\section{Introduction}
\label{sec:intro}
Statisticians frequently encounter skewed data in biomedical, econometric, environmental, and social research. Commonly used  models, such as linear regression, least absolute deviations, and robust  regression,  presume symmetric errors and are prone to significant distortions when confronted with skewness.  To mitigate the issue, many researchers   prefer transforming the data beforehand, with logarithmic-type transformations being among the most popular choices. 
Alternatively, some researchers use modal regression \citep{lee1989mode} or median-based methods, which are less sensitive to the assumption of symmetric errors. However,  these approaches do not explicitly account for and describe skewness.

To comprehensively address this issue, adopting a ``joint" modeling approach becomes essential and beneficial. This paper simultaneously estimates location, scale, and skewness parameters, thus avoiding the risk of either concealing true skewness (\textit{masking}) or erroneously detecting spurious skewness (\textit{swamping}). This risk is present when using a stepwise procedure, such as fitting a modal regression and then assessing skewness based on residuals  \citep{boos1987}. Our primary aim   is not only to accommodate skewness, as many papers do, but to \textit{explicitly} capture and  characterize its effects.

Various distributions have been proposed in the literature for  modeling skewed data. \cite{azzalini1985class} proposed a skewed density family including the skewed normal density as an example.
\cite{fernandez1998bayesian} proposed a  two-piece skewed distribution family that sets the mode at zero, including the skewed Student  and Laplace distributions  for Bayesian quantile regression \citep{arellano2005statistical,yu2001bayesian}.  \cite{rubiosteel} extended the family by use of   two scale parameters and additional shape parameters. \cite{kottas2001bayesian} described an alternative two-piece skewed distribution family  that   keeps the median at zero, but the resulting density is discontinuous. For a historical account of two-piece distributions, interested readers may consult \cite{rubio2020family}.

The existing constructions  rely on a symmetric and unimodal raw density, introducing asymmetric scales based on either the mode or the median of the raw density. However, in numerous real-life applications, these assumptions may not hold. Particularly, the point at which skewness is enforced, termed the ``\textbf{pivotal point}" in this paper, could be situated at any position or quantile. Intriguingly, this pivotal point   distinguishes  itself from the commonly used shift parameter, as opposed to  the prevailing assumption in the existing literature. To overcome these limitations, there is a demand for a novel skewed distribution family that offers flexibility, continuity, and  adaptability to any pivotal point of interest.

\begin{figure}[htp!]
\includegraphics[width=0.48\textwidth, height=2.5in]{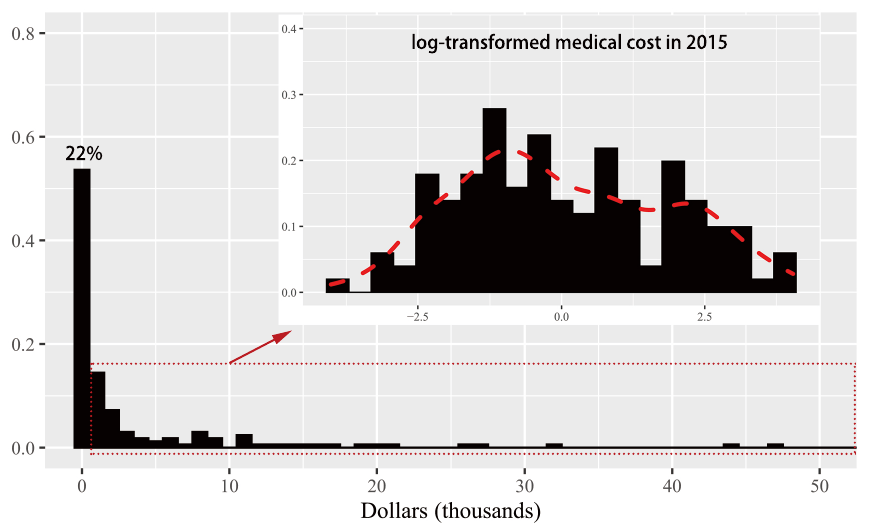}
\includegraphics[width=0.48\textwidth, height=2.5in]{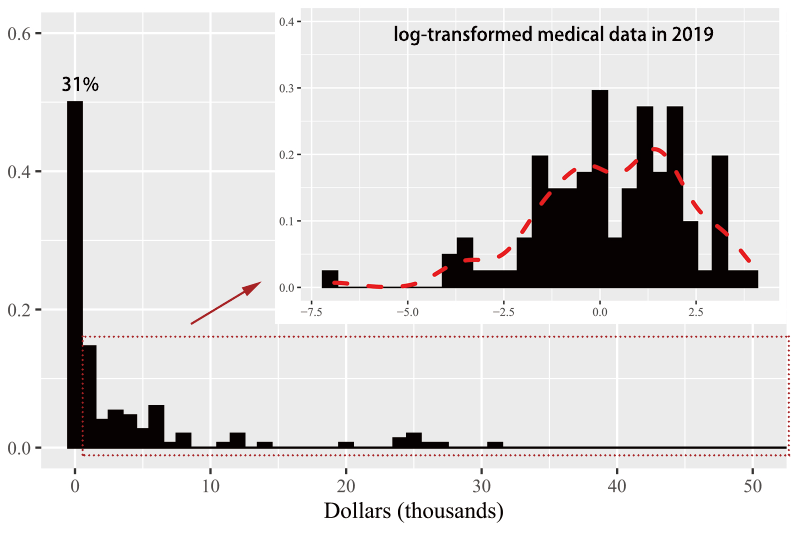}
\centering
\caption{\footnotesize Some histograms of  medical expenditures. Left: stratum ID 1098 (first PSU) of MEPS 2015, right: stratum ID 2109  (third PSU) of MEPS 2019.  In each example, the main plot  shows an excessive portion of zeros; the top right panel, excluding the zeros, plots the  log-transformed positive values of response (with estimated density  in red), which  still exhibit asymmetry despite the transformation.    \label{mepsraw}}
\end{figure}

This study draws inspiration from the Medical Expenditure Panel Survey (MEPS) data, which is obtained from national surveys investigating the impact of various demographic variables on the medical expenses of patients in the United States. Notably, this dataset features a skewed response that includes a significant number of zeros, a phenomenon known as ``\textbf{semicontinuous outcomes}" in the realms of economics and longitudinal studies \citep{olsen2001two}. To provide a visual representation, we employed two datasets \citep{us2015medical,us2019medical}, as depicted in Figure \ref{mepsraw}, for illustration.

According to Figure  \ref{mepsraw}, more than    $20\%$ patients have zero medical expenditure,  while the remaining exhibit highly skewed positive medical costs. Given that these zeros represent   precisely zero medical expenses, rather than truncation, a \textit{two-part} (or \textit{hurdle}) model \citep{mullahy1998much} is a more appropriate choice   than the \textit{Tobit} model \citep{tobin1958estimation}. In this approach, the binary part of the model captures zero-nonzero patterns, while the continuous part of the model addresses strictly positive outcomes.  However, it is important to highlight that applying a standard log-normal two-part model may not yield sufficient power, owing to the asymmetry depicted in the upper-right panels of  Figure \ref{mepsraw}.
We have frequently observed that conventional transformations, such as logarithmic or power functions, not only fail to entirely eliminate skewness but also introduce nontrivial points around which asymmetric scales arise.    Consequently, there may be a necessity to ``reinforce" the transformed model to counteract the  skewness effectively.

Another closely related challenge within the context of MEPS data analysis involves developing an  {interpretable} two-part model. This entails the identification of a subset of medical cost-relevant predictors that apply to the \textit{entire} population, serving as valuable guidance for policymakers. To the best of our knowledge, very few existing two-part models have considered the issue of joint variable selection, wherein each predictor can contribute to the response in a  composite  manner  through the binary and continuous parts.

This paper attempts to address some aforementioned challenges for possibly   skewed, semicontinuous outcomes. Our contributions are as follows.
\begin{enumerate}
\item
We introduce a novel    skewed pivotal-point adaptive  family, designed  to  infuse skewness around   {an unknown} pivotal point. The key ``\textbf{skewed  pivot-blend}'' technique is versatile and can be applied  to any raw density, regardless of its symmetry or unimodality. The resulting density remains continuous and  accommodates  many previous proposals. \item We introduce the \textbf{SPEUS} framework (Skewed Pivot-Blend Estimation with Unsymmetric Scales) for simultaneous estimation of scales, pivotal point, and other location parameters. This  framework offers  useful variants, especially for     modeling  semicontinuous outcomes  with joint variable selection. The resulting two-part method is capable of identifying relevant variables across the entire population and concurrently addressing the excessive skewness introduced by imperfect transformations


\item

We conduct nonasymptotic   analysis for sparse skewed two-part models, utilizing a notion of effective noise and  Orlicz norms to derive sharp statistical error bounds in the presence of skewness and heavy tails. Our work quantifies  how skewness and tail decay impact regularization parameters, prediction and estimation errors.\end{enumerate}

\paragraph{Notations and symbols.}  Given two vectors $\alpha,\beta\in \mathbb{R}^n$, their inner product is $\langle\alpha,\beta\rangle=\alpha^T\beta$ and  their elementwise product is denoted by the vector $\alpha\circ\beta$. Given a scalar function $l$ and a vector $a$, $l(a)=[l(a_i)]^n_{i=1}$, i.e., $l$ is applied componentwise. Throughout the paper, we use $1_A(x)$ to denote the indicator function of $A$,  taking 1 if $x\in A$ and 0 otherwise. In particular,     given any vector $a\in\mathbb{R}^n$, define two indicator vectors
$
    1_{+}(a)=\left[1_{a_i>0}\right]^n_{i=1}, 1_{-}(a)=\left[1_{a_i<0}\right]^n_{i=1}.
$ 
Define $\mathbb R_+ = [0, \infty)$. Given a continuous density $f$ (with respect to the Lebesgue measure $\mu$), we use $f(\cdot| A)$ to denote the conditional density given $A$, or $f(\cdot | A) = \frac{f(\cdot)}{\int_A f \rd \mu}$.
 Given any matrix $A=[a_1,\ldots,a_p]^T\in \mathbb{R}^{p\times m}$, its spectral norm and Frobenius norm are denoted by $\|A\|_2$ and $\|A\|_F$, respectively.  The (2,1)-norm of $A$ is defined as $\|A\|_{2,1}=\sum^{p}_{j=1}\|a_j\|_2$.   We use $A_k$ to denote the $k$th column of $A$. 
Given $a,b\in \mathbb{R}$, we use the shorthand notation $a\lor b$ ($a\land b$) to denote the maximum (minimum) of $a$ and $b$.


\section{Skewed Pivotal-Blend Estimation}
\label{sec:pb-speus}
\subsection{Skewed Pivot-Blend for Density  Pasting}
\label{subsec:spb}

How to define a skewed distribution family from a  unimodal, continuous, and symmetric  density $\phi$ has attracted a lot of attention in the literature.
\cite{azzalini1985class}   multiplied $\phi$ by a perturbation function to define a so-called ``skewed symmetric distribution" family, one well-known example being  the skewed normal  distribution. We   refer the reader to \cite{nadarajah2003skewed}, \cite{wang2004skew}, and \cite{azzalini2005skew} for   variants and further extensions.  On the other hand, the associated skewed distribution function often lacks an explicit form, and determining its mode can be a challenging task \citep{ma2004flexible}.

``Two-piece'' skewed distributions are  popularly used in recent years. \cite{fernandez1998bayesian} introduced a  two-piece transformation  that rescales $\phi$'s negative   and positive parts differently using  an asymmetry parameter, allowing it to maintain the mode at zero. A  reparametrization   of the approach, following \cite{arellano2005statistical},  includes the skewed Student  and  epsilon-skew-normal distributions \citep{fernandez1998bayesian,Mudholkar:2000}. Later, \cite{rubiosteel} extended this idea to
include two scale parameters (and additional shape parameters). The motivation for our work largely stems from their two-piece form, even though it assumes that the median of $\phi$ is zero.
Another two-piece distribution family due to  \cite{kottas2001bayesian} can guarantee a  median at zero, but the resulting density  is  discontinuous, which  may cause difficulties and instability in parameter  estimation.  Interested readers may refer to \cite{jones2014generating} for a systematic framework of how to construct skewed distributions from a given symmetric density.

Despite the     research in this area, two issues have caught our particular attention and deserve further investigation. Firstly,  the majority of existing works stipulate that $\phi$ should be unimodal and symmetric. However, situations can arise where skewness   manifests when dealing with non-unimodal data. 
There might also be a need for additional reinforcement to counteract skewness, even when employing an asymmetric density. Another more critical concern is that in previous works, the transition point at which unequal scales are imposed, referred to  as the {\textbf{pivotal point}} in this paper, is typically set at the mode or the median. Nevertheless, skewness can persist when the density deviates from  the assumptions above or below \textit{any} quantile, a common occurrence when using an imperfect transformation.


In the following, we introduce a process known as  ``\emph{skewed pivot-blend}''   (or   sometimes pivot-blend for brevity) to  characterize skewness as a combination of both first-order and higher-order statistical effects. We define a versatile two-piece distribution framework with skewness, designed to (i) accommodate any asymmetric or non-unimodal $\phi$, (ii) model skewness associated with any pivotal point of interest, and (iii) maintain continuity. See Figure \ref{fig:forward} for an illustration.

\begin{figure}[!htp]
\centering
\includegraphics[width=1.0\textwidth, height=2.5in]{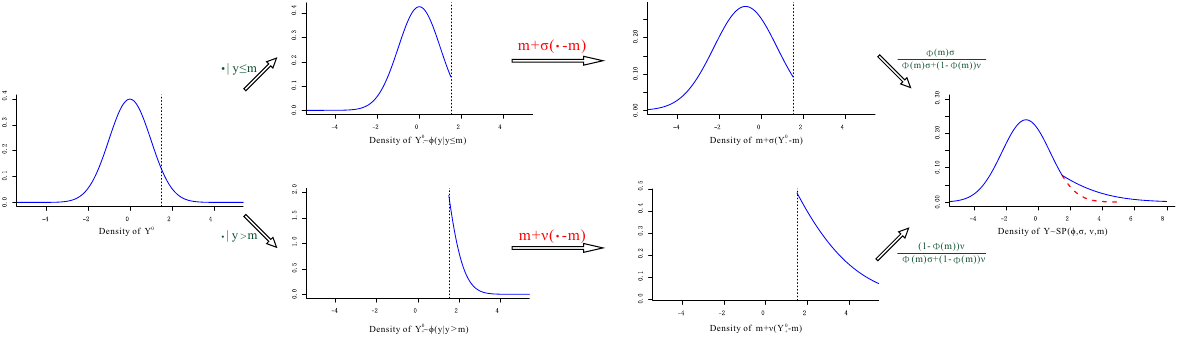}
\caption{\footnotesize   Diagram  showing the process of ``skewed pivot-blend'' for constructing a skewed density: conditioning, affine transformations, and mixing. The two affine transformations ensure that the   cut points remain fixed  (alignment in the \textit{$x$-direction}), and the mixing process guarantees the continuity of the resulting density (alignment in the \textit{$y$-direction}).\label{fig:forward} }
\end{figure}

We provide a step-by-step guide for  constructing a new density function from an arbitrary continuous density (denoted as $\phi$), during which skewness is imposed around a pivotal point $m$ ($0 < \Phi(m) < 1$) and is defined by the left and right scales, $\sigma$ and $\nu>0$.

 a)   \textit{Pivotal-point conditioning}: The density $\phi$  is  conditioned into two separate densities,  one for $y \leq m$,  and the other for $y > m$,  resulting in   $\frac{\phi(y)1_{y \leq m}}{\Phi(m)}$ and $\frac{\phi(y)1_{y > m}}{1 - \Phi(m)}$.

 b) \textit{Affine transformation}:  Apply two separate affine transformations to the random variables associated with the aforementioned densities, concerning the pivotal point $m$:  $m + \sigma(\cdot - m)$ and $m + \nu(\cdot - m)$. The resulting densities are $\frac{\phi\left(\frac{y-m}{\sigma}+m\right)1_{y \leq m}}{\Phi(m)\sigma}$ and $\frac{\phi\left(\frac{y-m}{\nu}+m\right)1_{y > m}}{(1-\Phi(m))\nu}$. It is crucial to emphasize that these transformations are not simple scalings, but are designed to ensure that the cut points of the density functions remain aligned.

 c) \textit{Continuous mixing}: Probability masses $p$ and $1-p$ are assigned to the two  densities obtained from the last step,   resulting in a new density function:
\begin{align}\label{twopiecederivequestion}
    f_{}(y)=  p\times \frac{1}{\Phi(m)\sigma}\phi\big(\frac{y-m}{\sigma}+m\big)1_{y\leq m}+(1-p)\times\frac{1}{\{1-\Phi(m)\}\nu}\phi\big(\frac{y-m}{\nu}+m\big)1_{y>m},
 \end{align}
 where $\Phi$ denotes the distribution function of $\phi$ throughout the paper unless otherwise specified.
   Given that $\phi(m) > 0$ typically holds,  ensuring the continuity of  $f_{}$  at $m$ requires that  $ {p}/({\Phi(m)\sigma})=({1-p})/\{ (1-\Phi(m)\nu\}$, which  leads to a \textbf{unique} choice of $p$:
\begin{align}\label{probdef}
p=\frac{\Phi(m)\sigma}{\Phi(m)\sigma+\{1-\Phi(m)\}\nu}, \mbox{ or } \ \    \frac{\mathbb{P}(Y\leq m)}{\mathbb{P}(Y>m)}=\frac{\Phi(m)\sigma}{\{1-\Phi(m)\}\nu}.
\end{align}
 We sometimes refer to the process as the ``forward'' pivot-blend transform (to contrast with the ``backward'' pivot-blend transform  to be introduced in Remark \ref{conversionremark}). When $\sigma=\nu$ or $m$ is   not in the support  of $\phi$, pivot-blend  operates as a location-scale transformation. Otherwise,  it  serves as a versatile tool for modeling skewed data, encompassing various existing skewed density functions.   Notably, the incorporation of a single pivotal point parameter $m$ substantially improves skewed data modeling in practical applications.

\begin{definition}[Skewed pivot-blend (\textbf{SP})  family]\label{skewdensitydef}
Given a  continuous density $\phi$ and a pivotal point $m$, we say that $Y$ is a skewed random variable with $m$-associated  left- and right-scale parameters $\sigma$ and $\nu$, i.e., $Y\sim \mbox{SP}^{(\phi)}(\sigma,\nu,m)$, if its density is given by
\begin{align}\label{skeweddensity}
   f (y; m, \sigma, \nu)= \frac{\phi\big(\frac{y-m}{\sigma}+m\big)1_{y\leq m}+\phi\big(\frac{y-m}{\nu}+m\big)1_{y>m}}{\Phi(m)\sigma+\{1-\Phi(m)\}\nu}.
\end{align}
\end{definition}

We occasionally write $Y \sim \mbox{SP}^{(\phi)}$ and omit the parameters when there is no ambiguity.
Throughout the paper, we use the term  \textit{skewness} to refer to asymmetric scales ($\sigma\ne \nu$), regardless of the shape   of $\phi$.

The pivotal location     $m$   can be translated to a \textit{pivotal  quantile}   $q$. Let $q=\Phi(m)$, then an equivalent form of \eqref{skeweddensity} is
\begin{align*}
  \frac{\phi\big(\frac{y-\Phi^{-1}(q)}{\sigma}+\Phi^{-1}(q)\big)1_{\Phi(y)\leq   q }+\phi\big(\frac{y-\Phi^{-1}(q)}{\nu}+\Phi^{-1}(q)\big)1_{\Phi(y)>  q }}{q\sigma+(1-q)\nu}.
\end{align*}
For the distribution function $F(y) =  \frac{1}{\sigma \Phi(m) + \nu(1 - \Phi(m))} \{\sigma \Phi(\frac{y-m}{\sigma}+m) 1_{y\le m} +  [\nu \Phi(\frac{y-m}{\nu}+m)+(\sigma - \nu)\Phi(m)]  1_{y> m} \}$,   the new quantile at $m$ is related to the original quantile $q$ by
$
F(m) = \frac{q\sigma}{q\sigma + (1-q)\nu}\lessgtr  q$ when $\sigma \lessgtr \nu$.

When working with the family described   in Definition \ref{skewdensitydef}, it  is a   common practice to add a shift or intercept
$\alpha \in \mathbb{R}$ and assume $Y-\alpha\sim\mbox{SP}^{(\phi)}(\sigma,\nu,m)$; however, it is crucial to note that $\alpha$ and $m$ are   generally   \textbf{not} redundant. This distinction arises because the operations of translation and asymmetric rescaling  utilized in the skewed pivot-blend process do not commute (cf.  Remark \ref{Pivotvsintercept}).
%


\begin{figure}[!ht]
\begin{subfigure}{0.45\textwidth}
  \includegraphics[width=\textwidth, height=3.5cm]{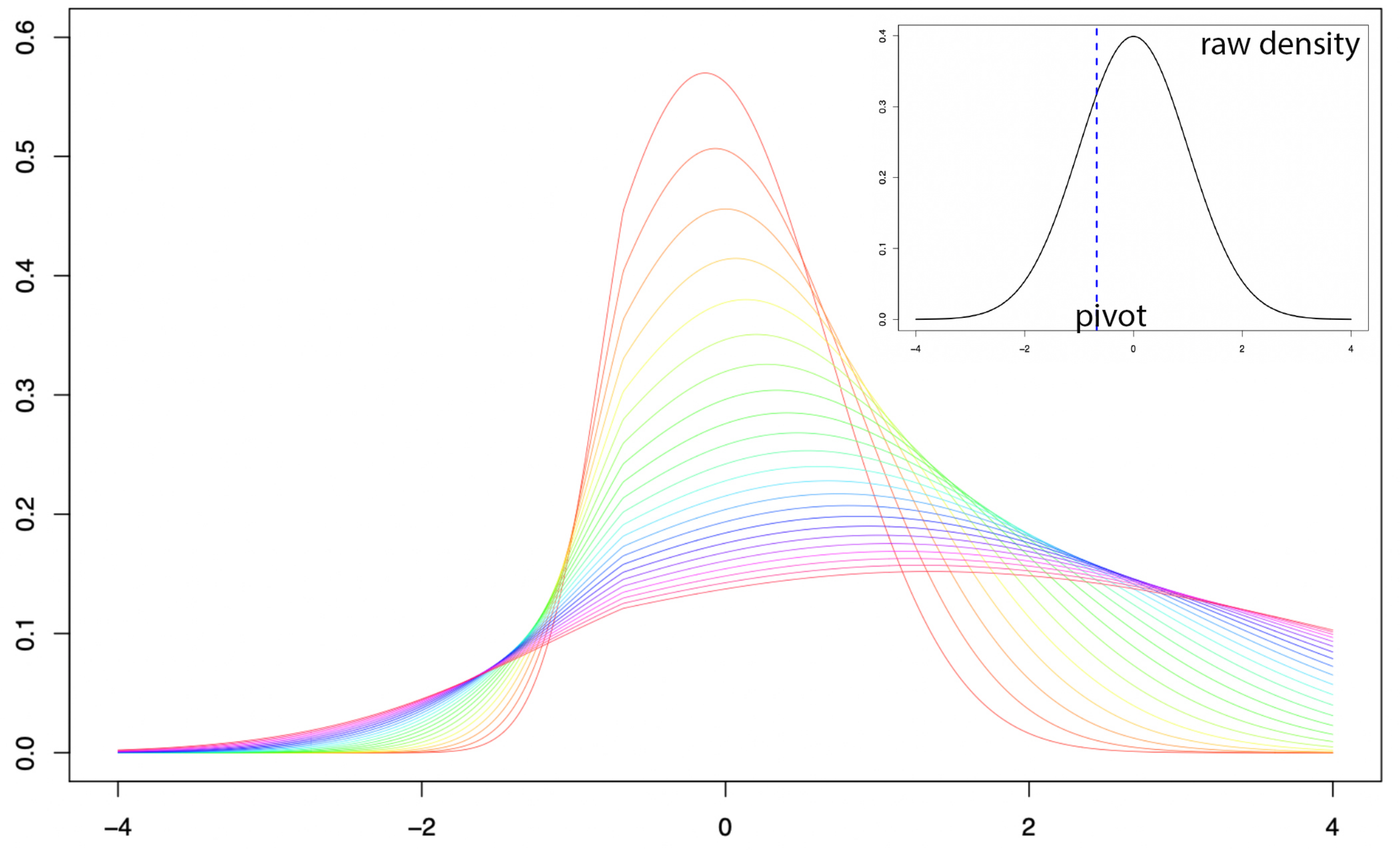}
  \caption{\scriptsize  $\phi\sim N(0,1)$,   $m=\Phi^{-1}(0.25)$}
  \label{normalspeusdemo}
\end{subfigure}\hfill
\begin{subfigure}{0.45\textwidth}
  \includegraphics[width=\textwidth, height=3.5cm]{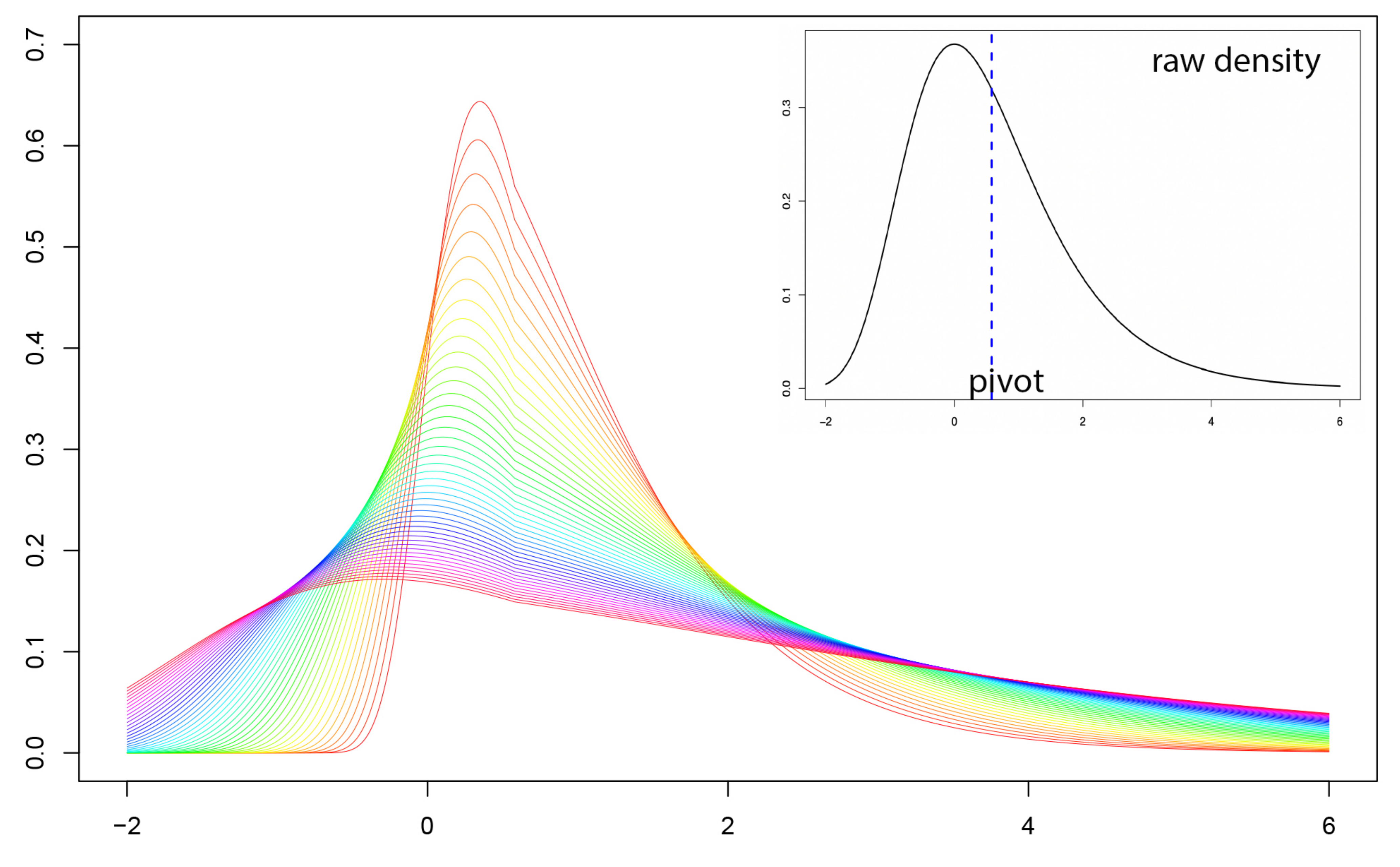}
  \caption{\scriptsize  $\phi$: Gumbel,   $m$: mean}
  \label{gumbelsdemo}
\end{subfigure}\\
\begin{subfigure}{0.45\textwidth}
  \includegraphics[width=\textwidth, height=3.5cm]{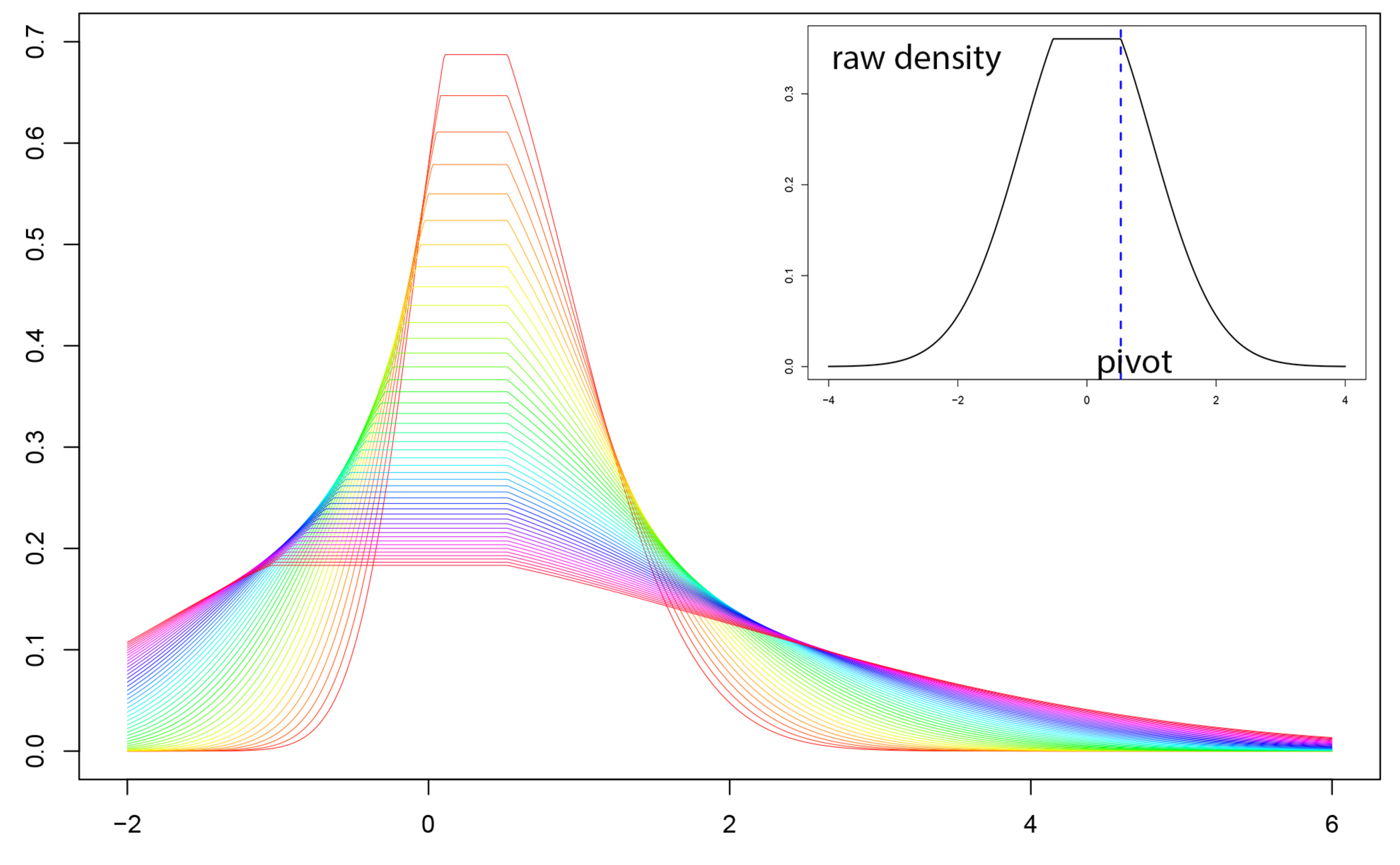}
  \caption{\scriptsize  $\phi$: flat modal,   $m$: (rightmost) mode}
  \label{flatnormalsdemo}
\end{subfigure}\hfill
\begin{subfigure}{0.45\textwidth}
  \includegraphics[width=\textwidth, height=3.5cm]{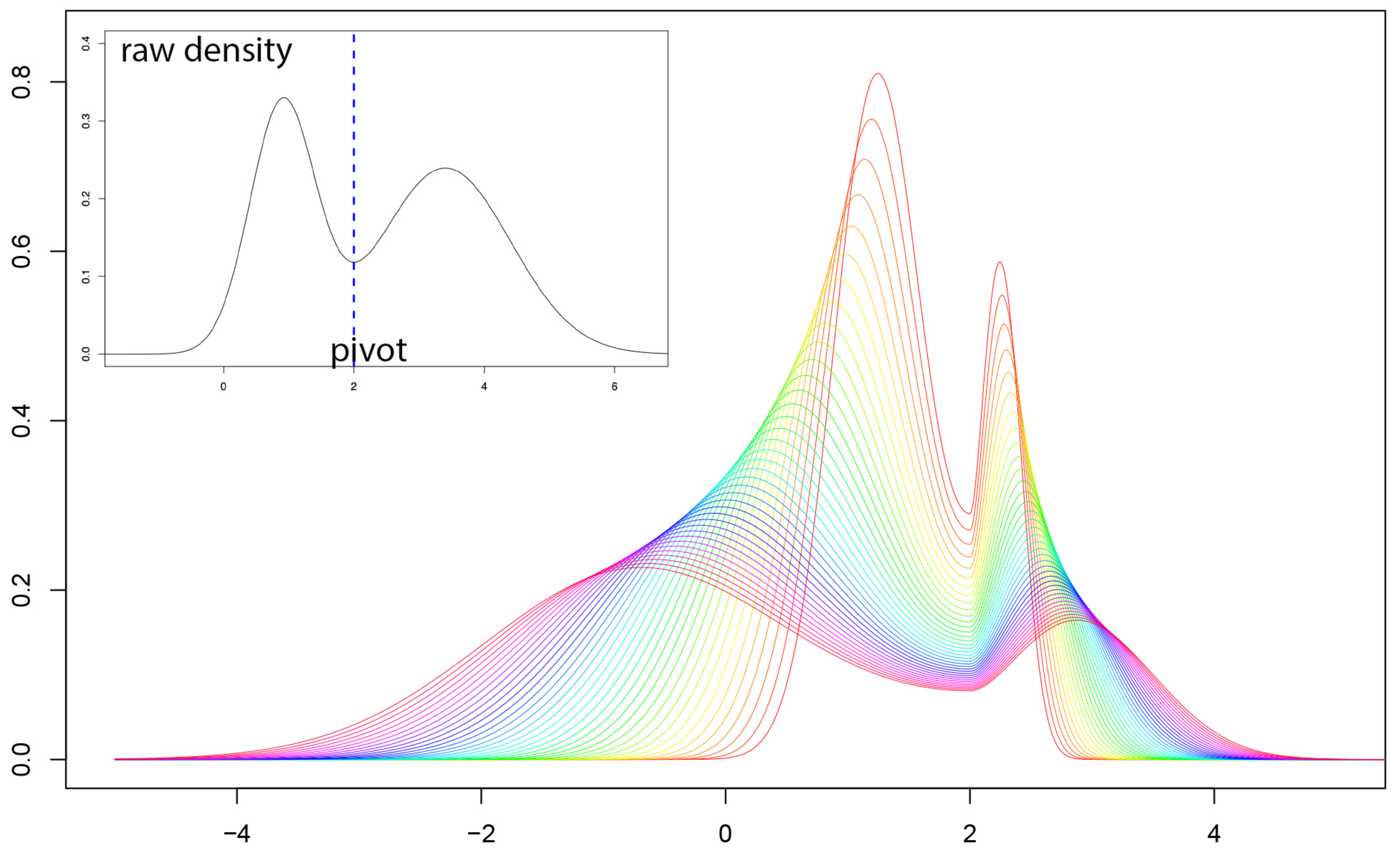}
  \caption{\scriptsize  $\phi\sim 0.4N(0.9,0.25)+0.6N(3.4,1), m=2$}
  \label{bimodaldemo}
\end{subfigure}
\caption{\footnotesize Illustration of some SP  families with \textit{varied} $\sigma$ and $\nu$ (while maintaining a constant ratio). These plots demonstrate the {\textit{versatility}} of skewed pivot-blend in generating a wide range of distributions for practical modeling, including asymmetry and diverse tails (which contrasts with traditional methods assuming symmetry in $\phi$ and median/mode in $m$). \label{fig:spexs} }
\end{figure}

    Figure \ref{fig:spexs} provides visual examples of introducing skewness through pivot-blend   around a nontrivial pivotal point, with variations in $\sigma$ and $\nu$ to demonstrate different tail decay behaviors. 

\begin{example}[Skewed double-gamma family]
 Skewed densities may involve heavy tails and multimodality. When applied to the deneralized double-Gamma density (GDG),   $p/\{2\gamma^d\Gamma(d/p)\} \allowbreak |y|^{d-1}\exp\{-|y|^p/\gamma^p\}$, an extension of  \cite{stacy1962generalization},  skewed pivot-blend reveals
{\small
\begin{align*}
\mbox{\emph{SP\textsuperscript{{\tiny (GDG)}}:}}\, \,
    \,\,&\frac{p}{2[\Phi(m)\sigma+\{1-\Phi(m)\}\nu] \gamma^d\Gamma(d/p)}\Big[\big|\frac{y-m}{\sigma}+m\big|^{d-1}\exp\big\{-\frac{|(y-m)/\sigma+m|^p}{\gamma^p}\big\}1_{y\leq m}
\\
&\quad +\big|\frac{y-m}{\nu}+m\big|^{d-1}\exp\big\{-\frac{|(y-m)/\nu+m|^p}{\gamma^p}\big\}1_{y> m}\Big],
\end{align*}}
\normalsize
where  $\gamma ,d,p>0$  are
 parameters.
 The skewed GDG family comprises bimodal types  such as the skewed double gamma ($p=1, m=0$) and    skewed double Weibull ($d=p, m=0$).
The   unimodal  skewed exponential power distribution family  \citep{zhu2009properties}  is  another instance ($\gamma=1$, $d=1$),  including  the   skewed Laplace distribution and skewed normal distributions \citep{arellano2005statistical,Mudholkar:2000}.
\end{example}

\subsection{SPEUS for Skewed Regression}

Skewed pivot-blend  is a valuable tool for statistical modeling of a skewed outcome $y \in \mathbb{R}^n$  associated with $p$ predictors collected in the matrix $X \in \mathbb{R}^{n\times p}$. Given a  density function $\phi$, if we  assume
\begin{align*}
y-X\beta^{*}\sim \mbox{SP}^{(\phi)}(\sigma^{*},\nu^{*}, m^{*}) 
\end{align*} and define
$\rho=-\log\phi$,   the  estimation of     $\beta^{*},\sigma^{*},\nu^{*},m^{*}$  can be formulated as a joint optimization problem
\begin{align}\label{locationfamily}
\begin{split}
        \min_{\beta,\sigma,\nu,m}\ & n\log\big[\sigma\Phi(m)+\nu\{1-\Phi(m)\}\big]+ \sum^{n}_{i=1}\Big\{\rho\big(\frac{r_i-m}{\sigma}+m\big)1_{r_i-m\leq 0}\\&+\rho\big(\frac{r_i-m}{\nu}+m\big)1_{r_i-m>0}\Big\},  \quad \mbox{s.t.} \quad
r=y-X\beta, \sigma>0,\nu>0,
\end{split}
\end{align}
where the first term arises from the so-called ``normalizing constant'' which is a joint function of $m , \sigma, \nu$.
   Henceforth, we refer to  the framework of \eqref{locationfamily}   as the  Skewed  Pivot-blend Estimation with  Unsymmetric  Scales (\textbf{SPEUS}). We always assume that $\rho$ is constructed from a given  density function  unless  otherwise specified.  \eqref{locationfamily}  is thus an instance of   maximum likelihood estimation (MLE), and  standard MLE asymptotic theory  guarantees  consistency and other  properties.   In practical implementation,   the values of $r_i$ are rarely equal to $m$ and so conventional optimization algorithms like gradient descent, Newton's method, and quasi-Newton methods can  be readily applied. Given the nonconvex nature, initialization impacts estimates, especially for small sample sizes. We usually start location parameters at 0, but using a preliminary  estimate like \cite{yang2019robust} tends to yield better performance. 
A Bayesian approach can be developed as well. It is also worth pointing out that skewed pivot-blend, like other skewness-introducing methods, operates on a given   density with asymmetric scales to handle skewed data. We do not explore nonparametric approaches in this paper (but refer to Appendix \ref{app:furtherext} for potential ideas  involving kernels and data ranks).

\begin{remark}[\bf Pivotal Point vs. Intercept]\label{Pivotvsintercept}
Typically, an intercept  $\alpha$ is included the model, and so $r=y-X\beta=y-X^\circ\beta^{\circ}-1\alpha$, where $X=[1,X^\circ], X^\circ = [\tilde x_1, \ldots, \tilde x_n]^T, \beta=[\alpha,(\beta^\circ)^{T}]^T$. Interestingly, when skewness is present, the pivotal point $m$ diverges from the intercept $\alpha$.

Specifically, based on previous discussions, we have the following density form
$$\sum_{i=1}^n\frac{\phi\big(\frac{y_i- \tilde x_i^T \beta^{\circ}-\alpha-m}{\sigma}+m\big)1_{y_i- \tilde x_i^T \beta^{\circ}\leq m+\alpha}+\phi\big(\frac{y_i- \tilde x_i^T \beta^{\circ}-\alpha-m}{\nu}+m\big)1_{y_i- \tilde x_i^T \beta^{\circ}>m+\alpha}}{\Phi(m)\sigma+\{1-\Phi(m)\}\nu} .$$
It is evident that $m$ plays a more intricate role compared to $\alpha$. If $\sigma=\nu$, the  expression within the sum  can be rewritten in a location-scale form: $(1/\sigma)\phi((y_i- \tilde x_i^T \beta^{\circ} - \alpha')/\sigma)$, where $\alpha' = \alpha + (1 - \sigma)m$. In this special case,  $m$ can be absorbed  into the combined intercept  $\alpha'$, which  is unique (ensuring the final model has no ambiguity).
  This also  applies to $m \le \min r_i$ or $m\ge \max  r_i$,   regardless of scale differences. However, in situations beyond the simple  unskewed case (e.g., when $\sigma$ and $\nu$ are not exactly equal and  $m$ is within the support of $r_i$ or $ 1/n < q< 1-1/n$), $m$ \emph{cannot}    be incorporated into the intercept or casually   discarded.

To the best of our knowledge, the distinct roles of   pivotal point and
intercept in the context of skewness have  received little attention in  existing literature. Our proposal is one of the first attempts to introduce pivotal point estimation into the statistical modeling of skewed data.
\end{remark}

\begin{remark}[\bf Backward Pivot-blend  for Residual Diagnostics]\label{conversionremark}

Let's start by rewriting the forward pivot-blend transform for generating a random variable following $\mbox{SP}^{(\phi)}(\sigma,\nu,m)$:   With   $Y^0_-\sim\phi(y\mid y\leq m),Y^0_+\sim\phi(y\mid y> m)$  and  an independent Bernoulli variable   $U\sim \mbox{Ber}(\sigma\Phi(m)/[\Phi(m)\sigma+\{1-\Phi(m)\}\nu])$, we can construct
\begin{align*} 
\begin{split}
    Y=&\,\begin{cases} m+\sigma(Y^0_- - m) \mbox{ if }  {U=1} \\
                       m+\nu(Y^0_+ - m) \mbox{ if }  {U=0},\end{cases}
\end{split}
\end{align*}
and guarantee that $Y$ follows   $\mbox{SP}^{(\phi)}( \sigma, \nu, m)$.  Conversely, given   $f $ representing $ \mbox{SP}^{(\phi)}(\sigma,\nu,m)$, we can use   $ Y_-\sim f (y\mid y\leq m)$, $ Y_+\sim f (y\mid y> m)$ and an independent Bernoulli random variable      $ V\sim \mbox{Ber}(\Phi(m))$\ to construct a  random variable $Y^0\sim \phi$ using  the  ``\emph{backward}'' pivot-blend:
\begin{align}\label{backtransform}
 Y^0= \Big(\frac{Y_--m}{\sigma}+m\Big) 1_{V=1}+\Big(\frac{Y_+-m}{\nu}+m\Big)1_{ V=0}.  \end{align}
In addition to employing the inverse of the affine transformations in the forward process, the Bernoulli distribution here features a different probability.
 Thus, an $\mbox{SP}^{(\phi)}$ sample can be transformed into a \emph{\textbf{weighted}} sample that follows $\phi$. Importantly,  the functional form of the distribution $\Phi$ is not required to calculate the probability weights; instead,  we can turn to the mixing probability formula \eqref{probdef}
\begin{align}\label{probwts-backPB}
\EP(V=1)=\Phi(m)  = \frac{\nu \EP(Y\le m)}{\nu \EP(Y\le m) + \sigma  \EP(Y > m)}, \EP(V=0)  = \frac{ \sigma  \EP(Y > m) }{\nu \EP(Y\le m) + \sigma  \EP(Y > m)}
\end{align}
and directly estimate these quantities    from the data (also applicable to nonparametric skew estimation in   Appendix \ref{app:furtherext}).

In the context of SPEUS,  \eqref{backtransform} can be used to generate  ``back-transformed" residuals for model diagnostics:   once  the parameters $\beta, \sigma, \nu, m$ are determined, a weighted sample can be created from the residual vector $r = y - X\beta$, which, if the model assumption holds, should adhere to   $\phi$.   First,    define    $\mathcal L=\{i: r_i\leq m\}$,  $L=|\mathcal L| $,   $\mathcal R=\{i: r_i > m\}$, and $R=|\mathcal R|$. As aforementioned,         $\Phi(m)$ can be estimated by $L\nu/\{L\nu+R\sigma\}$. Next, define  $\tilde{r}=[\tilde{r}_i]\in \mathbb{R}^n$:
\begin{align*}
    \tilde{r}_i= \begin{cases} (\frac{r_i-m}{\sigma}+m ) & \mbox{ if }  {r_i\leq m} \\ (\frac{r_i-m}{\nu}+m ) & \mbox{ if }  { r_i> m,}\end{cases} \quad\mbox{ for } 1\le i \le n. 
\end{align*}
   Finally, assign two sets of \emph{nonuniform} probabilities $p_i$ to      $\tilde{r}_i$ based on \eqref{probwts-backPB}: \begin{align*}
p_i =   \nu/(L\nu+R\sigma) \text{ for } i\in \mathcal L, \text{ and }   \sigma/(L \nu+R \sigma) \text{ for } i\in \mathcal R.\end{align*}

Now we can use the R software to plot a \emph{weighted histogram}  of $\tilde r_i$  and compare it to the hypothetical density $\phi$.  This is akin to standard OLS diagnostics for checking the goodness-of-fit of residuals under the Gaussian assumption. In practical data analysis,  after fitting the SPEUS model, one can display the back-transformed residual plot to verify if skewness has been adequately addressed (cf. Figure \ref{fig:backward}).  
\begin{figure}[!h]
\centering
\includegraphics[width=1\textwidth,height=4.5cm]{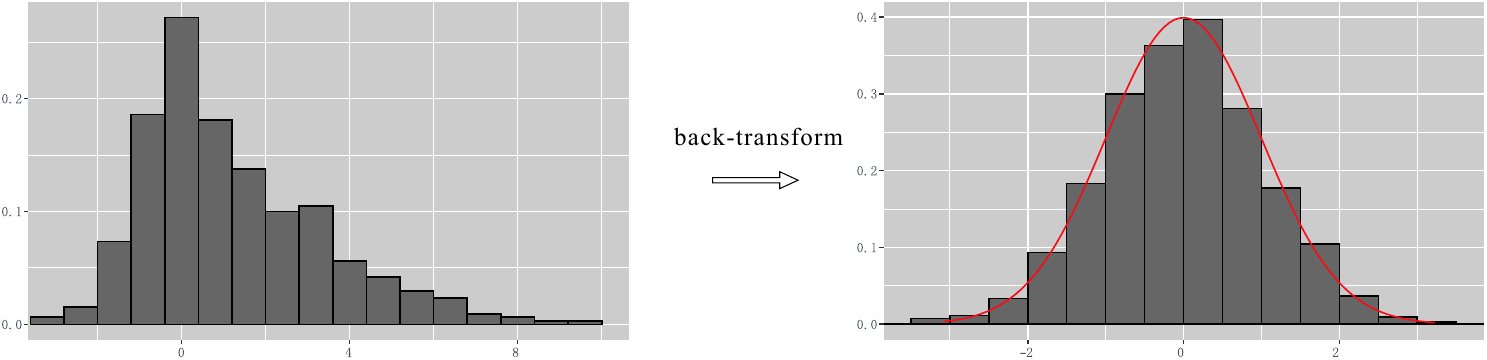}
\caption{\footnotesize An illustration of the {weighted histogram} of back-transformed residuals. The left panel shows residuals from a {skewed} SP$^{(\phi)}$  model with   a symmetric $\phi$, and the right panel displays backward pivot-blend   residuals using estimated parameters. The model assumption is considered valid when the back-transformed residuals closely resemble $\phi$ (the red curve) and,  importantly, exhibit \textit{symmetry}, indicating effective handling of skewness. \label{fig:backward} }
\end{figure}

\end{remark}

\subsection{Expansions of Skewed Pivot-Blend and Relevant Works}\label{subsec:exts}
\subsubsection{Extensions and Beyond}
Our main focus is on applications where the loss function $\rho$ is   derived from a single density function $\phi$. However, there are also variations of skewed pivot-blend   that hold value across different applications and fields. 


\paragraph{Skewed Pivot-blend for two  densities.} Skewed pivot-blend extends  capabilities to seamlessly   ``paste'' two distinct densities with varying scales, while ensuring continuity at the pivotal point.  Consider $\phi$ and $\psi$ as two continuous densities with respective distributions  $\Phi$ and $\Psi$, and  $m$ an interior point within the support of both densities. The   process of conditioning, affine transformations, and mixing,  using two scales, $\sigma$ and $\nu$ in relation to $m$, leads to
\begin{align}
p \frac{\phi(\frac{y-m}{\sigma} + m)}{\Phi(m) \sigma}1_{y\le m} + (1-p) \frac{\psi(\frac{y-m}{\nu} + m)}{(1 - \Psi(m)) \nu}1_{y > m}.  \label{pb-2den-starting}
\end{align}
Choosing
\begin{align*} p = \frac{\psi(m) \Phi(m)\sigma}{\phi(m)(1 - \Psi(m))\nu+\psi(m) \Phi(m)\sigma}\end{align*} results in the following continuous density:
\begin{align}
\frac{\psi(m) \phi(\frac{y-m}{\sigma} + m)1_{y\le m}+\phi(m) \psi(\frac{y-m}{\nu} + m)1_{y > m}}{\phi(m)(1 - \Psi(m))\nu+\psi(m) \Phi(m)\sigma }.\label{twodenpasted}
\end{align}
This offers a means of fusing two distinct tail types with exceptional  flexibility.
 Remarkably, even when     $\sigma=\nu$, the pasted density in \eqref{twodenpasted}  does not conform to a location-scale form (in contrast to  the single-density scenario, cf. Remark \ref{Pivotvsintercept}), and  $m$ cannot be simply interpreted  as  a  location shift.
Beyond the estimation of scales, an  intriguing  question   is to determine  the pivotal point at which the two densities coalesce. Furthermore, the concept of skewed pivot-blend can be iteratively applied to paste multiple densities with varying scales. In multidimensional spaces, the pivotal point can be extended to a \textit{pivotal hyperplane} for combining two densities, which is another intriguing  topic for future exploration.

\paragraph{Skewed pivot-blend for bounded losses.}
In our discussions, we generally assume that $\rho$ is a negative log-likelihood---for example,   a convex $\rho$ function like Huber's loss corresponds to a log-concave density. However,  it is well established in robust statistics that bounded nonconvex losses are  more effective     in handling  extreme outliers with high leverage. Two prominent examples are Tukey's bisquare loss and Hampel's three-part loss, both of which  are bounded (or winsorized, preventing them from reaching $+\infty$) and are   constructed using piecewise polynomials  \citep{hampel2011robust}. We can   formulate a general objective for estimating the location parameters   
\begin{align}\label{skrobustloss}
  \sum^{n}_{i=1}\big\{\rho (\frac{r_i-m}{\sigma}+m )1_{r_i\le m }+\rho (\frac{r_i-m}{\nu}+m )1_{r_i> m}\big\} +  n\cc\log \big(\frac{\sigma\Phi(m)/\{1-\Phi(m)\}+\nu}{1+\Phi(m)/\{1-\Phi(m)\}}\big).
\end{align}
Here,    $r=y-X\beta$, $0\le \Phi(m) \le 1$, 
and $\chi_0$ is    for the purpose of calibration.
The user can specify the particular forms of $\rho$ and $\Phi$ (and in the convex-$\rho$ case,  $\Phi(m)$ may take $ \lim_{M\to\infty} {\int^{m}_{-M}\exp(-\rho(t))\rd t}/{\int^{M}_{-M}\exp(-\rho(t))\rd t}$). 
 In robust statistics, it is often recommended to first perform a separate ad-hoc robust scale estimation \citep{maronna2006robust}, before proceeding to optimize \eqref{skrobustloss} for the location parameters $\beta, m$.  But various selections for $\sigma$ and $\nu$ influence the structure of the resulting  asymmetric loss. This practice prompts a theoretical inquiry:  is it possible to set a finite-sample error bound for location estimation using data-dependent scales, regardless of scale construction or the data distribution of   $y$?
For a nonasymptotic analysis using statistical learning theory, see Theorem \ref{riskbound} for     insights on how skewness adds to problem complexity and increases  ``excess risk".

Finally, extensions of the skewed pivot-blend in nonparametric estimation, such as methods based on data ranks and kernels, can be found  in Appendix \ref{app:furtherext} due to space constraints.
\subsubsection{Other Related Works}

Section \ref{sec:intro} and Section \ref{subsec:spb} provide a list of relevant works on two-piece distributions. Moreover, as pointed out by  a reviewer, our skewed pivot-blend idea shares similarities with the  ``composite models'' in the fields of finance and actuarial sciences.

 In such a context, researchers often aim to create a new size distribution by combining two distributions: one that is lighter-tailed on the left (e.g., lognormal or Weibull), and another that is heavier-tailed on the right (e.g., Pareto) \citep{cooray2005modeling}. This can be expressed     as  $p \cdot \frac{\phi(y)}{\Phi(m)} \cdot 1_{y\leq m} + (1-p) \cdot \frac{\psi(y)}{1 - \Psi(m)} \cdot 1_{y > m}$  \citep{Scollnik2007}, with the choice of $p$   to ensure the resulting density is smooth. An alternative proposal  appeared in   \cite{bernardi2018two},  which, however, does not guarantee continuity.   For further  discussions on composite models, we refer to \cite{klugman2012loss} and \cite{dominicy2017distributions}.

It is not difficult to see that  the composite form described above corresponds to a specific instance of our skewed pivot-blend density  \eqref{pb-2den-starting} with $\sigma=1, \nu=1$. However, our research is   driven by the need to tackle data skewness,  where the values of $\sigma$ and $\nu$ are typically unknown and can vary. Our primary goal is to estimate these potentially distinct scales while also identifying the central pivotal point in the context of skew data analysis.  \eqref{pb-2den-starting} or \eqref{twodenpasted}  is notably different from the composite distribution even when $\sigma=\nu$ but not equal to 1.

  Additionally,   the composite model may be rigid and    restrictive due to limited  choices for the mixing parameter $p$
\citep{Scollnik2007}. In contrast, Figure \ref{fig:spexs}   illustrates the flexibility of SP distributions, taking on various shapes through adjustments in $\sigma$ and $\nu$. The versatility of the skewed pivot-blend, capturing both asymmetry and varied tail behaviors, offers a variety of  distributions for practical modelingg.

\section{Skewed Two-Part Model with Joint Sparsity}
\label{skewedtwo-partmodelsection}

As  mentioned in Section \ref{sec:intro}, our work is driven by the study of ``\textit{semicontinuous} \textit{outcomes}'' \citep{olsen2001two}, as defined by a significant proportion of values equaling 0, with the remaining values following a continuous, often skewed, distribution. For example, the  MEPS datasets  have  many  patients showing no medical expenditure (including  the sum of out-of-pocket payment, insurance, Medicaid, Medicare, and other payments), and the rest with positive, highly skewed and heavy-tailed medical costs (see Figure~\ref{mepsraw}). Semicontinuous outcomes   are  frequently encountered in    biomedical and economic  applications, as well as rainfall levels and  daily drinking records \citep{hyndman2000applications,liu2008multi, sarul2015application}.

   Because zero medical cost means no medical service, rather than an outcome resulting from  truncation or sampling, commonly used biometric models like the Tobit model and zero-inflated models \citep{tobin1958estimation, lambert1992zero}    are not suitable.  Instead, the \textit{two-part} model \citep{cragg1971some, mullahy1998much}, sometimes also referred to as a {\textit{hurdle}} model, is more appropriate. This model can be expressed as:       \begin{align}\label{generaltwo-part}
 \mbox{\small Two parts for semicontinuous $y$:}\quad     \begin{cases}
  \mathbb{P}(y_i=0)=\pi_i=1/\big\{1+\exp(-\tilde{x}^T_{i}b)\big\}\\
 y_i\mid y_i>0\sim f(y_i;\tilde{x}^T_{i}\beta).
\end{cases}
  \end{align}
In the \uline{binary} part, the probability of observing a zero response is typically modeled using logistic regression or a probit model; in the \uline{continuous} part, the density function $f$ represents a {\textit{positive}} random variable with  parameter $\tilde{x}_i^T \beta$.     Without loss of generality, let's assume that for $1\leq i\leq n$,   $y_i>0$, while for $n<i\leq N$,  $y_i=0$, and so the response and   the overall prediction matrix $\widetilde{X}$ can be partitioned as
\begin{align}
     y=
  \big[[y_1, \ldots, y_n ] \,\,  [ 0, \ldots, 0  ]\big]^T, \quad
     \widetilde{X}=\big[
        [\tilde{x}_{1}, \dots, \tilde{x}_{n}] \,\,  [\tilde{x}_{n+1}, \dots, \tilde{x}_{N} ]\big]^T= [X^T\,\,  Z^T]^T. \label{extyX}
\end{align}
We  then derive the negative log-likelihood from the distribution defined in \eqref{generaltwo-part}, with respect to a combination of the Lebesgue measure on $\mathbb R_+$ and a counting measure at 0:
\begin{align}\label{new two-part model loss}
  &-\sum^N_{i=1}\big[1_{y_i=0}\log \pi_i + 1_{y_i>0} \log \{(1-\pi_i) f(y_i;\tilde{x}^T_{i}\beta)\}\big] \notag\\
 = &   \sum^N_{i=1}\big[-\tilde{x}^T_{i} b 1_{y_i=0}+\log\big\{1+\exp(\tilde{x}^T_{i} b)\big\}\big]+\sum^{n}_{i=1}-\log f(y_i;\tilde{x}^T_{i}\beta).
\end{align}
Each predictor makes a composite contribution to the response through two parts, but   \eqref{new two-part model loss} is separable with respect to   $b$ and $\beta$, making it amenable to optimization. Below, we will introduce two  modifications to the classical two-part model to better address some  challenges in modern applications: (a) mitigating {skewness} in the positive part through the use of pivot-blend, and (b) enhancing {interpretability} by incorporating joint variable selection across both model components.

First, specifying an ideal density function for the positive values of $y_i$
 can be challenging.   As a result, many researchers opt to employ a transformation   $T(\cdot): (0, \infty) \rightarrow \mathbb R$, and assume that the transformed response $T(y_i)$ ($1\le i \le n$) follows a symmetric distribution, such as   a normal or a Laplace. With    $\rho$ denoting  the corresponding symmetric negative log-likelihood, the last term $ \sum^{n}_{i=1}-\log f(y_i;\tilde{x}^T_{i}\beta)  $ in \eqref{new two-part model loss} now takes the form
\begin{align}\label{2part-plainpositive}
\sum^{n}_{i=1} \rho ( {r_i }  )  \quad \text{ with } r=T(y)-X\beta
\end{align}
where $r\in\mathbb R^n$ is the residual vector associated with the   continuous component of the model. Nevertheless,  asymmetry continues to manifest  in the transformed data in various scenarios, as observed by  \cite{chai2008use}. Our experience shows that routine transformations may not only fail to completely rectify skewness but also introduce a nontrivial pivotal point around which asymmetric scales arise. The technique detailed in Section \ref{sec:pb-speus} offers an effective remedy by replacing  \eqref{2part-plainpositive} with the following
\begin{align}\label{two-partspeuspart}
\begin{split}
  &n\log [\sigma\Phi(m)+\nu\{1-\Phi(m)\} ]+\sum^{n}_{i=1}\big[\rho\big(\frac{r_i-m}{\sigma}+m\big)1_{r_i\le m}+\rho\big(\frac{r_i-m}{\nu}+m\big)1_{r_i>m}\big],
\end{split}
\end{align}
where $\sigma, \nu, m$ are all unknown.

Second, practitioners of two-part models encounter another pressing challenge---the abundance of predictors  collected. Variable selection provides a valuable tool for enhancing model  interpretation and prediction, but   in the context of two-part models, it is crucial to identify predictors that are relevant to the entire population, rather than just focusing on  the subpopulation with positive responses or the subpopulation with zero responses. In other words, a predictor can only be eliminated if it bears zero coefficients in \textit{both} the binary and continuous parts of the model.

Combining both elements, the incorporation of   regularization and  skewed pivot-blend   allows us to  formulate a sparse skewed 2-part $(\textbf{S}^2)$ criterion for modeling semicontinuous outcomes with joint variable selection:\begin{align} \label{penalizedsparsesemiloss}
\begin{split}
&\text{\textbf{S}$^2$}:  \min_{b,\beta,\sigma,\nu,m} n\log\big[\sigma\Phi(m)+\nu\{1-\Phi(m)\}\big]+\sum^{n}_{i=1}\big\{\rho\big(\frac{r_i-m}{\sigma}+m\big)1_{r_i-m\leq0}\\
 &\quad + \rho\big(\frac{r_i-m}{\nu}+m\big)1_{r_i-m>0}\big\}+\sum^N_{i=1}\big[-\tilde{x}^T_{i} b 1_{y_i=0}+\log\big\{1+\exp(\tilde{x}^T_{i} b)\big\}\big]\\& \quad  + \lambda\|B\|_{2,1}+ P_2( \sigma, \nu, m, \beta; \tau)\  \ \mbox{s.t.}  \,\,\, r=T(y)-X\beta, B=[ \sqrt n \beta, \sqrt N b], \sigma>0,\nu>0.
\end{split}
\end{align}
Practically, it is common to include two intercepts, one for the binary part and one for the continuous part of the model, which are not subject to any penalty.
 The (2,1)-norm applied to matrix $B$ enforces the desired row-wise sparsity for joint variable selection, but can be substituted with a row-wise nonconvex penalty like group SCAD or MCP.
Incorporating the scaling factors in the construction of
$B$ is essential for the use of a single regularization parameter.
The term $P_2$ represents an $\ell_2$-penalty, akin to ridge regression, to account for significant noise and design collinearity. An example is adding      $(\tau/2) (1/\sigma^2 + 1/\nu^2)$    (especially  when $p\ge n$), which from a Bayesian perspective  amounts to an inverse gamma prior on $\sigma^2$ and $\nu^2$. Empirical studies show that $\tau$ is not  sensitive, and a small  $\etaa$      often suffices. Likewise, we suggest including an $\ell_2$-penalty on $m$ which translates to   a Gaussian prior (or alternatively,   a beta prior  for the quantile parameter $q$), because in cases involving asymmetric scales $\sigma\ne \nu$, a moderate $|m|$ value   (or $q$ near 1/2) can   exert a considerable influence on the model,  warranting deeper exploration in  applications. (This might contrast with outlier effects which  exhibit more of a tail behavior  inconsistent with most data  and are complex to model with a single distribution due to heterogeneity.) Adding the $\ell_2$ penalties also facilitates the  theoretical analysis in the next section.

\eqref{penalizedsparsesemiloss}  involves the estimation of coefficients $b$, $\beta$, a pivotal point $m$, and two scales $\sigma, \nu$ and is one-sided directionally differentiable \citep{shebregman2021}. In contrast to the conventional two-part  \eqref{new two-part model loss}, this criterion no longer shows separability in $(b, \beta)$ and includes a non-differentiable penalty. Efficient computation of the estimates can be achieved through optimization techniques. In handling the nondifferentiable (2,1)-penalty, we can express $\|B\|_{2,1}$ as $(1/2)\sum_{j=1}^p (B_{j,1}^2+ B_{j,2}^2)/a_j + a_j$ with each $a_j>0$. This  yields a differentiable criterion, facilitating the use of standard optimization solvers such as Newton or quasi-Newton methods. 
 Alternating optimization can also be used to improve scalability. 

\section{Analysis of Sparse Skewed Two Parts}
\label{sec:THspeus2parts}
In this section, we delve into the theoretical underpinnings of the  sparse skewed two-part estimation \eqref{penalizedsparsesemiloss} introduced earlier. Our  investigation differs from classical asymptotics which assume a fixed number of predictors and an infinite sample size. Nonasymptotic theory remains relatively unexplored when considering the interplay of skewness, regularization, and heavy tails collectively. A significant challenge entails comprehending the impact of asymmetric scales, pivotal points, and sparsity on both statistical accuracy and the choice of the regularization parameter in finite sample sizes.

  The key implications and contributions of our theoretical framework are as follows:
(i)  {Applicability}. Unlike conventional consistency studies, which frequently assume an i.i.d. data structure and require fixed $p$ with $n\rightarrow +\infty$,  our theory is applicable to {any} values of $n$ and $p$, and does not require the design matrix to have i.i.d. rows.
(ii)
{Effective Noise and Flexibility in Tails}. Our investigation reveals that prediction and estimation errors, as well as the choice of the regularization parameter, are  linked to the tail decay characteristics of the   ``effective noise".   The concept diverges from raw noise and often exhibits \textit{lighter} tails  (cf. Remark \ref{effectivenoisevsrawnoise}).
  Moreover, we employ   Orlicz $\psi$-norms  to model  various tail behaviors (cf. Lemmas \ref{componenttovectorlemma}--\ref{psinormconvertlemma}),  providing flexibility in real applications.  In essence, when the effective noise shows light  tails,  implying  a finite Orlicz norm with a ``large" $\psi$ function, the presence of $\psi^{-1}$ in the choice of the regularization parameter leads to a reduced error bound (cf. \eqref{thm1result} or \eqref{l1estimationerrorrate}).
(iii)
 Misspecification Tolerance. 
 The core theorems do not require a zero-mean effective noise, as demonstrated in Theorems \ref{globalgroupl1}, \ref{globalgrouplreconstanttheta}, \ref{globalgrouplregeneral}, \ref{inftyboundprop}, \ref{globalthmsubgaussian}. Consequently, our   analysis applies to {misspecified} models (where the risk function associated with the given loss  does not necessarily vanish  at the statistical truth). 

\subsection{Preliminaries: Reparametrization and Effective Noise }
Recall the loss  function    in \eqref{penalizedsparsesemiloss}, which  can be represented by
\begin{align}\label{lossdefthm}
\begin{split}
 l_0(\beta,b,m,\sigma,\nu)=&\sum^{n}_{i=1}\bigg\{ \rho\big(\frac{r_i-m}{\sigma}+m\big)1_{r_i\leq m}+\rho\big(\frac{r_i-m}{\nu}+m\big)1_{r_i>m} \\ &+\log\big[\sigma\Phi(m)+\nu\{1-\Phi(m)\}\big]\bigg\}+\langle \lbb(Z b;\deltay),1_N\rangle,
\end{split}
\end{align}
where $r=y-X\beta$ denotes the plain residuals,   $\rho$ and $\lbb$ are two  differentiable losses that  respectively operate on   the continuous part and  binary part. Here, the observed data are represented by    $y, \deltay, X, Z$, with    $y\in \mathbb{R}^n$,
$X\in \mathbb{R}^{n\times p}$,   $\deltay\in\{0,1\}^N$ and   $Z\in\mathbb{R}^{N\times p}$ (cf. \eqref{extyX}).
  For simplicity,   the section
assumes that the number of rows in  $X$ and $Z$ is \textit{each} bounded  by $cn$, with $n$ representing the order of the sample size and $c$  a positive constant.

To ease theory and presentation, we introduce some   concatenated symbols. First,
$\beta$ and $b$ can be combined into $\bar\beta$ as the coefficient vector for an extended design matrix  $\bar{X}$:
\begin{align*} 
  \renewcommand{\arraystretch}{0.7}
  \bar\beta=\mbox{vec}([\beta, b]) = \begin{bmatrix}
\beta \\
b  \\
\end{bmatrix},\,\,    \bar{X}=\begin{bmatrix}
X & 0\\
0 & Z \\
\end{bmatrix}.
\end{align*}
We denote $[\beta_k,b_k]^T$ by $\bar\beta_k$, the coefficients associated with the $k$th and $(p+k)$th columns of $\bar X$.  Based on the problem structure in \eqref{lossdefthm},   we define
\begin{align*}
\bar{m}=m1_n,\quad \renewcommand{\arraystretch}{0.7}
  \varsigma=\begin{bmatrix}
(1/\sigma)1_n\\
(1/\nu)1_n  \\
\end{bmatrix},
\end{align*}
where  $1_{n}$ is   to match the scale of the design matrices   when considering  prediction errors. Introduce $\unv$ as
the overall unknown vector,  as well as $\gamma$ and $\mu$
\begin{align}\label{systemnota}
    \unv=[\bar\beta^T,\bar m^T,\varsigma^T]^T,\quad \gamma=[\bar m^T,\varsigma^T]^T,\quad \mu=[\bar\eta^T,\bar m^T,\varsigma^T]^T,\quad \bar\eta=\bar{X}\bar{\beta}.
\end{align}

With the above notations, we can rewrite the general problem of interest  as
\begin{align}\label{quadraticpenaltyloss}
   l(\mu)+\|\mtbb\bar{\beta}\|_{2,P}+\frac{\etaa}{2}\|\gamma\|^2_2,
\end{align}
where $l$ is the  loss on $\mu$,
 $\|\bar{\beta}\|_{2,P}:=\sum^p_{k=1}P(\|\bar{\beta}_k\|_2;\lambda)$ and $P$  is a  sparsity-promoting penalty. Including $\varrho$ in the penalty allows for a scale adjustment based on the size of the designs, enabling a universal choice of $\lambda$ that is independent of the sample size in later theorems.    In alternating optimization algorithms, $\varrho$ can takes  $\rhotwoinf $ which   represents the maximum column $\ell_2$-norm of $\bar X$,   as a measure of the size  of the design:
 \begin{align*}
  \rhotwoinf=\max_{1\le k \le 2p} \| \bar X_k\|_2 = \max  \big\{\|X_1\|_2,
\ldots, \|X_p\|_2,  \|Z_1\|_2,\ldots,  \|Z_p\|_2
\big\}.
 \end{align*}
This quantity is typically on the order of   $\sqrt{n}$.
When
$P$ is the $\ell_1$-penalty,   \eqref{quadraticpenaltyloss} reduces to the previous  $\textbf{S}^2 $-criterion \eqref{penalizedsparsesemiloss} for two-part models  with skew and sparsity,
\begin{align}\label{l1thmloss}
      l(\mu)+\lambda\|\mtbb\bar{\beta}\|_{2,1}+\frac{\etaa}{2}\|\gamma\|^2_2,
\end{align}
where $\|\bar{\beta}\|_{2,1}$
is short for $\sum^p_{k=1}\|\bar\beta_k\|_2$.\\

Next, we introduce the notion of ``effective noise"
to account for randomness, conditional on the design matrices $X,Z$. Given $l(\mu)$, define the effective noises  associated with $\bar{\eta}^*,\gamma^*$ as
\begin{align}\label{subvectoreffectivenoise}
    &\epsilon_{\bar\eta}=-\nabla_{\bar\eta} l(\mu)\big|_{\mu=\mu^*},\quad  \epsilon_{\bar m}=-\nabla_{\bar m} l(\mu)\big|_{\mu=\mu^*},\quad  \epsilon_\varsigma=-\nabla_{\varsigma} l(\mu)\big|_{\mu=\mu^*},
\end{align}
where  $l$ is assumed to be differentiable  at  the statistical truth $\mu^*$.

When formulating statistical assumptions related to effective noises, it is important to account for different types of tail decay. We employ  Orlicz $\psi$-norms, as well as  some  nonconvex variants capable of  handling significantly  heavier tails (cf. Appendix \ref{subsec:orlicz}). In the context of the Orlicz $\psi$-norm for a random variable (or vector) $X$, represented as $\|X\|_\psi$, $\psi(\cdot)$ is consistently assumed to be a nondecreasing, nonzero function defined on $\mathbb{R}_+$ with $\psi(0)=0$ (but not necessarily convex). For the Orlicz-norm of a random vector, please see \eqref{orliczvecdef}. The inverse of $\psi$ is defined as     $\psi^{-1}(x)=\sup\{t\in\mathbb{R}_+:\psi(t)\leq x\}$.

   Some notable  examples encompass   the  sub-Weibull $\psi_q$-norms, with $\psi$ defined as
\begin{align}\psi_q(x)=\exp(x^q)-1, \ x\in \mathbb R_+ \label{subW-psi}\end{align} for  $q> 0$. \eqref{subW-psi}   covers   sub-Gaussian   ($q=2$) and sub-Exponential  ($q=1$) random variables, but  as $q<1$, the sub-Weibull tails become much heavier \citep{gotze2021concentration}.  Another class is the  $L_q$-norms, with $\psi(x)=x^q$ ($q\geq1$). Orlicz norms provide a useful framework for analyzing skewed random variables (even when they lack a zero mean).

\begin{remark}[\bf Effective Noise vs. Raw Noise]\label{effectivenoisevsrawnoise}
 The effective noise, jointly determined by the data and  the loss function, may differ from the plain ``raw noise'' defined by
  \begin{align}\label{rawnoisedef}
     \epsilon^{\mbox{\scriptsize{raw}}}:=y-\bar m^*-\eta^*,
 \end{align}
 where $\eta^*=X\beta^*$.
 Comparing \eqref{subvectoreffectivenoise} with \eqref{rawnoisedef}, one appealing aspect of $\epsilon_{\bar \eta}$ is that it tends to have \emph{light} tails, even when $ \epsilon^{\mbox{\scriptsize{raw}}}$ does not.
Indeed,  a straightforward derivative calculation based on \eqref{lossdefthm} shows that
\small
\begin{subequations}
\begin{align}
       &\,\epsilon_{\bar\eta,i}=
   \begin{cases}
    \frac{\rho'(\epsilon^-_i)}{\sigma^*}1_{-}(\epsilon^{\mbox{\scriptsize{raw}}}_i)+\frac{\rho'(\epsilon^+_i)}{\nu^*}1_{+}(\epsilon^{\mbox{\scriptsize{raw}}}_i),\quad  1\leq i\leq n,\\\lbb'(Z^T_ib^*;\deltay_i),\quad  n< i\leq N,
   \end{cases}\\
         &\,\epsilon_{\bar{m},i}=(\frac{1}{\sigma^*}-1)\rho'(\epsilon^-_i)  1_{-}( \epsilon^{\mbox{\scriptsize{raw}}}_i)+(\frac{1}{\nu^*}-1)\rho'(\epsilon^+_i)  1_{+}( \epsilon^{\mbox{\scriptsize{raw}}}_i)+\frac{ \Phi'(m^*)(\nu^*-\sigma^*)}{\sigma^*\Phi(m^*)+\nu^*(1-\Phi(m^*))},  \\
     &\, \epsilon_{\varsigma,i}=\begin{cases}
       - \epsilon^{\mbox{\scriptsize{raw}}}_i\rho'(\epsilon^-_i)  1_{-}( \epsilon^{\mbox{\scriptsize{raw}}}_i)+\frac{ \sigma^{*2}\Phi(m^*)}{\sigma^*\Phi(m^*)+\nu^*(1-\Phi(m^*))},\quad 1\leq i\leq n,\\  -\epsilon^{\mbox{\scriptsize{raw}}}_i\rho'(\epsilon^+_i) 1_{+}( \epsilon^{\mbox{\scriptsize{raw}}}_i)+\frac{ \nu^{*2}\{1-\Phi(m^*)\}}{\sigma^*\Phi(m^*)+\nu^*(1-\Phi(m^*))},\quad n< i\leq 2n,
     \end{cases}
\end{align}
\end{subequations}
\normalsize
where $\epsilon^-_i=\epsilon^{\mbox{\scriptsize{raw}}}_i/\sigma^*+m^*, \epsilon^+_i=\epsilon^{\mbox{\scriptsize{raw}}}_i/\nu^*+m^*$.
Therefore, if $|\rho'|\leq M,|\lbbd|\leq B$ for some positive $M,B$ (e.g., when using Huber's loss for $\rho$ and logistic deviance for $\mathfrak L$),
then  {\small \begin{subequations}\label{epscal-bounded}
 \begin{align}
  &\,  \big|\epsilon_{\bar\eta,i}\big|\leq \frac{M}{\sigma^*\land \nu^*} + B,\label{etabound}\\
    &\, \big|\epsilon_{\bar{m},i}\big|\leq \big(|1-\frac{1}{\sigma^*}|\lor |1-\frac{1}{\nu^*}|\big) M+\frac{|\sigma^*-\nu^*|}{\sigma^*\land\nu^*}, \label{mbound} \\
    &\,\big|\epsilon_{\varsigma,i}\big|\lor \big|\epsilon_{\varsigma,i+n}\big|\leq M|\epsilon^{\mbox{\scriptsize{raw}}}_i|+(\sigma^*\lor \nu^*). \label{varsigmabound}
\end{align}
\end{subequations}}\normalsize
 It is evident that {all} components of $\epsilon_{\bar\eta }$ and $\epsilon_{\bar m}$ are   {\emph{bounded}},  thereby possessing a finite $\psi_2$-norm regardless of  heavy tails that $\epsilon^{\mbox{\scriptsize{raw}}}$ may exhibit. Finally, it is worth noting that our theorems below impose Orlicz-norm conditions on the entire random vectors in \eqref{subvectoreffectivenoise}, which is   more flexible  than assuming that the vectors have independent  components, each with a finite Orlicz norm and a mean of 0 (cf. Lemma \ref{componenttovectorlemma}).  Furthermore, one can employ  generalized Bernstein-Orlicz norms for random vector marginals, as described in \cite{kuchibhotla2018moving} to  develop  sharper bounds under an additional minimum sample size constraint. We will not explore this further in the current paper.

\end{remark}

\subsection{Nonasymptotic Error Bounds}

This part demonstrates some error bounds when using the $(2,1)$-penalty. Additional results can be found in the appendices, such as Theorem \ref{globalgrouplregeneral} providing a universal form for $\lambda$      applicable to  a  broad range of  tails, Theorem \ref{inftyboundprop} presenting an elementwise error bound, and Theorem \ref{globalthmsubgaussian} examining a general sparsity-inducing penalty. 

In what follows, we denote  the group support of $\bar{\beta}$ as  $\mathcal{J}(\bar{\beta})=\{k:\bar\beta_{k}=[\beta_k,b_k]^T\neq 0,1\leq k\leq p\}$ and $J(\bar\beta)$  is the cardinality of  $\mathcal{J}(\bar{\beta})$. Also,   define   $   \mathcal{J}^*=\mathcal{J}(\bar{\beta}^*), J^*=J(\bar\beta^*)$, and $\hat{\mathcal{J}}=\mathcal{J}(\hat{\bar\beta})$ for short, and let $\mathcal J^{*C}\subset\{1, \ldots, p\}$ denote the complement of $\mathcal J^*$.
The   {generalized  Bregman function} $    \mathbf{\Delta}_{l}$  is useful in defining an appropriate error measure and making regularity conditions: given a function $l$ differentiable at $\eta'$,
$
    \mathbf{\Delta}_{l}(\eta,\eta'):=l(\eta)-l(\eta')-\left\langle \nabla l(\eta'),\eta-\eta'\right\rangle
$  and   $\bar{\mathbf{\Delta}}_{l}(\eta,\eta'):=\{\mathbf{\Delta}_{l}(\eta,\eta')+\mathbf{\Delta}_{l}(\eta',\eta)\}/2$. The  differentiability can be replaced by directional differentiability   \citep{shebregman2021}.
If $l$ is also strictly convex, $\mathbf{\Delta}_{l}(\eta,\eta')$ becomes the standard
Bregman divergence $\mathrm{\mathbf{D}}_l(\eta,\eta')$ \citep{bregman1967relaxation}. For the specific case of  $l(\eta)=\|\eta-y\|_{2}^{2}/2$, $\mathbf{\Delta}_{l}(\eta,\eta')=\|\eta-\eta'\|_{2}^{2}/2$, or $  \mathrm{\mathbf D}_{2}(\eta,\eta')$ for short.

\begin{theorem}\label{globalgroupl1}
Assume that the effective noises $\epsilon_{\bar\eta}, \epsilon_{\bar m}$, and $\epsilon_\varsigma$ are bounded in Orlicz norms:  $\|\epsilon_{\bar\eta}\|_{\psi}\leq \omega_{\bar\eta},\|\epsilon_{\bar m}\|_{\psi}\leq \omega_{\bar m}$, and $\|\epsilon_\varsigma\|_{\varphi}\leq \omega_\varsigma$,  where $\psi, \varphi$ satisfy:   i)  $\psi(x)$ is convex and    $ \psi(x)\psi(y)\le c_1 \psi(c_0 xy) $,   $\forall x, y \ge c_2$, for some positive  $c_0, c_1, c_2$ (dependent on $\psi$ only), (ii)
 $\{\psi^{-1}(t)\}^2$ is concave  or $\{\psi^{-1}(t)\}^2\lesssim t$ on $\mathbb{R}_+$; iii) $\{\varphi^{-1}(t)\}^2$ is concave  or $\{\varphi^{-1}(t)\}^2\lesssim t$ on $\mathbb{R}_+$.
 Consider the estimator  $\hat\unv =[\hat{\bar\beta}^{\,T},\hat\gamma^{\,T}]^T$  by minimizing \eqref{l1thmloss} with $\varrho\geq \rhotwoinf$ and   $ \lambda=A\big\|\bar X^T\epsilon_{\bar\eta}\big\|_{\infty}/\mtbb $, where $A$
  a large enough constant. Then
  \small
\begin{align}\label{thm1result}
  \mathbb{E}\Big\{ \mathbf{\Delta}_{l}(\hat{\mu},\mu^*)\vee  {\etaa} \mathrm{\mathbf{D}}_2(\hat{\gamma},\gamma^*)\Big\}\,  \lesssim \  & c_\psi\omega_{\bar\eta}\,\mtbb\,\psi^{-1}(   p)\|\bar\beta^*\|_{2,1}  +\frac{1}{\etaa}\{\psi^{-1}(1)\}^2\omega^2_{\bar m}+\frac{1}{\etaa}\{\varphi^{-1}(1)\}^2\omega_\varsigma^2+\etaa\|\gamma^*\|_{2}^2.
\end{align}
 where $c_\psi =c_0    (1\vee c_1   \vee 2 \psi(c_2))^2  \psi^{-1}(1)$. \normalsize
\end{theorem}

Theorem \ref{globalgroupl1} provides a bound on prediction and estimator errors,  measured using   generalized Bregman functions. Notably, this bound does \textit{not} necessitate any regularity conditions on the design matrices or signal strength.

The assumptions (i)--(iii) on effective noise tails are mild, and the functions $\varphi$ and $\psi$ can be applied to a wide range of cases.   For example, it is straightforward to verify that    $\|\cdot\|_\psi$ can represent a   $\psi_q$-norm with $q\geq 1$  \citep{vanderVaart1996} where we can take $c_1=1, c_2=1, c_0 =2 ^{1/q}$; $\|\cdot\|_\varphi$  can be sub-Weibull  for some $q> 0$, or an   $L_q$-norm ($q\geq 2$) with  heavy polynomial tails.

The first term on the right-hand side of \eqref{thm1result}, $c_\psi\omega_{\bar\eta}\,\mtbb\,\psi^{-1}(p)\big\|\bar\beta^*\big\|_{2,1}$, is the dominant term scaling  with $p$. Remark \ref{effectivenoisevsrawnoise} emphasizes that $\epsilon_{\bar\eta}$ can have considerably lighter tails, enabling the choice of a large $\psi$ function. For instance, when $|\rho'|\leq M,|\lbbd|\leq B$, we can take  $\psi=\psi_2$,   $\omega_{\bar \eta}=c\{M/(\sigma^*\land \nu^*) + B\}$. Such a substantial $\psi$ function ensures that the error rate, which incorporates $\psi^{-1}$,  stays well controlled, even when the raw noise \eqref{rawnoisedef} exhibits heavy tails.

Furthermore,    with proper regularity conditions,     another   error bound that  depends on  $\bar\beta^*$ though its support $J^*$ can be derived.

\begin{theorem}\label{globalgrouplreconstanttheta}
Assume that the tails of effective noises  are bounded in Orlicz norms:  $\|\epsilon_\varsigma\|_{\varphi}\leq \omega_\varsigma$,    $\|\epsilon_{\bar\eta}\|_{\psi_q}\leq \omega_{\bar\eta},\|\epsilon_{\bar m}\|_{\psi_q}\leq \omega_{\bar m}$    for some  $q>0$.    Let  $\hat\unv$ denote the  optimal solution  for \eqref{l1thmloss} with $\varrho\geq \rhotwoinf$ and
 \begin{align*}
     \lambda=A\omega_{\bar\eta}(\log p)^{\frac{1}{q}} 
 \end{align*}
 for some large enough $A>0$.
Suppose that  there exist  a large $K >0$ and  a constant $\vartheta$  such that  for any $\bar\beta,\gamma$ 
\begin{align}\label{l1fastratere}
      (1+\vartheta)\lambda\mtbb\big\|(\bar{\beta}-\bar{\beta}^*)_{\mathcal{J}^*}\big\|_{2,1}\leq \mathbf{\Delta}_{l}(\mu,\mu^*)+\lambda\mtbb\big\|(\bar{\beta}-\bar{\beta}^*)_{\mathcal{J}^{*C}}\big\|_{2,1}+K\lambda^2J^*.
\end{align}
Then \begin{align}\label{l1estimationerrorrate}
   \mathbf{\Delta}_{l}(\hat{\mu},\mu^*)\vee{\etaa}\mathrm{\mathbf{D}}_2(\gamma,\gamma^*)\lesssim KA^2\omega^2_{\bar\eta}(\log p)^{\frac{2}{q}}J^*+ A^2\frac{\omega^2_{\bar m}}{\etaa}(\log p)^{\frac{2}{q}}+\frac{\omega^2_{\varsigma}}{\etaa}\{\varphi^{-1}(p^{A^q})\}^{2}+\etaa\|\gamma^*\|_{2}^2
\end{align}
holds with probability at least $1-C p^{-cA^q}$, where  $C,c$ are  positive constants.
\end{theorem}


To obtain a sufficient condition for the regularity condition  \eqref{l1fastratere}, we can  confine $\betadiff=\bar\beta-\bar\beta^*$ within a cone: $(1+\vartheta)\|\betadiff_{\mathcal{J}^*}\|_{2,1}\geq \|\betadiff_{\mathcal{J}^{C*}}\|_{2,1}$, and require either $\varrho^2\|\betadiff_{\mathcal{J}^*}\|^2_{2,1}\leq \tilde{K}J^*\mathbf{\Delta}_{l}(\mu,\mu^*)$  or $\varrho^2\|\betadiff_{\mathcal{J}^*}\|^2_{2}\leq \tilde{K}\mathbf{\Delta}_{l}(\mu,\mu^*)$ for some large $\tilde{K}>0$. These  conditions extend the compatibility    and restricted eigenvalue conditions which are   widely used in sparse  regression \citep{bickel2009simultaneous,van2009conditions}. But  \eqref{l1fastratere} is less technically demanding.

In proving the theorem, we   establish a more general result where $\psi_q$ can be replaced with a  general  $\psi$ and the appropriate choice for $\lambda$ is of the order
$
\omega_{\bar\eta}\psi^{-1}\left(p\psi {A\psi^{-1}(p) }\right).
$ 
 For more, please refer to Appendix \ref{appsub:proofofth2}.


 To  illustrate the bound \eqref{l1estimationerrorrate}, let's consider  a scenario where   $|\rho'| \vee |\lbbd|\leq M$ for some  $M>0$. Under the assumptions of  independent centered effective noise components and
 $\|\epsilon^{\mbox{\scriptsize{raw}}}_i\|_{\psi_q}\leq \omega$ for some $q\in(0,2]$,    we can deduce based on \eqref{epscal-bounded} in  Remark \ref{effectivenoisevsrawnoise} that the error bound in \eqref{l1estimationerrorrate} is of the following order (treating $K,A$ as constants):
\begin{align*} 
\begin{split}
&   \log p\cdot\Big [ \big(1\vee\frac{1}{\sigma^{*2}   }\vee \frac{1} { \nu^{*2}}  \big) M^2J^*+\frac{1}{\etaa} \big\{[(1-\frac{1}{\sigma^{*}})^2\lor (1-\frac{1}{\nu^{*}})^{2}]M^{2}+ \big(\frac{\sigma^{*}\lor\nu^{*}}{\sigma^{*}\land\nu^{*}}-1\big)^2 \big\}\Big]\\&\quad +(\log p)^{2/q}\cdot\frac{1}{\etaa} \{M\omega+  \sigma^{*}\lor\nu^{*}  \}^2+\etaa\|\gamma^*\|_{2}^2.
\end{split}\end{align*}
   The bound  varies with the number of predictors logarithmically, and  quantifies the impact of asymmetric scales in the context of sparse skewed two-part models.


\begin{remark}\label{rem-moreres}
  Under suitable regularity conditions similar to those in Theorem \ref{globalgrouplreconstanttheta},   the following $(2,\infty)$-norm bound holds (cf. Theorem \ref{inftyboundprop}   in Appendix \ref{inftyboundsubsect}): \begin{align}\label{l1infboundtext}
  \|\hat{\bar{\beta}}-\bar{\beta}^*\|_{2,\infty}
    \leq C\frac{\sqrt{K\alpha}\lor\vartheta}{\alpha\sqrt{n}}A\Big\{\omega_{\bar\eta}+\frac{\omega_{\bar m}+\omega_{\varsigma}}{\sqrt{J^*}} \Big\}(\log p)^{\frac{1}{q}}
\end{align}
with probability at least $1-C p^ {-(A\vartheta)^q  }-1/\varphi(cA\vartheta(\log p)^{ {1}/{q}})$, where  $C,c$ are positive constants and $A,K,\alpha,\vartheta$ can often be  treated as constants.  Hence with a proper signal strength
$\min_{k\in\mathcal{J}(\bar\beta^*)}\|\bar\beta_k^*\|_2>  2CA(\sqrt{K\alpha}\lor\vartheta) \big\{\omega_{\bar\eta}+(\omega_{\bar m}+\omega_{\varsigma})/\sqrt{J^*}\big\}(\log p)^{\frac{1}{q}}/(\alpha\sqrt{n})$,  \eqref{l1infboundtext} guarantees faithful variable
selection,   $\mathcal{J}^*\subset \hat{\mathcal{J}}$,  with high probability.

Finally,  our analysis can be extended to a  general $P$  beyond the $\ell_1$-type penalty.
See Appendix \ref{generalpenaltysubsect} for more details.
\end{remark}

\section{Experiments}
\label{sec:exp}
We conducted a variety of synthetic and real data experiments to evaluate the performance of the proposed method.  Due to limited space, we  only present  a selection of   our  data analyses in the following subsections. Interested readers may refer to
Appendix \ref{appendixb} for more experiment results.
\subsection{Simulations}
\label{subsec:basicsimu}
In this part, we conduct simulation experiments to compare our proposed methods with some popularly used approaches for  skewed estimation. The predictor matrix $X=[X_1,\ldots,X_n]^T\in \mathbb{R}^{n\times p}$ is generated by $X_i\overset{i.i.d.}{\sim}N(0,\Sigma)$, where   $\Sigma=[\kappa^{|i-j|}]$ has a Toeplitz structure.  The response vector is generated according to $y=X\beta^*+1\alpha^*+\epsilon$, where $\epsilon_i \overset{i.i.d.}{\sim}\mbox{SP}^{(\phi)}( \sigma^*,\nu^*,m^*)$. In the   setups to be introduced, we  set  $\alpha^*=0$ and $\beta^*=[12,13,14]^T$.
The pivotal point will be  chosen as  the mean, median, mode,  quartiles, and more. 

The following methods are included for comparison: quantile regression (QR) \citep{koenker1978regression}, Bayesian quantile regression (BQR) \citep{yu2001bayesian}, $Z$-estimation quantile regression (ZQR) \citep{bera2016asymmetric},
adaptive M-estimation  (AME) \citep{yang2019robust}, epsilon-skew-normal regression (ESN) \citep{Mudholkar:2000}, in addition to SPEUS.
  The first two methods  require the user to specify a quantile parameter \citep{koenker2009quantreg,benoit2017bayesqr}; we employ    the \textit{oracle  quantile} ($\Phi(m^*)\sigma^*/[\Phi(m^*)\sigma^*+\{1-\Phi(m^*)\nu^*\}]$ (computed using  the truth)  in all experiments, and so      denote the  methods by QR$^*$ and BQR$^*$, respectively. In BQR$^*$,   the posterior mean  estimate is obtained   with 4000  MCMC draws after 1000 burn-in samples.
For  SPEUS, the scale parameters $\sigma$ and $\nu$, as well as the pivotal point $m$, are all considered {unknown} and are estimated from the data.  Given each   setup, we repeat the experiment for 50
times and evaluate the performance of each method based on       $\mbox{Err}(\beta)$,    $\mbox{Err}$($\sigma$) and  $\mbox{Err}$($\nu$).        $\mbox{Err}(\beta)$ is   the (absolute) root-mean-square error   on $\beta$, and   $\mbox{Err}$($\sigma$) and  $\mbox{Err}$($\nu$)       denote the (relative) root-mean-square errors on $\sigma$ and $\nu$, i.e.,   $\mbox{Err}(\sigma) =\{ \sum^{N}_{t=1}(\hat{\sigma}_t/\sigma^*-1)^2/N\}^{1/2}$, where $\hat{\sigma}_t$ is the estimate on the $t$th simulation dataset.


\textbf{Ex 1.} (Skewed Gaussian and skewed Laplace, $m^*=0$): Let $\phi$ be the standard normal or Laplace density,    $n=300 $,  $\kappa= 0.5$,    $\sigma^*=0.2$ and  $\nu^*=0.4,0.6,1.2$.

\textbf{Ex 2.} (Skewed Gaussian and skewed Laplace, $m^*\neq 0$): Let $\phi$ be the standard normal or Laplace density,     $n=300 $,  $\kappa=0.2$, $m^*=\Phi^{-1}(0.75)$ (third quartile),   $\sigma^*=0.3$ and  $\nu^*=0.5,0.7,0.9$. (We also tried $m = \Phi^{-1}(0.25)$ but the results are similar and omitted.)

\begin{table}[!h]
\scriptsize
\centering
\resizebox{\columnwidth}{!}{
\begin{tabular}{ccccccccccccc}
\hline\hline\noalign{\vskip 1mm}
     \multicolumn{13}{c}{Skewed Gaussian }\\
   \cline{3-13}
          &       &   \multicolumn{3}{c}{$\nu^*/\sigma^*=2$ }                    &  & \multicolumn{3}{c}{$\nu^*/\sigma^*=3$ }                     &  & \multicolumn{3}{c}{$\nu^*/\sigma^*=6$ }
     \\
     \cline{3-5} \cline{7-9} \cline{11-13}

 &    & Err($\beta$)  & Err($\sigma$) & Err($\nu$) &   & Err($\beta$)  & Err($\sigma$) & Err($\nu$) &    & Err($\beta$)  & Err($\sigma$) & Err($\nu$) \\
 \hline
QR$^*$ &  & 0.05 & --- & --- &  & 0.06 & ---   & --- &          & 0.08  &   ---    &  ---   \\
BQR$^*$   &  & 0.04 & 0.20 & 0.20 &  & 0.05 & 0.20  & 0.20 &          & 0.07  &  0.20    &  0.20   \\
AME   &  & 0.04 & 0.44 & 0.44 &  & 0.04 & 0.44   & 0.44 &          & 0.06  &   0.44   &  0.44  \\
ESN  &  & 0.04 & 0.42 & 0.42 &  & 0.04 & 0.42   & 0.42 &          & 0.06  &   0.42  &  0.42 \\
 SPEUS &  & 0.04 & 0.13 & 0.08 &  & 0.05 & 0.18   & 0.06 &          & 0.06  &  0.29    &  0.05 \\
  \hline\noalign{\vskip 1.5mm}
 \multicolumn{13}{c}{Skewed Laplace}\\   \cline{3-13}
          &       &   \multicolumn{3}{c}{$\nu^*/\sigma^*=2$ }                    &  & \multicolumn{3}{c}{$\nu^*/\sigma^*=3$ }                     &  & \multicolumn{3}{c}{$\nu^*/\sigma^*=6$ }
     \\
     \cline{3-5} \cline{7-9} \cline{11-13}

 &    & Err($\beta$)  & Err($\sigma$) & Err($\nu$) &   & Err($\beta$)  & Err($\sigma$) & Err($\nu$) &    & Err($\beta$)  & Err($\sigma$) & Err($\nu$) \\
  \hline
QR$^*$ &  & 0.04 & --- & --- &  & 0.05 & ---   & --- &          & 0.17  &   ---    &  ---   \\
BQR$^*$   &  & 0.04 & 0.05 & 0.05 &  & 0.05 & 0.05   & 0.05 &          & 0.17  &   0.06    &  0.06    \\
AME   &  & 0.04 & 0.23 & 0.20 &  & 0.05 & 0.25   & 0.21 &          & 0.17  &   0.24    &  0.20   \\
ZQR &  & 0.04 & 0.08 & 0.06 &  & 0.05 & 0.09   & 0.06 &          & 0.17  &   0.12   &  0.06  \\
 SPEUS &  & 0.04 & 0.10 & 0.07 &  & 0.05 & 0.12   & 0.07 &          & 0.17  &  0.18    &  0.06 \\
        \hline\hline
\end{tabular}
}
\caption{\footnotesize Skewed normal and skewed Laplace with   pivotal point at zero ({Ex 1}) }
\label{sknormallaplace_trivialm}
\end{table}

\begin{table}[!htp]
\scriptsize
\centering
\resizebox{\columnwidth}{!}{
\begin{tabular}{ccccccccccccc}
\hline\hline\noalign{\vskip 1mm}
     \multicolumn{13}{c}{Skewed Gaussian }\\
      \cline{3-13}
          &       &   \multicolumn{3}{c}{$\nu^*/\sigma^*=1.7$ }                    &  & \multicolumn{3}{c}{$\nu^*/\sigma^*=2.3$ }                     &  & \multicolumn{3}{c}{$\nu^*/\sigma^*=3$ }
     \\
     \cline{3-5} \cline{7-9} \cline{11-13}

 &    & Err($\beta$)  & Err($\sigma$) & Err($\nu$) &   & Err($\beta$)  & Err($\sigma$) & Err($\nu$) &    & Err($\beta$)  & Err($\sigma$) & Err($\nu$) \\
 \hline
QR$^*$ &  & 0.16 & --- & --- &  & 0.17 & ---   & --- &          & 0.15  &   ---    &  ---    \\
BQR$^*$  &  & 0.16 & 1.04 & 0.30 &  & 0.16 & 0.67   & 0.43&          & 0.14  &   0.46    &  0.50  \\
AME   &  & 0.06 & 0.43 & 1.34 &  & 0.07 & 0.47   & 0.93 &          & 0.08  &  0.51    &  0.71  \\
ESN  &  & 0.06 & 0.56 & 1.05 &  & 0.06 & 0.56   & 0.73 &          & 0.07  &  0.57   &  0.58 \\
 SPEUS &  & 0.04 & 0.10 & 0.08 &  & 0.04 & 0.14   & 0.07&          & 0.05  &   0.17   &  0.06 \\
  \hline\noalign{\vskip 1.5mm}
 \multicolumn{13}{c}{Skewed Laplace}\\   \cline{3-13}
          &       &   \multicolumn{3}{c}{$\nu^*/\sigma^*=1.7$ }                    &  & \multicolumn{3}{c}{$\nu^*/\sigma^*=2.3$ }                     &  & \multicolumn{3}{c}{$\nu^*/\sigma^*=3$ }
     \\
     \cline{3-5} \cline{7-9} \cline{11-13}

 &    & Err($\beta$)  & Err($\sigma$) & Err($\nu$) &   & Err($\beta$)  & Err($\sigma$) & Err($\nu$) &    & Err($\beta$)  & Err($\sigma$) & Err($\nu$) \\
  \hline
QR$^*$ &  & 0.17 & --- & --- &  & 0.18 & ---   & --- &          & 0.17  &   ---    &  ---   \\
BQR$^*$   &  & 0.17 & 1.07 & 0.31 &  & 0.17 & 0.69   & 0.43 &          & 0.16  &  0.48    &  0.50   \\
AME   &  & 0.09 & 0.19 & 1.32 &  & 0.11 & 0.26   & 0.90 &          & 0.11  &   0.31    & 0.66   \\
ZQR &  & 0.10 & 0.52 & 0.55 &  & 0.11 & 0.52   & 0.34 &          & 0.12  &   0.52    & 0.22  \\
 SPEUS &  & 0.04 & 0.10 & 0.08 &  & 0.05 &0.10   & 0.07 &          & 0.06  & 0.14    & 0.07 \\
        \hline\hline
\end{tabular}
}
\caption{\footnotesize Performance comparison for skewed normal and skewed Laplace with a nontrivial pivotal point ({Ex 2})}
\label{sknormallaplace_nontrivialm}
\end{table}

Tables \ref{sknormallaplace_trivialm} and \ref{sknormallaplace_nontrivialm} provide a detailed comparison of the performance of various methods. According to Table \ref{sknormallaplace_trivialm},    when $m^*=0$, the $\beta$-errors   do not differ significantly. In this setup,   BQR$^*$, ZQR and SPEUS all perform well.
In the setup of Ex 2 with a nontrivial pivotal point, as shown in Table \ref{sknormallaplace_nontrivialm},          SPEUS significantly outperforms the other methods in both  Gaussian and Laplace cases. We also tried other quantiles for $m^*$ (results not reported here) and found that
BQR, AME and ZQR typically produce  highly inaccurate  scale estimates.   In contrast,
SPEUS is much more successful at accurately recovering the true location and scales,  and is  not sensitive to different values of $m^*$ even in the  heavy-tailed cases.

\subsection{Abalone Age}
\label{subsec:abalone}
Determining the age of abalone is a tedious and challenging task that often involves counting growth rings under a microscope. We use 7 physical measures of blacklip abalone, including length, diameter, height,   weight and others, to predict the age of   1,526   infant samples originally collected by  \cite{nash1994population}. We split the data   into a training set ($70\%$) and a test set $(30\%)$.
The abalone age dataset is usually analyzed using a regression model based on ordinary least squares (OLS)
\citep{gitman2018novel,chang2019prediction}, but the histogram  in the left panel of   Figure \ref{ageplot} suggests the  possibility of  skewness.
 To address this, we employed SPEUS by applying  skewed pivot-blend  to the normal density function.   Moreover, we   included the skewed methods    AME \citep{yang2019robust} 
and ESN \citep{Mudholkar:2000}  for comparison.

\begin{figure}[!htp]
\begin{subfigure}[t]{0.45\textwidth}
  \includegraphics[width=\textwidth, height=4.5cm]{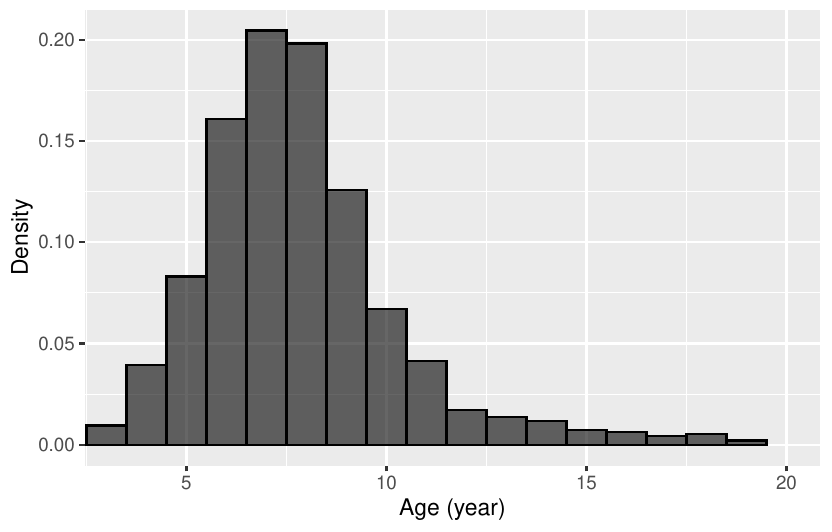}
  \label{agehist}
\end{subfigure}
\hfill
\begin{subfigure}[t]{0.45\textwidth}
  \includegraphics[width=\textwidth, height=4.5cm]{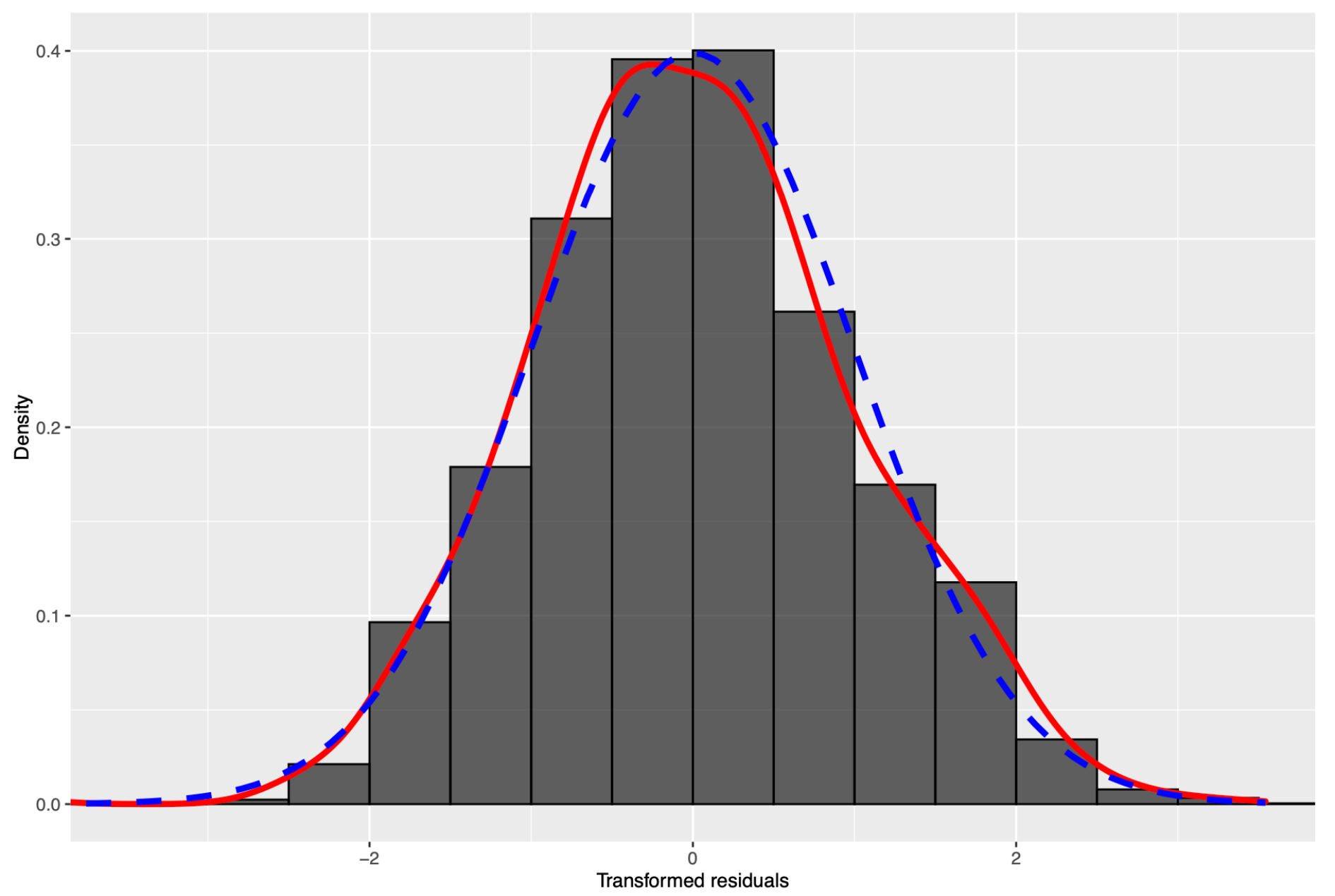}
  \label{bootstrapresultboxplot}
\end{subfigure}
\caption{ \footnotesize Left:  histogram of abalone age. Right:  ``back-transformed" SPEUS residuals,  the associated  density estimate (red solid curve), and
the standard normal density (blue dashed curve).}
\label{ageplot}
\end{figure}

%
%

\begin{figure}[!ht]
\begin{subfigure}{0.48\textwidth}
  \includegraphics[width=\textwidth, height=4.5cm]{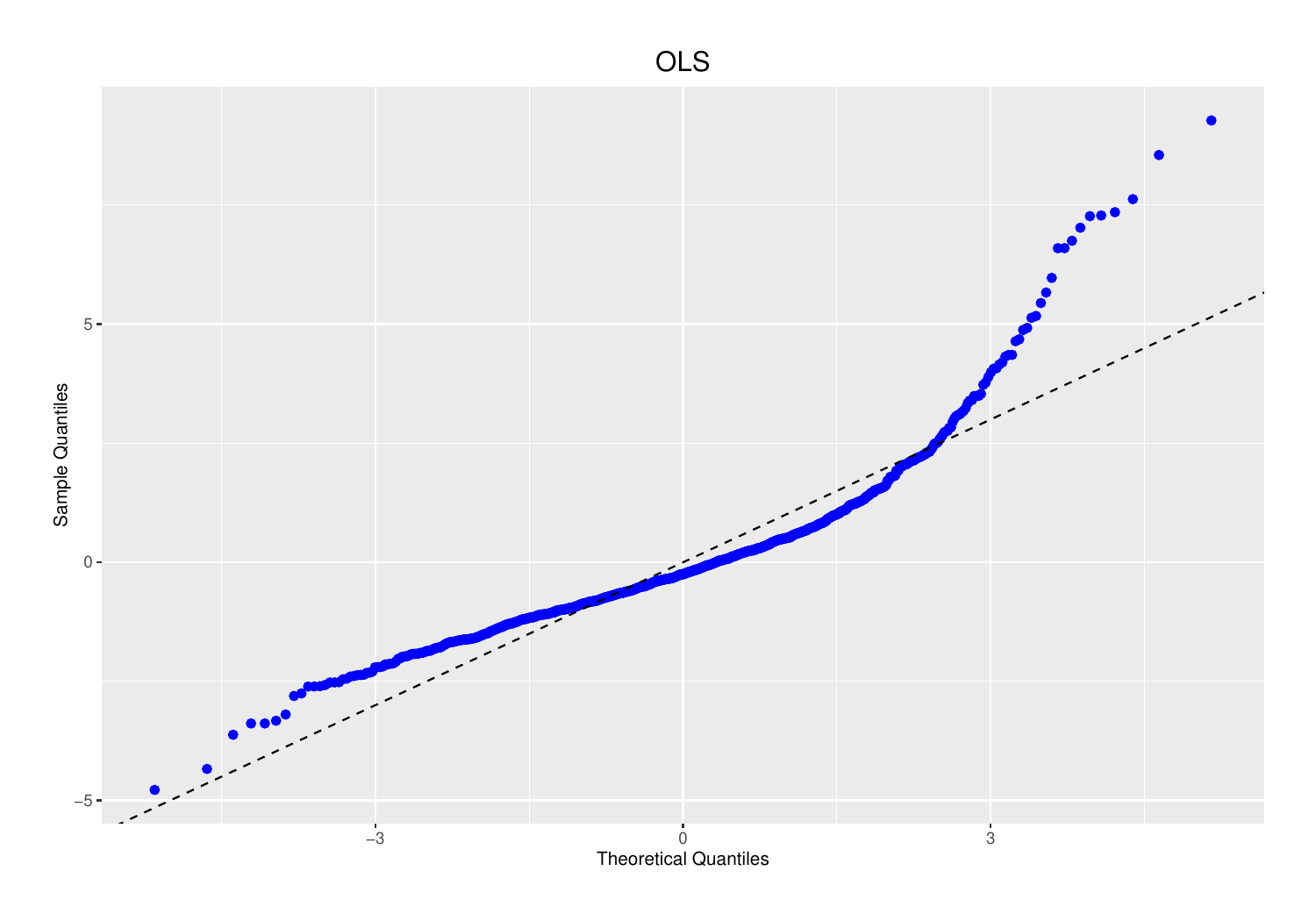}
  \caption{\scriptsize OLS}
  \label{ols_trainingqqplot}
\end{subfigure}
\hfill
\begin{subfigure}{0.48\textwidth}
  \includegraphics[width=\textwidth, height=4.5cm]{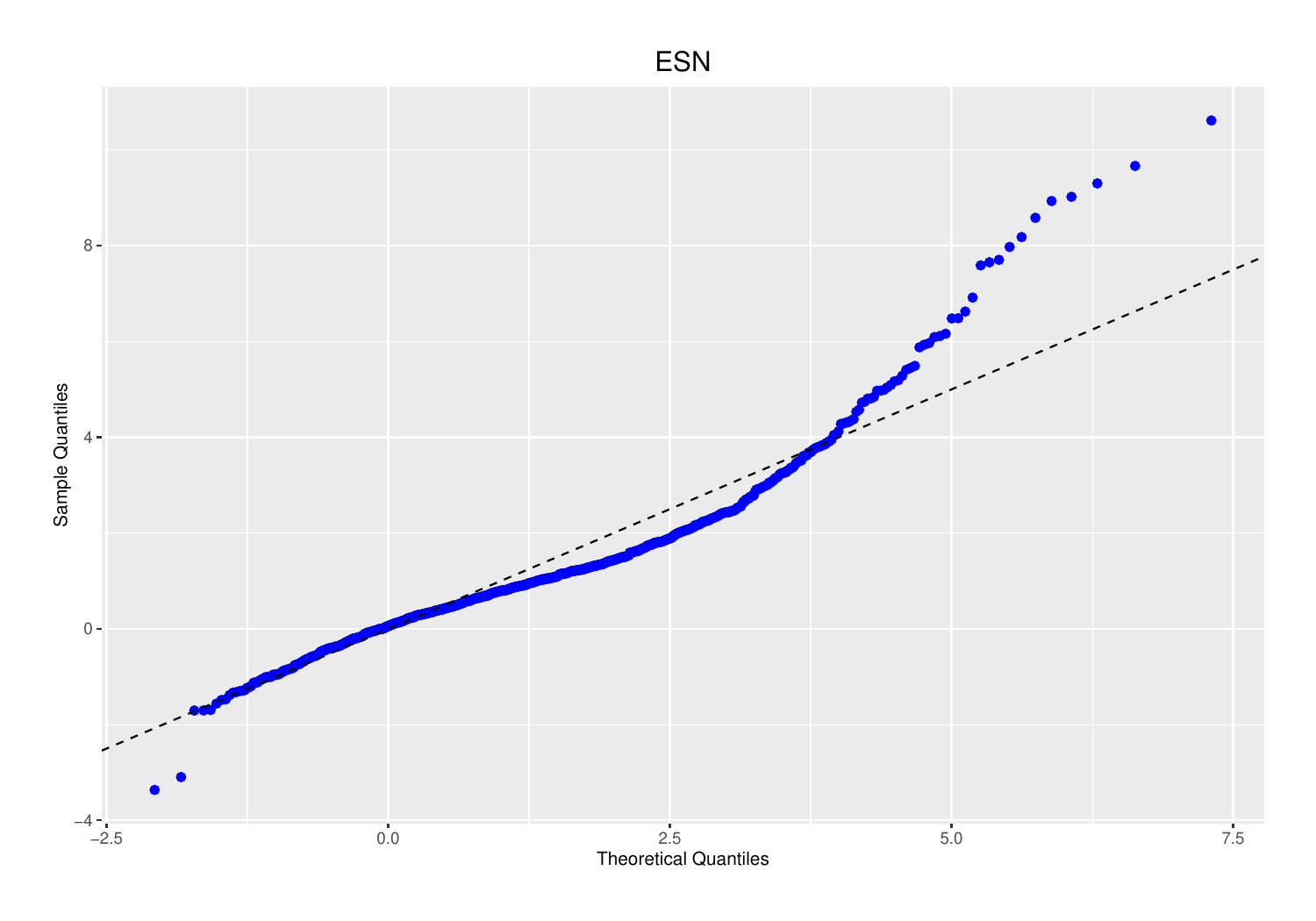}
  \caption{\scriptsize ESN}
  \label{sn_trainingqqplot}
\end{subfigure}
\begin{subfigure}{0.48\textwidth}
  \includegraphics[width=\textwidth, height=4.5cm]{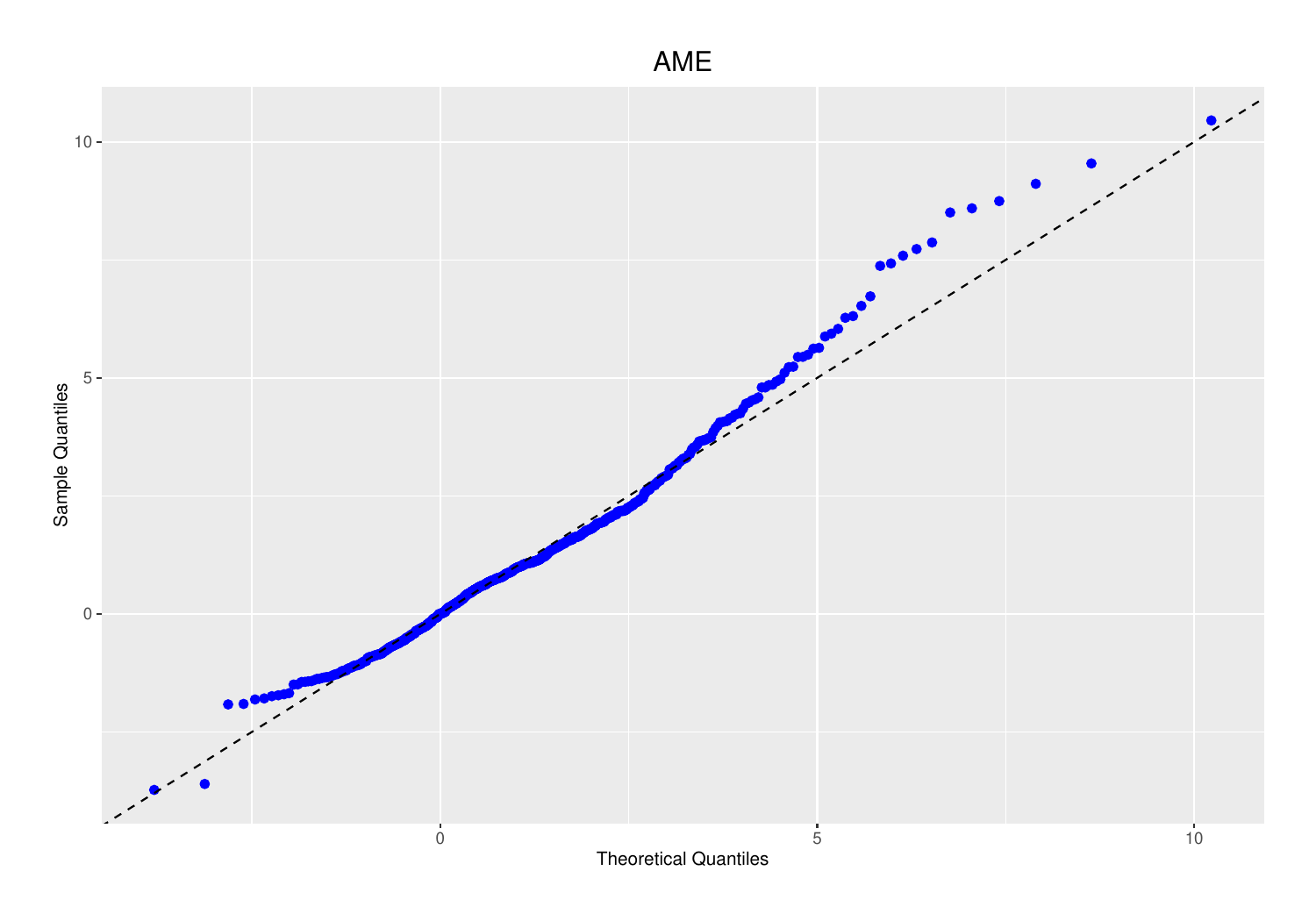}
  \caption{\scriptsize AME}
  \label{ame_training_qqplot}
\end{subfigure}
\hfill
\begin{subfigure}{0.48\textwidth}
  \includegraphics[width=\textwidth, height=4.5cm]{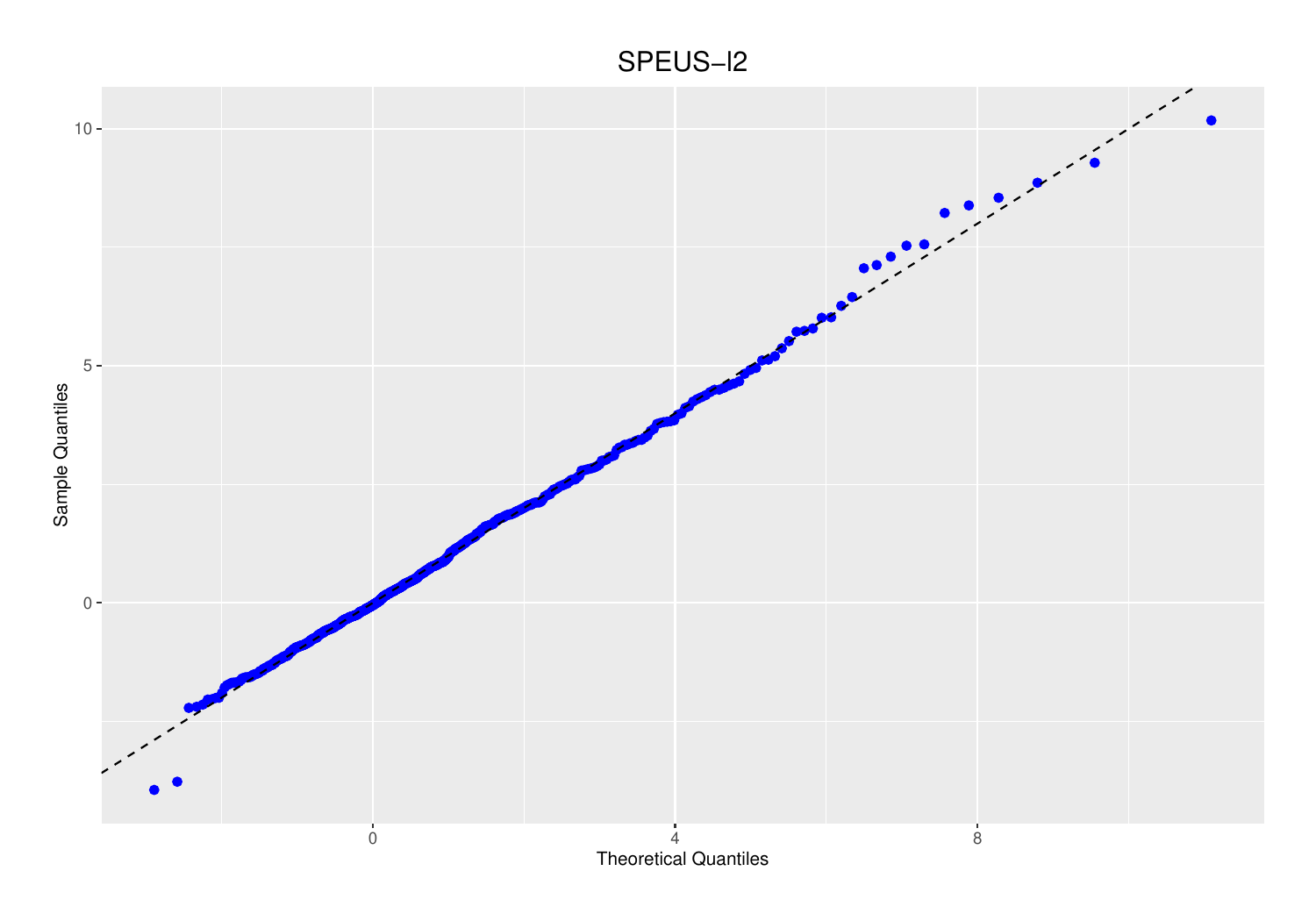}
  \caption{\scriptsize SPEUS}
  \label{speus_l2_trainingqqplot}
\end{subfigure}
\caption{ \footnotesize Q-Q plots of  model residuals  on the  abalone data. 
}
\label{trainingqqplots}
\end{figure}

%

To assess the goodness of fit of several different models, we present Q-Q plots of the residuals in Figure \ref{trainingqqplots}. Although OLS is widely adopted for the data, it clearly demonstrates a lack of fit. The ESN and AME models offer significant improvement through  scale and tail adjustments,  but their Q-Q plots still exhibit substantial right-skewness.  In contrast, the sample quantiles in the SPEUS model nearly match the theoretical quantiles. The symmetry after the back-transform, as shown in the right panel of Figure \ref{ageplot}, corroborates this point.

To compare the fit of the  methods, which optimize different criteria based on various distributional assumptions, we conducted   Kolmogorov-Smirnov tests   and calculated the associated p-values: 0.008 for OLS,  0.003 for ESN, 0.08 for AME, and 0.3 for SPEUS. These findings validate our model's superior fit. It effectively addresses pivotal point and skew effects in the data, surpassing alternative approaches that rely solely on adjustments to intercept and scale.

%

\begin{figure}[!htp]
  \includegraphics[width=.5\textwidth]{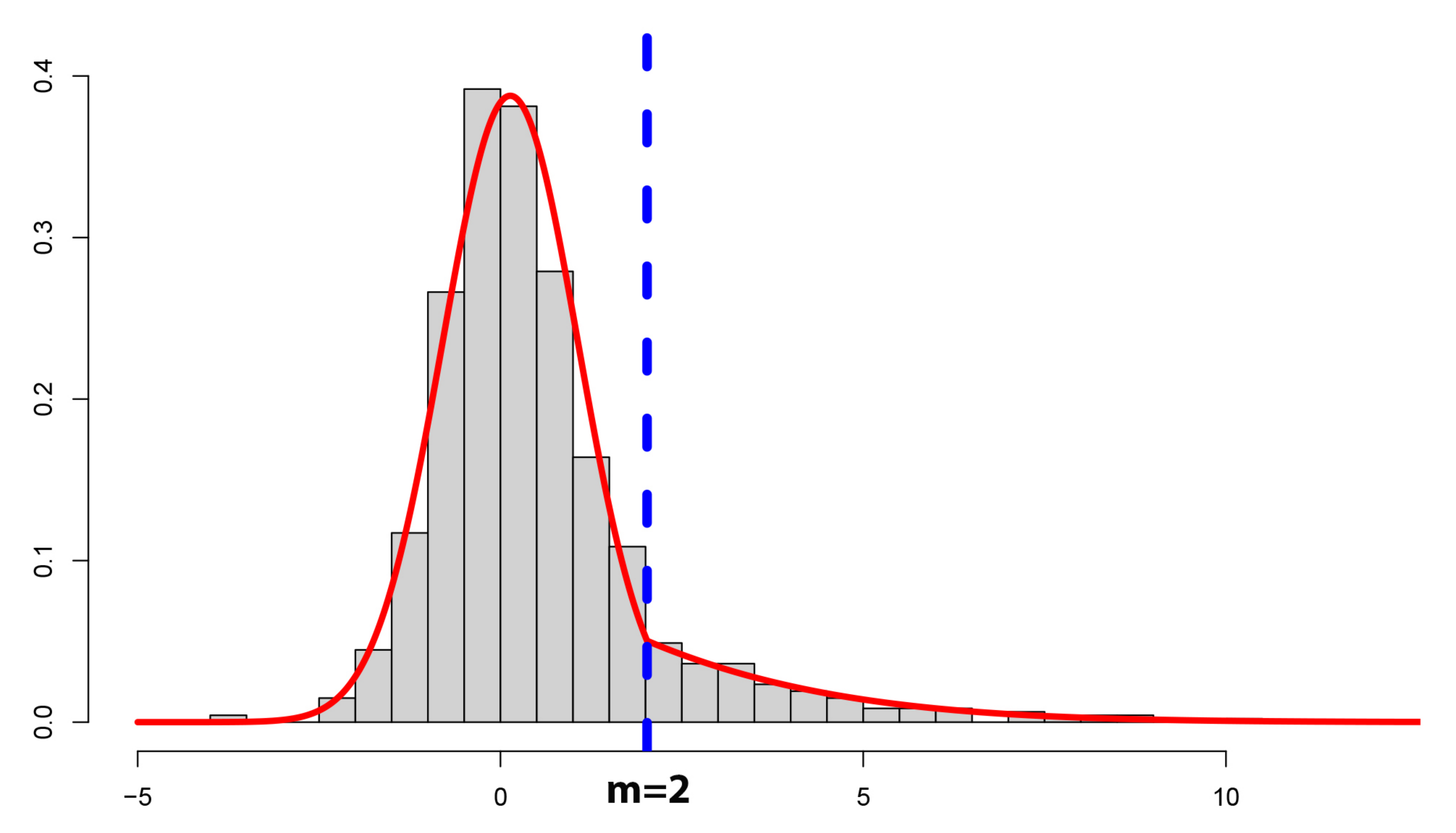}
  \includegraphics[width=.4\textwidth]{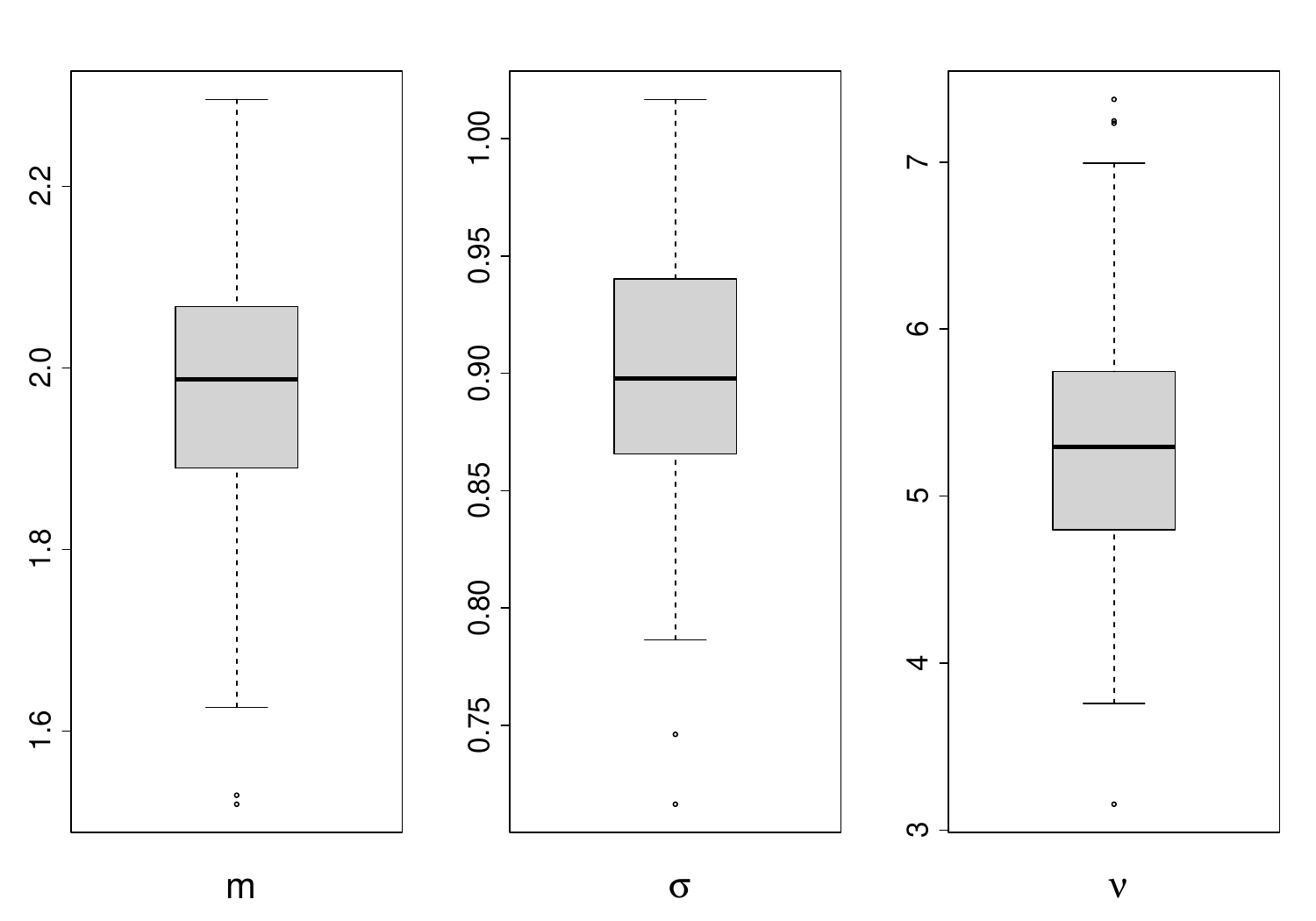}
\caption{\footnotesize  Abalone data. Left:  Histogram of SPEUS residuals with a labeled pivotal point. Right: Bootstrap estimates of $m,\sigma,\nu$ (with  100 replications), where $m\ne 0$ and $\sigma\ne \nu$ are significant.}
\label{histogramspeusplot}
\end{figure}

Figure \ref{histogramspeusplot} shows the histogram of the SPEUS residuals along with  the  bootstrap results for $m, \sigma, \nu$.
The   scale estimates $\hat \sigma=0.9$, $\hat \nu=5.4$, with  $90\%$ confidence intervals  $[0.8,1.0]$ and $[4.0,6.5]$, respectively, suggest significant skewness in the data compared to a standard Gaussian distribution. The
estimated pivotal point,  $\hat m= 2$    (with a  90\% confidence interval $[1.7,2.2]$)   is likely associated with
the legal minimum size limits in Tasmania  during the 1980s (where and when the data were collected).
The determination of size limits included adding an estimated two years' growth to the size at which abalone reached sexual maturity in different areas, aiming to ensure abalone could reproduce before being harvested \citep{tarbath1999estimates}. However, blacklip abalone do not mature in size, and the significant growth variability among various abalone stocks led to frequent changes in size limits, impacting abalone of various ages.
 Our estimated pivotal point appears to correspond with the 2-year protection regulation.


\subsection{Medical Expenditure}

Modeling medical cost data and identifying relevant predictors are valuable yet demanding tasks. MEPS conducts large-scale surveys across the United States and provides nationally representative information about medical expenditures. We  model medical expenditures on 17 features on a subset (stratum ID 2109, third PSU) of the 2019 MEPS data, which includes 150 participants. Of the 17 features, 15 are from \cite{linero2020semiparametric}, including, for example, the amount of total utilization of prescribed medications (\texttt{RXTOT19}), the number of dental care visits in 2019 (\texttt{DVTOT19}), age (\texttt{AGE19X}), each participant's rating about their own health status (\texttt{RTHLTH31}), and categorized family income (\texttt{POVCAT19}). We also include the variables \texttt{SEX} and \texttt{ACTLIM31}, with the latter being a binary variable indicating whether a participant has any physical restrictions  that impede his/her ability to engage in physical labor.

\begin{figure}[!ht]
\begin{subfigure}{0.47\textwidth}
  \includegraphics[width=\textwidth]
{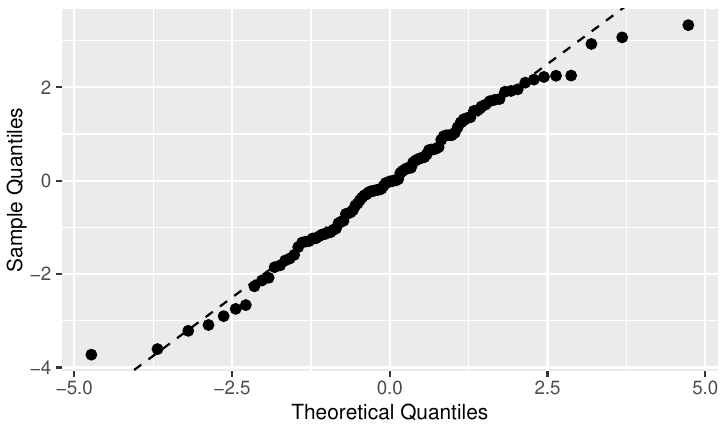}
  \label{qqplot_SPEUS_huber}
\end{subfigure}
\hfill
\begin{subfigure}{0.47\textwidth}
  \includegraphics[width=\textwidth]
{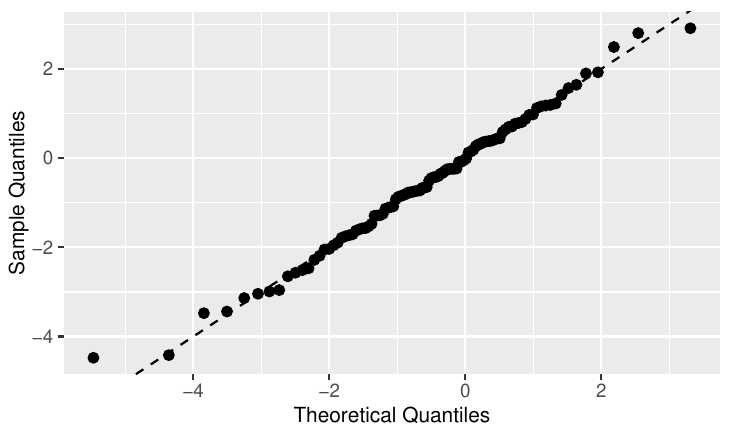}
  \label{qqplot_sigma_Equal_nu}
\end{subfigure}
\hfill
\caption{\footnotesize Q-Q plots of  the continuous-part  residuals on MEPS. The left panel corresponds to the standard log-normal two-part model and the right panel corresponds to its skewness-enhanced counterpart.  \label{trainingqqplots_meps}}
\end{figure}

Traditional medical cost data analysis often employs a two-part model with logistic and log-normal components. We compared this to the sparse skewed two-part model (cf. Section \ref{skewedtwo-partmodelsection}) using a logarithmic function for
$T$ and Huber's loss for $\rho$.
The regularization parameter is tuned by 5-fold  selective cross-validation \citep{she2019cross}.  Figure \ref{trainingqqplots_meps} demonstrates the superior fit of the latter  in terms of the continuous component.
In binary component analysis, 100 repeated classification tests on 75/25 training/test splits show that our method improved accuracy from 78\%  to 84\%.


\begin{figure}
\centering
  \centering
  \includegraphics[width=.75\textwidth]{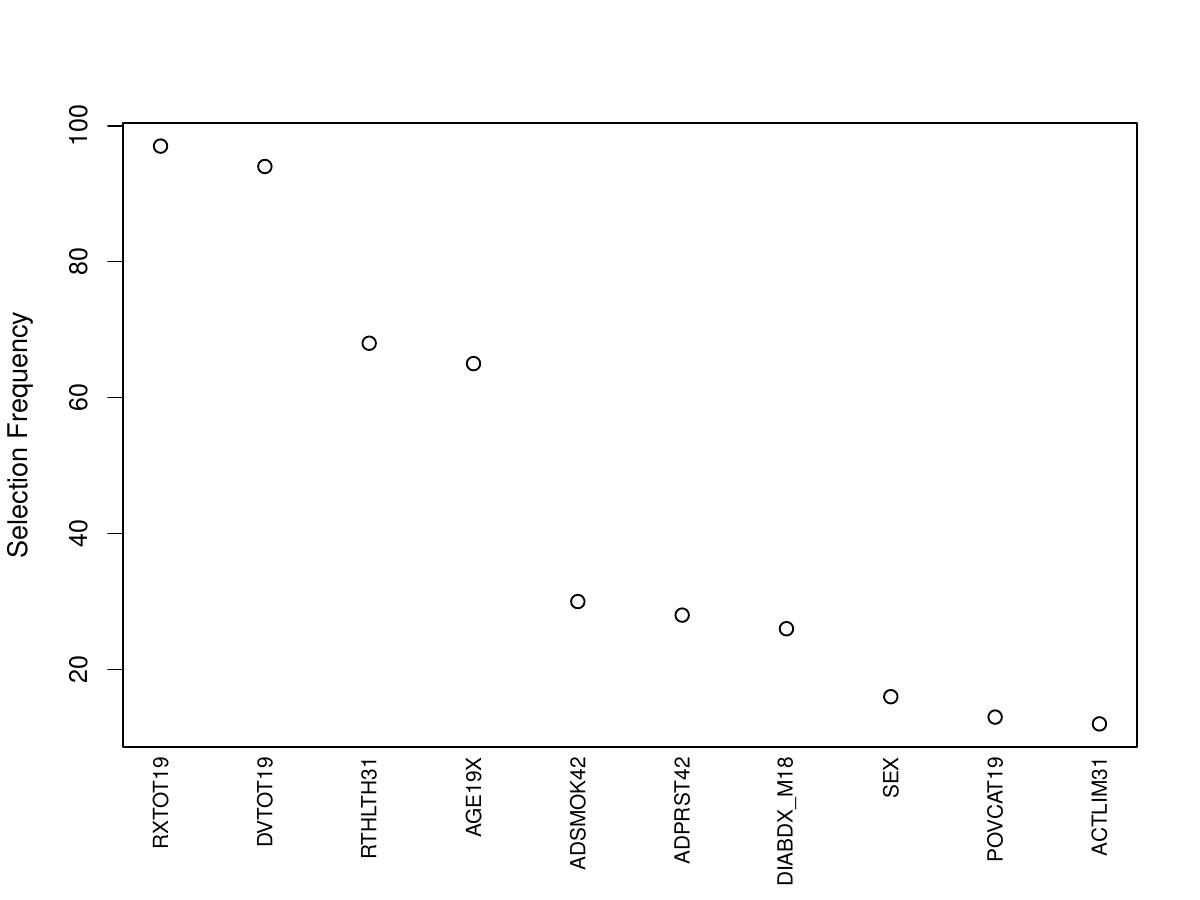}
\caption{\footnotesize Selection frequencies of the top 10 MEPS features over the 100 bootstrap replications. \label{fig:MEPSbtfreqs}  } 
\end{figure}



Next, we analyze the variable selection outcomes using the proposed model. By bootstrapping the data 100 times, we plot the selection frequencies of the top 10 variables in  Figure \ref{fig:MEPSbtfreqs}.   It is  worth noting    that setting $\sigma=\nu$ resulted in all variables being selected at  low frequencies, less than $23\%$. This emphasizes the profound influence of skewness on  variable selection.

According to Figure \ref{fig:MEPSbtfreqs},  the first 4 variables, \texttt{RXTOT19}, \texttt{DVTOT19}, \texttt{AGE19X}, and \texttt{RTHLTH31}, exhibit   high selection frequencies ($>60\%$). In contrast, other variables exhibited significantly lower selection frequencies. Below, we provide some practical explanation and guidance regarding these four variables.


First, \texttt{RXTOT19}, representing the count of a person's total prescribed medications, is identified   as highly influential in predicting medical expenditures.  This predictor has the highest correlation (0.43) with the response   among all the predictors. Our conclusion is consistent with \citet{holle2021trends}, highlighting that prescribed medication expenses   constitutes a substantial portion of medical costs in the USA.
Moreover,  the patient's age (\texttt{AGE19X}) and  self-perceived health status (\texttt{RTHLTH31}) emerge as significant predictors influencing medical costs. This discovery aligns with   \cite{axon2021patterns}.

Perhaps  interestingly, our analysis also reveals that the number of dental care visits (\texttt{DVTOT19}), selected over 90\% of the time,  plays a significant role in determining the total  medical costs. Its contribution appears to be unique, as it has low correlations ($<0.09$) with the other three major predictors (the number of prescriptions, self-perception, and age). Beyond the direct costs of dental care, a plausible explanation could be that individuals with regular dental visits may be more health-conscious and have higher incomes, making them more willing to spend on healthcare.

\section{Summary}
Skewness poses a significant challenge in data science, and  many approaches  attempt to model skewness by introducing different scales based on the median of a symmetric, unimodal density. This paper    introduced a novel two-piece density family constructed through skewed pivot-blend. ``Pivot" refers to the central reference point around which different affine transformations are applied to two conditional densities, and it can be positioned anywhere.   ``Blend" signifies the merging of these asymmetrically scaled densities  using appropriate mixings to create a new continuous density.

We proposed a joint modeling framework that simultaneously estimates scales, the pivotal point, and other location parameters. In particular, we argued that the pivotal point does not correspond to the intercept when  skewness is present, a key aspect  previously overlooked in the literature.  In practice, the inclusion of a single pivotal point parameter significantly enhances a model's capacity in real-world applications.

As an important  application, the paper also   investigated sparse skewed two-part models, a problem that has recently gained much attention in  biomedical and econometric studies.  Our non-asymptotic analysis showcases how skewness in random samples, especially those with potentially heavy tails, can affect statistical accuracy. The  quantification of the impact of asymmetrical scales on the choice of regularization parameters and the rates of statistical error  provides an insightful examination of skewness within a finite-sample context.

We aim to raise data analysts' awareness of data skew,  as well as
 potential distortions that can arise when applying common transformations and conventional log-likelihoods.  The technique of skewed pivot-blend  offers an effective strategy for mitigating these challenges.

\appendix
\numberwithin{equation}{section}
\numberwithin{theorem}{section}
\numberwithin{lemma}{section}
\numberwithin{definition}{section}
\numberwithin{table}{section}
\numberwithin{figure}{section}
\numberwithin{remark}{section}

\section{Technical Details}\label{appendixa}

Throughout the proofs, we use $C, c$ to denote positive constants, and they are not
necessarily the same at each occurrence. Given a matrix $A$, we use $\mathcal{R}(A)$ to denote the column space (range) of $A$ and  $A^+$ to denote the Moore-Penrose inverse of $A$.  Denote by $P_A=A(A^TA)^+A^T$ the orthogonal projection on $\mathcal{R}(A)$.

\subsection{Basics} 
\label{subsec:orlicz}
 The conventional   definition of  Orlicz $\psi$-norms goes as follows:
Given  a strictly increasing convex function   $\psi$  on $\mathbb R_+ := [0,+\infty)$ with $\psi(0)=0$, then the Orlicz $\psi$-norm of a random variable $Y$ is defined as  $  \|Y\|_{\psi}=\inf\{t>0: \mathbb{E}\psi(|Y|/t)\leq 1\}$.


 Some well-known examples of Orlicz $\psi$-norms  are the $L_q$-norms: $\|Y\|_q=(\mathbb{E}|Y|^q)^{1/q}$ associated with
$
  \psi(x)=x^q $ $(q\geq 1)$ (e.g.,  the Pareto distribution), and the $\psi_q$-norms ($q\geq1$) with
 \begin{align}\label{psiqfun}
     \psi_q(x)=\exp(x^q)-1.
 \end{align}
\eqref{psiqfun} encompasses both  sub-Gaussian   and sub-Exponential type  random variables for $q=2,1$,  without the requirement for the random variables to be centered.

Our upcoming theorems  often relax the strict requirements of \textit{strict} monotonicity and \textit{convexity} for $\psi$.
This flexibility allows us to handle random variables with much heavier tails.   For instance, we consider the extension of \eqref{psiqfun}, known as \textit{sub-Weibull} random variables, which have finite $\psi_q$-norms for   $q>0$   (cf. \cite{kuchibhotla2018moving}). As $0 < q < 1$, $\psi_q$ is nonconvex, and these random variables, including Weibull, exhibit heavier tails compared to sub-Exponential ones \citep{gotze2021concentration}.

Throughout the paper, when  we refer to the Orlicz norm of a random variable, denoted as  $\| \cdot \|_\psi$ or sometimes  $\| \cdot \|_\varphi$, it is always understood  that $\psi(\cdot)$ (or $\varphi(\cdot)$) is  an  nondecreasing  function defined on $\mathbb{R}_+$ with $\psi(0)=0$.  Theorem \ref{globalgroupl1} also requires       $\psi$  to satisfy the regularity condition   $\allowbreak\limsup_{x,y\to \infty}\psi(x)\psi(y)/\psi(cxy)\allowbreak<\infty$ for some constant $c>0$ \citep{vanderVaart1996}. It is easy to verify that  all these conditions are met by $L_q$ with $q>0$  and $\psi_q$ with $0<q\le 2$. Orlicz norms provide a useful framework for analyzing skewed random variables, including those without zero mean.

In this paper, we define that a random vector  $\epsilon\in\mathbb R^n$ has its Orlicz $\psi$-norm $\|\epsilon\|_{\psi}$ bounded above by $\omega$ if
\begin{align}
\|\langle\epsilon,\alpha\rangle\|_{\psi}\leq \omega\|\alpha\|_2,\ \forall \alpha\in\mathbb R^n. \label{orliczvecdef}
\end{align}
  Note that \eqref{orliczvecdef} is defined using the Euclidean norm $\|\cdot \|_2$, and the function $\psi$ may not necessarily be convex.   Furthermore, the components of $\epsilon$ are not required to be independent or centered. However, if $\epsilon$ does have centered, independent components, its vector Orlicz $\psi$-norm is bounded by the largest Orlicz $\psi$-norm among its components (up to a multiplicative constant)..
\begin{lemma}\label{componenttovectorlemma}
Let    $\epsilon_1,\ldots,\epsilon_n$ be centered, independent random variables satisfying  $\|\epsilon_i\|_{\psi_q}\leq \omega$ for some $q\in(0,2]$. Given any $\alpha \in \mathbb{R}^n$, we have $\|\langle \alpha,\epsilon \rangle\|_{\psi_q}\leq C\|\alpha\|_2\omega$, where $C$ is a constant  depending on $q$ only.
\end{lemma}

To prove the lemma, we first  introduce two lemmas. The first is       Theorem 1.5 in \cite{gotze2021concentration}. \begin{lemma}\label{lemmagotze}
 Let $\epsilon_1,\ldots,\epsilon_n$ be independent random variables satisfying $\|\epsilon_i\|_{\psi_q}\leq \omega$ for some $q\in(0,2]$. Let $f(\epsilon):\mathbb{R}^n\to \mathbb{R}$ be a polynomial of  degree $D\in \mathbb{N}$ and denote by $f^{(d)}$ the $d$-tensor of its $d$-th order partial derivatives for $1\leq d\leq D$. Then for all $t>0$, we have
 \begin{align}
     \mathbb{P}(|f(\epsilon)-\mathbb{E}f(\epsilon)|\geq t)\leq 2\exp\big(-C\min_{1\leq d\leq D} (\frac{t}{\omega^d \|\mathbb{E}f^{(d)}(\epsilon)\|_{HS}} )^\frac{q}{d}\big),
 \end{align}
 where $C$ is a constant that depends on $D$ and $q$ and $\|\cdot\|_{HS}$ denotes the Hilbert-Schmidt norm (or the Frobenius norm in the case of a matrix).
 \end{lemma}

The second fact is a slight modification of Lemma 2.2.1 in \cite{vanderVaart1996}.
  \begin{lemma}\label{fubinithm}
Let $X$ be a random variable such that for some $q>0$,
\begin{align}\label{tailfubinilem}
    P(|X|>t)\leq c\exp\big(-\big(\frac{t}{\omega}\big)^q\big),\quad \forall t>0,
\end{align}
 where $\omega>0$ and $c$ is constant,
 then we have $\|X\|_{\psi_q}\lesssim \omega$.
\end{lemma}
The proof is straightforward:
 \begin{align*}
    \mathbb{E}\big(\exp(M|X|^q)-1\big)= \, & \mathbb{E}\int^{|X|^q}_{0}M\exp(Mu)\rd u \\ =\, & \int^{\infty}_{0}\mathbb{P}(|X|>u^{\frac{1}{q}})M\exp(Mu)\rd u \\
=\, & \int^{\infty}_{0}cM\exp(Mu-\frac{u}{\omega^q})\rd u\leq \frac{cM}{\omega^{-q}-M}.
\end{align*}
It suffices to take $M^{-1/q}\geq c'\omega$ to have $\mathbb{E} (\exp(M|X|^q)-1)\le 1$. Therefore, $\|X\|_{\psi_q}\lesssim \omega$.

Now, given any $\alpha \in \mathbb{R}^n$, let $f_\alpha(\epsilon)=\langle \alpha,\epsilon \rangle=\sum^n_{i=1}\alpha_i\epsilon_i$. Then (i) $f_\alpha(\epsilon)$  is an 1-degree polynomial of $\epsilon_1,\ldots,\epsilon_n$,  (ii) $\|\mathbb{E}[\nabla f_\alpha(\epsilon)]\|_2=\|\alpha\|_2$ and $\mathbb{E}[\sum^n_{i=1}\alpha_i\epsilon_i]=0$, and (iii) $\epsilon_1,\ldots,\epsilon_n$ are independent and satisfy  $\|\epsilon_i\|_{\psi_q}\leq \omega$. Applying Lemma \ref{lemmagotze} with $D=1, d=1$ yields
 \begin{align}\label{tailbound}
     \mathbb{P}\big(|\sum^n_{i=1}\alpha_i\epsilon_i|>t\big)=\mathbb{P}\big(|f_\alpha(\epsilon)|>t\big)\leq 2 \exp\Big(-C\big(\frac{t}{\omega\|\alpha\|_2}\big)^q\Big),
 \end{align}
 where $C$ is a constant depending on $q$ only.
By Lemma \ref{fubinithm},
\eqref{tailbound} implies $\|\langle \alpha,\epsilon \rangle\|_{\psi_q}\leq C\|\alpha\|_2\omega$, where $C$ is a constant depending on $q$ only. The proof of Lemma \ref{componenttovectorlemma} is complete.

The following lemma is useful for stating the assumptions on effective noises.
\begin{lemma}\label{psinormconvertlemma}

Let $\psi,\varphi$ be any two    nondecreasing nonzero functions  defined on $\mathbb{R}_+$ with $\psi(0)=\varphi(0)=0$ (not necessarily convex). Define
$\varphi^{-1}(t):=\sup\{x\in\mathbb{R}_+:\varphi(x)\leq t\}.
$
and $\psi^{-1}$   similarly.
(i) Suppose that $\psi(\varphi^{-1}(t)/c_{0})$ is concave in $t$ on $\mathbb{R}_+$ for some $c_{0}> \varphi^{-1}(1)/\psi^{-1}(1)$. Then for any random variable $X$,  we have $\|X\|_\psi\leq c_{0}\|X\|_\varphi$.
(ii) Suppose that $\psi(\varphi^{-1}(t)/c_{0})\leq t$ for some $c_{0}>0$, then  $\|X\|_\psi\leq c_{0}\|X\|_\varphi$.
\end{lemma}

We remark that the condition for $c_0$ in part (i) can be replaced by $c_{0}\ge  \varphi^{-1}(1)/\psi^{-1}(1)$ when $\psi$ is continuous at $1$.
For completeness, we provide the proof below.

First, by the definition of $\varphi^{-1}$, $u\leq \varphi^{-1}(\varphi(u))$
 for any $u\geq0$.
Therefore, $X/\|X\|_\varphi\leq \varphi^{-1}\big\{\varphi(X/\|X\|_\varphi)\big\}$, from which it follows that
 \begin{align}\label{lemmabasicineq1}
       \psi\Big(\frac{X}{c\|X\|_{\varphi}}\Big)\leq  \psi\Big[\frac{1}{c}\varphi^{-1}\big\{\varphi(\frac{X}{\|X\|_\varphi})\big\}\Big].
 \end{align}

To prove part (i),
let $f(u)=\psi\big(\varphi^{-1}(u)/c\big)$ with $c>0$ to be determined.
Then $f(u)$ is an increasing function on $\mathbb{R}_+$. (In fact, for any $u\geq u'\geq0$,  $ \varphi^{-1}(u')\leq \varphi^{-1}(u)$, and so
$\psi\{\varphi^{-1}(u')/c\}\leq \psi\{\varphi^{-1}(u)/c\}$.)
From \eqref{lemmabasicineq1}, picking $t=\varphi(X/\|X\|_{\varphi})$ gives
\begin{align}\label{lemmaftineq}
     \psi\Big(\frac{X}{c\|X\|_{\varphi}}\Big)\leq  \psi\Big[\frac{1}{c}\varphi^{-1}\big\{\varphi(\frac{X}{\|X\|_\varphi})\big\}\Big]=f(t).
\end{align}
With $c>\varphi^{-1}(1)/\psi^{-1}(1)$ (or $c\ge \varphi^{-1}(1)/\psi^{-1}(1)$ when $\psi$ is continuous at $1$), we can use   Jensen's inequality to get
 \begin{align}
    \mathbb{E}\Big\{\psi\Big(\frac{X}{c\|X\|_{\varphi}}\Big)\Big\}\leq \mathbb{E}\Big(\psi\Big[\frac{1}{c}\varphi^{-1}\big\{\varphi(\frac{X}{\|X\|_\varphi})\big\}\Big]\Big)=\mathbb{E}\big(f(t)\big)\leq f\big(\mathbb{E} (t)\big)\leq f(1)\leq1.
\end{align}
To prove part (ii), we still set $f(u)=\psi\big(\varphi^{-1}(u)/c\big)$ with $c>0$ and $t=\varphi(X/\|X\|_{\varphi})$. Based on \eqref{lemmabasicineq1},  we get
\begin{align}\label{lemmaineqpartb}
     \psi\big(\frac{X}{c\|X\|_{\varphi}}\big)\leq  f(t)\leq t= \varphi\big(\frac{X}{\|X\|_\varphi}\big).
\end{align}
The proof of Lemma \ref{psinormconvertlemma} is complete.

\subsection{An  Excess Risk Bound}

This part establishes an excess risk bound for location estimation using pivot-blend, which is uniform in scale parameters, shedding light on the impact of asymmetrical scales (skewness) and the presence of an unknown pivotal point. This result is non-asymptotic and non-parametric, making it applicable to various scenarios.

Let  $y_i\in\mathbb{R}$ and $X_i\in \mathbb{R}^p$ which satisfy satisfy $(X_i,y_i)\overset{i.i.d.}{\sim} F^{*}$, where    $F^{*}$ is a  distribution  that  depends on  $\beta^*,m^*,\sigma^*,\nu^*$, where we use superscript $^*$ to denote the statistical truth.  As discussed in Section \ref{subsec:exts}, we consider the practice of estimating the scales beforehand and then minimizing a  criterion over all location parameters $\beta_j (1\le j\le p), m$.  For ease of presentation, given $(\sigma,\nu)$, let $  l_{\sigma,\nu}(\cdot;X_i,y_i)$, imposed on $   \theta:=(\beta,m)$, denote a loss motivated by skewed pivot-blend:
\begin{align}\label{globallossdef}
\begin{split}
    l_{\sigma,\nu}(\theta;X_i,y_i):=\,&\, \rho\big(\frac{r_i-m}{\sigma}+m\big)1_{r_i-m\leq 0}+\rho\big(\frac{r_i-m}{\nu}+m\big)1_{r_i-m>0}\\&+\cc\log\big[\sigma\Phi(m)+\nu\{1-\Phi(m)\}\big] \quad \mbox{with} \quad r_i=y_i-X^T_i\beta,
\end{split}
\end{align}
where the calibration parameter $\cc>0$  and $0\le \Phi(m)\le 1$. Note that  $\Phi$ is not necessarily directly associated with $\rho$.
 Define
$\hat\theta_{\sigma,\nu}$ by   empirical risk minimization
\begin{align}\label{erm}
    \hat\theta_{\sigma,\nu}=\argmin_{\theta }\sum^{n}_{i=1}l_{\sigma,\nu}(\theta;X_i,y_i).
\end{align}
Certainly, opting for different values of $\sigma$ and $\nu$  produces diverse  asymmetric losses and influences the overall risk.  In our analysis, \textit{no} restrictions will be placed on $\sigma$ and $\nu$. The values of $(\sigma,\nu)$ can be specified based on domain knowledge or determined in a data-dependent manner. For notational simplicity,  we sometimes  abbreviate $\hat\theta_{\sigma,\nu}$ as $\hat\theta$ when there is no ambiguity.

 To evaluate the generalization performance of   $\hat\theta$,  let $(X_0,y_0)$ be a new observation that follows $F^{*}$ but is independent of the training data $(X_i,y_i),(1\leq i\leq n)$, and define the population risk of $\hat\theta$ by
\begin{align}\label{riskdef}
    R_{\sigma,\nu}(\hat\theta):=\mathbb{E}_{(X_0,y_0)}\big[l_{\sigma,\nu}(\hat\theta;X_0,y_0)\big],
\end{align}
where the expectation is taken with respect to the new observation $(X_0,y_0)$ only.
Due to the finite number of observations  in estimation, $R_{\sigma,\nu}(\hat\theta_{\sigma,\nu})$ is always greater than the population risk of the ideal $\theta^*_{\sigma,\nu}=\argmin_{\theta} R_{\sigma,\nu}(\theta)$, or $ R_{\sigma,\nu}(\theta^*_{\sigma,\nu})=\inf_{\theta\in \Omega}R_{\sigma,\nu}(\theta)$.
In such a setup, the notion of  \textbf{excess risk} $\mathcal{E}(\hat\theta_{\sigma,\nu};\sigma,\nu)$
is helpful \citep{alma991014707539705251}:
 \begin{align}
  \mbox{Excess Risk:}\quad   \mathcal{E}(\hat\theta_{\sigma,\nu};\sigma,\nu)\,:=\,R_{\sigma,\nu}(\hat\theta_{\sigma,\nu})-R_{\sigma,\nu}(\theta^*_{\sigma,\nu}).
 \end{align}

Our main objective is to establish a non-asymptotic bound for $\mathcal{E}(\hat\theta_{\sigma,\nu};\sigma,\nu)$ regardless of the data distribution for a broad range of $\rho$  that satisfy the following assumption.

$\mathrm{ASSUMPTION}$  $\mathcal{A}$: Assume that the loss $\rho$ satisfies (i) $\rho$ is bounded: $\rho\in [0,\rhob]$ for some $\rhob>0$, and (ii) $\rho$ is regular in the sense that $\rho$ is  piecewise polynomial   on $\kk\geq1$ intervals   $\rho(t)=P_i(t), \, \forall t\in [u_{i-1},u_i),\, i=1,\ldots,\kk,$
where $u_0=-\infty,u_{\kk}=\infty$, and each polynomial function $P_i$ has degree at most $\dd\geq 0$.

Assumption $\mathcal{A}$ encompasses a wide range of practically used loss functions in robust regression and classification, specifically designed to handle extreme outliers.  For example, some loss functions like   $\rho(t)=\int_{0}^{|t|}\psi(s)\rd s$, with a redescending  $\psi$,  such as  Tukey's bisquare $\psi(t)=t\{1-(t/c)^2\}^2$ if $|t|\leq c$, and $0$ otherwise \citep{hampel2011robust}, fit in the category.   Some other $\rho$ functions can be effectively approximated by piecewise polynomial functions.

\begin{theorem}\label{riskbound}
As long as   the loss $\rho$ satisfies Assumption $\mathcal{A}$,
  the   estimator $\hat\theta_{\sigma,\nu}$ defined in \eqref{erm}
 satisfies the following probabilistic bound for all  $\sigma,\nu>0$,
{ \small
 \begin{align*}
  \mathbb{P}&\left\{  \sup_{\sigma,\nu>0}\, \mathcal{E}(\hat\theta_{\sigma,\nu};\sigma,\nu)-C\Bigg[\right.\frac{\rhob\sqrt{p\log\big\{(\kk+1)(\dd+1)\big\}}}{\sqrt{n}}+\frac{(\rhob\lor \cc)\big\{\log\frac{\sigma\lor\nu}{\sigma\land\nu}+|\log(\sigma\land\nu)| \big\}}{\sqrt{n}}\times\nonumber
\\
&\quad\quad\,\,\,\,\,\bigg\{ \sqrt{\log\Big(\log\frac{\sigma\lor\nu}{\sigma\land\nu}\lor|\log(\sigma\land\nu)|\lor1\Big)} + \sqrt{\log\frac{1}{\epsilon}}\bigg\}\left.\Bigg]\leq 0\right\}\geq 1-\epsilon.
\end{align*}}
\normalsize
\end{theorem}

The theorem provides a bound for  the excess risk $\mathcal{E}(\hat\theta_{\sigma,\nu};\sigma,\nu)$ that holds   uniformly in $\sigma$ and $\nu$ with probability at least $1-\epsilon$, as characterized by the following rate
\begin{align}\label{riskboundorderterm}
   &\frac{\rhob\sqrt{p\log\big\{(\kk+1)(\dd+1)\big\}}}{\sqrt{n}}+\frac{\rhob\lor\cc\big\{\log\frac{\sigma\lor\nu}{\sigma\land\nu}+|\log(\sigma\land\nu)| \big\}}{\sqrt{n}}\times\bigg\{\sqrt{\log\frac{1}{\epsilon}}+\nonumber\\&\quad\sqrt{\log\Big(\log\frac{\sigma\lor\nu}{\sigma\land\nu}\lor|\log(\sigma\land\nu)|\lor1\Big)}\bigg\}.
\end{align}
The first term in \eqref{riskboundorderterm} illustrates the influence of loss complexity   and problem dimensions on the excess risk. The second term, which incorporates both $\sigma$ and $\nu$, arises due to the  pivot estimation. Clearly, when there is no skewness, $
    \sigma=\nu  \Rightarrow  \log\frac{\sigma\lor\nu}{\sigma\land\nu}=0
$ and the rate becomes
$$
   \frac{\rhob\sqrt{p\log\big\{(\kk+1)(\dd+1)\big\}}}{\sqrt{n}}+\frac{(\rhob+\cc) |\log(\sigma\land\nu)|  }{\sqrt{n}}\times\bigg\{  \sqrt{\log (1+|\log(\sigma\land\nu)| )}+\sqrt{\log\frac{1}{\epsilon}}\bigg\}.
$$
 In the more practical scenario of unequal scales, the risk for location estimation can significantly increase, and the provided bound quantitatively characterizes how skewness inflates the risk non-asymptotically.

To prove the theorem, we first introduce a basic excess-risk bound  for fixed $\sigma,\nu>0$.
\begin{lemma}\label{fixsigmanubound}
Suppose that the loss $\rho$ satisfies Assumption $\mathcal{A}$.  Fixing the values of $\sigma,\nu>0$ in \eqref{globallossdef},  the corresponding estimator $\hat\theta_{\sigma,\nu}$ from \eqref{erm} satisfies
\begin{align}\label{fixriskbound}
     \mathcal{E}(\hat\theta_{\sigma,\nu};\sigma,\nu)\lesssim  \rhob\sqrt{\frac{p\log\big\{(\kk+1)(\dd+1)\big\}}{n}}+\frac{\rhob\lor\cc\big(|\log\sigma\wedge \nu|+\log\frac{\nu\vee\sigma}{\sigma\wedge \nu}\big)}{\sqrt{n}}\sqrt{\log\frac{1}{\epsilon}},
\end{align}
with probability
 at least $1-\epsilon$.

\end{lemma}

\begin{proof}
 Let $\{X,y\}$ denote the training data, i.e., $\{X,y\}:=\big\{(X_i,y_i),1\leq i\leq n\big\}$ with $(X_i,y_i)\overset{i.i.d.}{\sim}F^{*}$.
Given $\sigma,\nu>0$, define a function class consisting of all $l_{\sigma,\nu}(\theta;\cdot)$   $\theta\in \Omega=\mathbb R^p\times \mathbb R$
  \begin{align}
    \mathcal{L}_{\sigma,\nu}(\Omega):=\big\{l_{\sigma,\nu}(\theta;\cdot):\theta\in \Omega\big\}.
  \end{align}
 For  simplicity, we often use the shorthand notations
 $R(\cdot)$, $R^{(n)}(\cdot)$, and $l(\theta;\cdot)$ to denote
 $R_{\sigma,\nu}(\cdot)$, $(1/n)\sum^{n}_{i=1}l_{\sigma,\nu}(\cdot;X_i,y_i)$, and $l_{\sigma,\nu}(\theta;\cdot)$, respectively, when there is no ambiguity.

 First,  the standard bound for excess risk through uniform laws yields\small
\begin{align}\label{symmetrizationbound}
    R(\hat\theta_{\sigma,\nu})-R(\theta^*_{\sigma,\nu})\leq &\, R(\hat\theta_{\sigma,\nu})-R^{(n)}(\hat\theta_{\sigma,\nu})+R^{(n)}(\hat\theta_{\sigma,\nu})-R^{(n)}(\theta^*_{\sigma,\nu})+R^{(n)}(\theta^*_{\sigma,\nu})-R(\theta^*_{\sigma,\nu})\nonumber\\
    \leq &\,
    2\sup_{\theta\in\Omega}\big|R(\theta)-R^{(n)}(\theta)\big|  =2\sup_{l\in\mathcal{L}_{\sigma,\nu}(\Omega)}\big|\mathbb{P}_nl-\mathbb{P}l\big|=2\big\|\mathbb{P}_n-\mathbb{P}\big\|_{\mathcal{L}_{\sigma,\nu}(\Omega)},
\end{align}
\normalsize
where $\mathbb{P}$ is the distribution $F^*$ and $\mathbb{P}_n$ is the empirical measure that places  probability mass $1/n$ on each $(X_i,y_i),1\leq i\leq n$.

Let
\begin{align}
    g_{\sigma,\nu}(m):=\cc\log\big[\sigma\Phi(m)+\nu\{1-\Phi(m)\}\big].
\end{align}
Without loss of generality, we assume that $\sigma\leq \nu$, that is, $\sigma\land\nu=\sigma,\, \sigma\lor\nu=\nu$, then
\begin{align}
    g_{\sigma,\nu}(m)=\cc\log\sigma+\cc\log\big[1+(\frac{\nu}{\sigma}-1)\{1-\Phi(m)\}\big].
\end{align}
 Because $\log(1+(\frac{\nu}{\sigma}-1)t)$ is an increasing function for $t\geq0$, we know
\begin{align}
  \cc\log\sigma\leq   g_{\sigma,\nu}(m)\leq \cc\log\sigma+\cc\log\frac{\nu}{\sigma},
\end{align}
from which it follows that
\begin{align}\label{gbound}
  |g_{\sigma,\nu}(m)|\leq \big|\cc\log\sigma+\cc\log\frac{\nu}{\sigma}\big|\lor |\cc\log\sigma|\leq |\cc\log\sigma|+\cc\log\frac{\nu}{\sigma}.
\end{align}
By Assumption $\mathcal{A}$ and \eqref{gbound},
\begin{align}\label{lbounded}
  |l_{\sigma,\nu}|\leq|\rho|+|g_{\sigma,\nu}|\leq \rhob+\cc\big(|\log\sigma|+\log\frac{\nu}{\sigma}\big).
\end{align}

Due to \eqref{lbounded} and $(X_i,y_i)\overset{i.i.d.}{\sim}F^{*}$,
applying   McDiarmid's inequality and symmetrization in  empirical process theory  yields a data-dependent bound with probability at least $1-\epsilon$,
\begin{align}\label{boundeddifferencebound}
    \big\|\mathbb{P}_n-\mathbb{P}\big\|_{\mathcal{L}_{\sigma,\nu}(\Omega)}\leq 2\mathcal{R}_{X,y}(\mathcal{L}_{\sigma,\nu}(\Omega))+6\Big\{\rhob+\cc\big(|\log\sigma|+\log\frac{\nu}{\sigma}\big)\Big\}\sqrt{\frac{\log(3/\epsilon)}{2n}},
\end{align}
where $\mathcal{R}_{X,y}(\mathcal{L}_{\sigma,\nu}(\Omega))$ is the empirical Rademacher complexity of $\mathcal{L}_{\sigma,\nu}(\Omega)$ with respect to the training data $\{X,y\}$:
\begin{align}\label{originalrademacher}
   \mathcal{R}_{X,y}(\mathcal{L}_{\sigma,\nu}(\Omega)) :=\frac{1}{n}\mathbb{E}_\epsilon\sup_{l\in\mathcal{L}_{\sigma,\nu}(\Omega)}\sum^n_{i=1}\epsilon_i l(\theta;X_i,y_i),
\end{align}
and $\epsilon_i$'s are i.i.d. Rademacher random variables;
see, e.g., Theorem 3.4.5 in \cite{gine2021mathematical}. Note that the  expectation in \eqref{originalrademacher} is taken with respect to $\epsilon$ only, and \eqref{originalrademacher} depends on $\sigma,\nu$  through  the function class $\mathcal{L}_{\sigma,\nu}(\Omega)$.

It remains to bound the empirical Rademacher complexity.
Toward this, denote by  $\nd_\sigma,\nd_\nu$  two augmented design matrices
\begin{align}
    \nd_\sigma=-\frac{1}{\sigma}[X, (1-\sigma)1_n],\quad \nd_\nu=-\frac{1}{\nu}[X, (1-\nu)1_n],
\end{align}
where $1_n$ is a column vector of $n$ ones.
Let
\begin{align}
    \alpha_\sigma=\frac{1}{\sigma}y,\quad  \alpha_\nu=\frac{1}{\nu}y.
\end{align}
Suppose that  $\mathrm{rank}(\nd_\sigma)\leq r$ and $\mathrm{rank}(\nd_\nu)\leq r$.  By the singular value decomposition,
\begin{align}
    \nd_\sigma=U_\sigma D_{\sigma}V^T_\sigma,\quad\nd_\nu=U_{\nu}D_\nu V^T_\nu,
\end{align}
 where $U_\sigma,U_\nu$ are orthogonal matrices with   $r$ columns: $U^T_{\sigma}U_{\sigma}=I_{r\times r},U^T_{\nu}U_{\nu}=I_{r\times r}$.
Define
\begin{align}\label{Udef}
    \bar{U}_{\sigma}=[U_{\sigma},\alpha_\sigma],\quad \bar{U}_{\nu}=[U_{\nu},\alpha_\nu].
\end{align}
Now, by the sub-additivity of sup and \eqref{Udef},
\begin{align}\label{newradecomplex bound}
       &\,\mathcal{R}_{X,y}(\mathcal{L}_{\sigma,\nu}(\Omega)) \nonumber\\\leq &\, \frac{1}{n}\mathbb{E}_\epsilon\sup_{\theta\in\Omega}\langle\epsilon, \rho(\nd_\sigma\theta+\alpha_\sigma)\rangle+\frac{1}{n}\mathbb{E}_\epsilon\sup_{\theta\in\Omega}\langle\epsilon, \rho(\nd_\nu\theta+\alpha_\nu)\rangle+ \frac{1}{n}\mathbb{E}_\epsilon\sup_{m\in\mathbb{R}}\langle\epsilon, g_{\sigma,\nu}(1m)\rangle\nonumber\\\leq &\, \frac{1}{n}\mathbb{E}_\epsilon\sup_{\xi\in\mathbb{R}^{r+1}}\langle\epsilon, \rho(\bar{U}_{\sigma}\xi)\rangle+\frac{1}{n}\mathbb{E}_\epsilon\sup_{\xi\in\mathbb{R}^{r+1}}\langle\epsilon, \rho(\bar{U}_\nu\xi)\rangle+ \frac{1}{n}\mathbb{E}_\epsilon\sup_{m\in\mathbb{R}}\langle\epsilon, g_{\sigma,\nu}(1m)\rangle.
\end{align}

To  bound the first term in \eqref{newradecomplex bound},
 let
\begin{align}
    Z(X,y,\sigma)=\bar{U}_{\sigma}=[Z_1,\cdots,Z_n]^T,
\end{align}
and given $\rho$, define a class of functions
\begin{align}
    \mathcal{F}:=\big\{\rho(\langle\xi,\cdot\rangle):\xi\in\mathbb{R}^{r+1}\big\},
\end{align}
which does not depend on  the two scale parameters.
Given $Z(X,y,\sigma)$,  let $\mathbb{Q}_n$ be the empirical measure determined by $Z_i$, $1\leq i\leq n$,
 and
\begin{align}
    \|f-\tilde{f}\|^2_{\mathbb{Q}_n}:=\frac{1}{n}\sum^n_{i=1}\big\{f(Z_i)-\tilde{f}(Z_{ i})\big\}^2,\quad \forall f,\tilde{f}\in \mathcal{F}.
\end{align}
Then
\begin{align}\label{rademacherequivalence}
    \frac{1}{n}\mathbb{E}_\epsilon\sup_{\xi\in\mathbb{R}^{r+1}}\langle\epsilon, \rho(\bar{U}_{\sigma}\xi)\rangle=\frac{1}{\sqrt{n}}\mathbb{E}_\epsilon\sup_{f\in\mathcal{F}}\frac{1}{\sqrt{n}}\sum^n_{i=1}\epsilon_i f(Z_i)\leq \frac{C}{\sqrt{n}}\int^{\rhob}_{0}\sqrt{\log  \mathcal{N}(\varepsilon,\mathcal{F},\|\cdot\|_{\mathbb{Q}_n})}\rd \varepsilon,
\end{align}
where the last inequality is due to Dudley's integral bound. Because $|f|\leq \rhob,\, \forall f\in \mathcal{F}$, by Theorem 2.6.7 in \cite{vanderVaart1996} we know
\begin{align}\label{coveringnumberf}
    \mathcal{N}(\varepsilon,\mathcal{F},\|\cdot\|_{\mathbb{Q}_n})\leq C\vc(\mathcal{F})\big(\frac{c\rhob}{\varepsilon}\big)^{2\vc(\mathcal{F})},
\end{align}
where $\vc(\mathcal{F})$  denotes the VC-dimension  of  $\mathcal{F}$  defined through the notion of subgraph (cf. Definition 3.6.8 in \cite{gine2021mathematical}).
In more details,
$\vc(\mathcal{F})$ is the VC-dimension of the following $\{0,1\}$-valued function class defined on $\mathbb{R}^{r+1}\times \mathbb{R}$:
\begin{align}
    \mathcal{H}:=\big\{h_{\xi}(z,t)=1_{\rho(\langle\xi,z\rangle)> t}=\sign(\rho(\langle\xi,z\rangle)- t):\xi\in\mathbb{R}^{r+1}\big\}.
\end{align}
Here,  the sign function is defined by $\sign(a)=1$ if $a>0$, and $0$ otherwise, and recall that
 $1_{\rho(\langle\xi,z\rangle)> t}$ is the indicator function of the set $\big\{(z,t):\rho(\langle\xi,z\rangle)> t,z\in\mathbb{R}^{r+1},t\in\mathbb{R}\big\}$ or the subgraph of $\rho(\langle\xi,\cdot\rangle)$ for each given $\xi$.

To bound $\vc(\mathcal{H})$,
 we introduce
a more general function class $\tilde{\mathcal{H}}$
\begin{align}\label{allsubdef}
 \tilde{\mathcal{H}}:=\big\{\tilde{h}_{\xi,\omega_1,\omega_2}(z,t)=\sign\big(\omega_1\rho(\langle\xi,z\rangle)+\omega_2 t\big):\xi\in\mathbb{R}^{r+1},\omega_1,\omega_2\in\mathbb{R}\big\},
\end{align}
 where $\tilde{h}_{\xi,\omega_1,\omega_2}(z,t)$  is defined on $(z,t)\in\mathbb{R}^{r+1}\times \mathbb{R}$ and has two additional parameters $\omega_1$ and $\omega_2$ than $h_{\xi}(z,t)$ in the definition of $\mathcal{H}$.
  Since $\mathcal{H}\subset   \tilde{\mathcal{H}}$ (by fixing $\omega_1=1,\omega_2=-1$),
\begin{align}\label{networkvcdim}
     \vc(\mathcal{F})= \vc(\mathcal{H})\leq \vc( \tilde{\mathcal{H}}).
\end{align}
 $\tilde{\mathcal{H}}$ corresponds to the set of functions computed by a
   neural network $\mathbf{N}$ shown in Figure \ref{networkfig}.
   It has two computation units, one in the   hidden layer with activation function $\rho$, and the other  in the output layer applying a sign operation on the combined inputs  $\omega_1\rho(\langle\xi,z\rangle)+\omega_2 t$.
Lemma \ref{networkthm} is Theorem 10 in \cite{bartlett2019nearly} and can be  proved based on Theorem 2.2 of  \cite{goldberg1995bounding}.

 \begin{figure}[h!]
         \centering
         \includegraphics[width=0.55\textwidth]{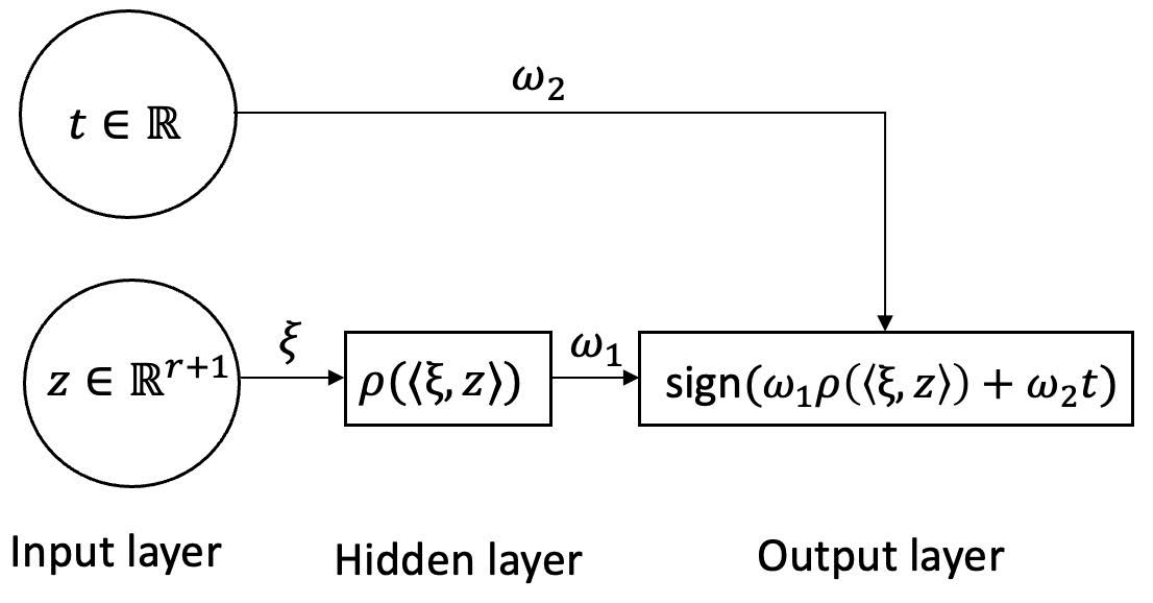}
        \caption{\small Architecture of the  network $\mathbf{N}$ that computes $\tilde{h}_{\xi,\omega_1,\omega_2}(z,t)\in \tilde{\mathcal{H}}$.\normalsize}
        \label{networkfig}
\end{figure}

 \begin{lemma}\label{networkthm}
 Suppose that a neural network $\mathbf{N}_0$  satisfies (i) $\mathbf{N}_0$ has a directed acyclic graph, that is, the connections from input or computation units to computation units  do not form any loops, (ii)
the unique output unit is the only computation unit in the output layer (layer $J$), where $J\geq 2$ denotes the length of the longest path in the graph of $\mathbf{N}_0$, and the activation function of the output unit is a sign function that  takes inputs from units in any layer $j<J$, including the input layer (layer $0$), and
(iii) within each computation unit except the output unit, $\Psi$ is the  activation function  and  is piecewise polynomial  on $I\geq 1$ intervals $\Psi(t)=P_i(t),\,\forall t\in[u_{i-1},u_{i}),i=1,\ldots,I$, where $u_0=-\infty,u_I=\infty$, and each polynomial function $P_i$ has degree at most $M\geq0$. Let  $W\geq 2$  be the number of parameters (weights and biases) and $U\geq 2$ be the number of computation units.
   Let $\mathcal{G}$  denote  the set of $\{0,1\}$-valued functions computed by the  network $\mathbf{N}_0$, then $\vc(\mathcal{G})\leq 2W\log_2\big[16e\{M^{U-1}+\sum^{U-1}_{i=0}M^i\}(1+I)^{U}\big]$.
 \end{lemma}

 Our network $\mathbf{N}$ as shown in Figure \ref{networkfig}  is a feed-forward neural network and has a directed acyclic graph structure. The  output unit of $\mathbf{N}$ in the output layer  takes the input $\rho(\langle\xi,z\rangle)$ from the computation unit  in the hidden layer  and the input unit  $t\in \mathbb{R}$ in the input layer. Under Assumption $\mathcal{A}$,
our network $\mathbf{N}$ satisfies the assumptions in  Lemma \ref{networkthm} and has $r+3$ parameters (no bias parameters) and two computation units. Using Lemma \ref{networkthm}, we get
\begin{align}\label{networkvcbound}
 \vc(\tilde{\mathcal{H}})\leq C(r+3)\log\big\{e(2 D+1)(1+K)^2\big\}\lesssim r\log\big\{(\kk+1)(\dd+1)\big\}.
\end{align}

Now, based on \eqref{networkvcdim} and \eqref{networkvcbound},
\begin{align}\label{pseudobound}
   \vc(\mathcal{F})\leq  Cr\log\big\{(\kk+1)(\dd+1)\big\},
\end{align}
and combining \eqref{rademacherequivalence}, \eqref{coveringnumberf}, and \eqref{pseudobound} results in
\begin{align}\label{raderho1}
   \frac{1}{n}\mathbb{E}_\epsilon\sup_{\xi\in\mathbb{R}^{r+1}}\langle\epsilon, \rho(\bar{U}_{\sigma}\xi)\rangle\leq&\,\frac{C}{\sqrt{n}}\int^{\rhob}_{0}\sqrt{\vc(\mathcal{F})\log\frac{\rhob}{\varepsilon}}\rd \varepsilon\lesssim \rhob\sqrt{\frac{r\log\big\{(\kk+1)(\dd+1)\big\}}{n}}.
\end{align}
Similarly, the second term in \eqref{newradecomplex bound} satisfies
\begin{align}\label{raderho2}
 \frac{1}{n}\mathbb{E}_\epsilon\sup_{\xi\in\mathbb{R}^{r+1}}\langle\epsilon, \rho(\bar{U}_{\nu}\xi)\rangle \lesssim \rhob\sqrt{\frac{r\log\big\{(\kk+1)(\dd+1)\big\}}{n}}.
\end{align}
Note that  \eqref{raderho1}  and  \eqref{raderho2} hold for all  $\sigma,\nu>0$.

Finally, we bound the third  term on the right-hand side of  \eqref{newradecomplex bound}. Let
\begin{align}
   q_{\sigma,\nu}(m):=\cc\log\big[1+(\frac{\nu}{\sigma}-1)\{1-\Phi(m)\}\big],
\end{align}
and so $0\leq q_{\sigma,\nu}(m)\leq \cc\log\frac{\nu}{\sigma}$.  Again, by the sub-additivity of sup and Dudley's integral bound,
\begin{align}\label{gconversion}
    \frac{1}{n}\mathbb{E}_\epsilon\sup_{m\in\mathbb{R}}\langle\epsilon, g_{\sigma,\nu}(1m)\rangle\leq&\,\frac{1}{n}\cc\log\sigma\mathbb{E}_\epsilon\sup_{m\in\mathbb{R}}\langle\epsilon, 1\rangle+\frac{1}{n}\mathbb{E}_\epsilon\sup_{m\in\mathbb{R}}\langle \epsilon, q_{\sigma,\nu}(1m)\rangle\nonumber\\\leq &\,\frac{C}{n} \int^{\frac{1}{2}\cc\log\frac{\nu}{\sigma}\sqrt{n}}_{0}\sqrt{\log\mathcal{N}\big(\varepsilon,\big\{[q_{\sigma,\nu}(1m)]: m\in \mathbb{R}\big\},\|\cdot\|_2\big)}\rd \varepsilon\nonumber\\\leq &\,\frac{C}{n} \int^{\frac{1}{2}\cc\log\frac{\nu}{\sigma}\sqrt{n}}_{0}\sqrt{\log\mathcal{N}\Big(\frac{\varepsilon}{\sqrt{n}},\big[0,\cc\log\frac{\nu}{\sigma}\big],|\cdot|\Big)}\rd \varepsilon\nonumber\\\lesssim &\,  \frac{\cc\log\frac{\nu}{\sigma}}{\sqrt{n}}.
\end{align}
Plugging \eqref{raderho1}, \eqref{raderho2}, and \eqref{gconversion} into \eqref{newradecomplex bound} yields
\begin{align}\label{rcomplexbound}
  \mathcal{R}_{X,y}(\mathcal{L}_{\sigma,\nu}(\Omega)) \lesssim \rhob\sqrt{\frac{r\log\big\{(\kk+1)(\dd+1)\big\}}{n}}+\frac{\cc\log\frac{\nu}{\sigma}}{\sqrt{n}}.
\end{align}
Summarizing  \eqref{symmetrizationbound}, \eqref{boundeddifferencebound}, and \eqref{rcomplexbound}, we have   the following  bound with probability at least $1-\epsilon$ for any $0<\epsilon<1$,
\begin{align}\label{partabound}
   &\, \mathcal{E}(\hat\theta_{\sigma,\nu};\sigma,\nu)\nonumber\\\leq &\,
    C\Bigg[\rhob\sqrt{\frac{r\log\big\{(\kk+1)(\dd+1)\big\}}{n}}+\frac{\cc\log\frac{\nu}{\sigma}}{\sqrt{n}}+\frac{\rhob+\cc\big(\log\frac{\nu}{\sigma}+|\log\sigma|\big)}{\sqrt{n}}\sqrt{\log\frac{1}{\epsilon}}\,\Bigg]\nonumber\\\lesssim&\, \rhob\sqrt{\frac{p\log\big\{(\kk+1)(\dd+1)\big\}}{n}}+\frac{\rhob\lor\cc\big(\log\frac{\nu}{\sigma}+|\log\sigma|\big)}{\sqrt{n}}\sqrt{\log\frac{1}{\epsilon}},
\end{align}
\normalsize
where the last inequality is due to $r\leq p+1\lesssim p$. The proof of Lemma \ref{fixsigmanubound} is complete.
\end{proof}

Lemma \ref{fixsigmanubound}  shows the effect of skewness for fixed values of  $\sigma,\nu$. For example, when $\sigma=\sigma^*,\nu=\nu^*$,  the excess risk $\mathcal{E}(\hat\theta_{\sigma^*,\nu^*};\sigma^*,\nu^*)$ satisfies
\begin{align}
   \mathcal{E}(\hat\theta_{\sigma^*,\nu^*};\sigma^*,\nu^*)\lesssim \frac{\rhob\sqrt{p\log\big\{(\kk+1)(\dd+1)\big\}}}{\sqrt{n}}+\frac{\rhob\lor\cc\big\{\log\frac{\sigma^*\lor\nu^*}{\sigma^*\land\nu^*}+|\log(\sigma^*\land\nu^*)| \big\}}{\sqrt{n}}\sqrt{\log\frac{1}{\epsilon}}
\end{align}
with probability at least $1-\epsilon$. Our objective is to establish a uniform law that applies to all scale parameters   $\sigma,\nu>0$. Toward this,    let
\begin{align}\label{taudef}
    \tau=|\log(\sigma\land\nu)|+\log\frac{\sigma\lor\nu}{\sigma\land\nu},
\end{align}
which is nonnegative.
By Lemma \ref{fixsigmanubound},  there exists a universal constant $C$ such that the  event $\ev(\tau,t)$
\begin{align}
 \ev(\tau,t) := \left\{ \mathcal{E}(\hat\theta_{\sigma,\nu};\sigma,\nu)- C\Bigg[\rhob\sqrt{\frac{p\log\big\{(\kk+1)(\dd+1)\big\}}{n}}+\big(\rhob+\cc\tau\big)t\Bigg]\geq0\right\} \end{align}
 occurs with probability
\begin{align}\label{eventprobsimple}
    \mathbb{P}\big[\ev(\tau,t)\big]\leq \exp(-2nt^2),
\end{align}
for any $t>0$.

Let $[0,+\infty)=\bigcup\limits_{l=0}^{\infty}[\tau_l,\tau_{l+1}]$ with $0=\tau_0<\tau_1<\ldots<+\infty$ to be determined. Let $Q(\tau)\geq 0$ ($\forall \tau\geq0$) be an increasing function.  Then, for all $l\geq0$ and $l\in\mathbb{Z}$, $\tau\in [\tau_l,\tau_{l+1}]$  implies  $\tau\geq \tau_{l}$, $Q(\tau)\geq Q(\tau_l)$, and $\tau Q(\tau)\geq\tau Q(\tau_l)\geq \tau_l Q(\tau_l)$.
Based on \eqref{eventprobsimple} and  the union bound,
\small
\begin{align}\label{discretebound}
    &\mathbb{P}\left(\sup_{\sigma,\nu>0}\mathcal{E}(\hat\theta_{\sigma,\nu};\sigma,\nu)- C\Bigg[\rhob\sqrt{\frac{p\log\big\{(\kk+1)(\dd+1)\big\}}{n}}+\big(\rhob+\cc\tau\big)\big\{t+Q(\tau)\big\}\Bigg]\geq 0\right)\nonumber\\\leq& \sum^{\infty}_{l=0}\mathbb{P}\Big[\ev\big(\tau_l,t+Q(\tau_l)\big)\Big].
\end{align}
\normalsize

Take $Q(\tau)=\sqrt{\log(1+\tau)/n}$ and $\tau_l=l$ for all $l\geq0$ and $l\in\mathbb{Z}$, then \eqref{discretebound} is bounded by
\begin{align}\label{discretizationbound}
\sum^{\infty}_{l=0}\mathbb{P}\Bigg[\ev\bigg(l,t+\sqrt{\frac{\log(1+l)}{n}}\bigg)\Bigg]\leq\sum^{\infty}_{l=0}\frac{1}{(l+1)^2}\exp(-2nt^2) \leq &C\exp(-2nt^2).
\end{align}
Plugging \eqref{taudef} into \eqref{discretizationbound} gives the following uniform law   for all $\sigma,\nu>0$ with probability $1-\epsilon$
\begin{align}
      \mathcal{E}(\hat\theta_{\sigma,\nu};\sigma,\nu)\leq &\,C\Bigg[\frac{\rhob\sqrt{p\log\big\{(\kk+1)(\dd+1)\big\}}}{\sqrt{n}}+\frac{\rhob\lor\cc\big\{\log\frac{\sigma\lor\nu}{\sigma\land\nu}+|\log(\sigma\land\nu)| \big\}}{\sqrt{n}}\nonumber\\&\quad\,\times\bigg\{\sqrt{\log\frac{1}{\epsilon}}+\sqrt{\log\big(\log\frac{\sigma\lor\nu}{\sigma\land\nu}\lor|\log(\sigma\land\nu)|\lor1)+\log 2}\bigg\}\Bigg]\nonumber\\\lesssim &\,\frac{\rhob\sqrt{p\log\big\{(\kk+1)(\dd+1)\big\}}}{\sqrt{n}}+\frac{\rhob\lor\cc\big\{\log\frac{\sigma\lor\nu}{\sigma\land\nu}+|\log(\sigma\land\nu)| \big\}}{\sqrt{n}}\nonumber\\&\quad\,\times\bigg\{\sqrt{\log\frac{1}{\epsilon}}+\sqrt{\log\big(\log\frac{\sigma\lor\nu}{\sigma\land\nu}\lor|\log(\sigma\land\nu)|\lor1)}\bigg\}, \quad\forall\sigma,\nu>0,
\end{align}
The proof of Theorem \ref{riskbound} is complete.

\subsection{Proof of Theorem \ref{globalgroupl1}}
\label{appsub:proofofth1}

By the  optimality of $\hat{\mu}$: $l(\hat{\mu})+\frac{\etaa}{2}\|\hat\gamma\|^2_2+\lambda\mtbb\|\hat{\bar{\beta}}\|_{2,1}\leq l(\mu^*)+\frac{\etaa}{2}\|\gamma^*\|^2_2+\lambda\mtbb\|\bar{\beta}^*\|_{2,1}$,  we obtain a basic inequality as follows
\begin{align}\label{basicineql2}
    &\,\mathbf{\Delta}_{l}(\hat{\mu},\mu^*)+\frac{\etaa}{2}\|\hat{\gamma}-\gamma^*\|^2_2\nonumber\\\leq &\, \big\langle\epsilon_{\bar{\eta}},\bar{X}\hat{\bar{\beta}}-\bar{X}\bar{\beta}^*\big\rangle+\langle\epsilon_{\bar m},\hat {\bar m}-\bar m\rangle+\langle\epsilon_{\varsigma},\hat {\varsigma}-\varsigma\rangle+\langle-\etaa\gamma^*,\hat{\gamma}-\gamma^*\rangle\nonumber\\&\,+\lambda\mtbb\|\bar{\beta}^*\|_{2,1}-\lambda\mtbb\|\hat{\bar{\beta}}\|_{2,1}\nonumber\\\leq &\,\big\langle\epsilon_{\bar{\eta}},\bar{X}\hat{\bar{\beta}}-\bar{X}\bar{\beta}^*\big\rangle+\langle\epsilon_{\bar m},\hat {\bar m}-\bar m\rangle+\langle\epsilon_{\varsigma},\hat {\varsigma}-\varsigma\rangle+\frac{b}{2}\|\hat\gamma-\gamma^*\|^2_2+\frac{\etaa^2}{2b}\|\gamma^*\|_{2}^2\nonumber\\&\,+\lambda\mtbb\|\bar{\beta}^*\|_{2,1}-\lambda\mtbb\|\hat{\bar{\beta}}\|_{2,1}
\end{align}
for any $b>0$.

First, we  bound $\big\langle\epsilon_{\bar{\eta}},\bar{X}\hat{\bar{\beta}}-\bar{X}\bar{\beta}^*\big\rangle$. Let $\bar{X}_k$ denote the $k$th column of $\bar{X}$    and  $\Xi_k:=[\bar{X}_k,\bar{X}_{p+k}]^T$ and  recall $ {\bar{\beta}}_k = [\beta_k, b_k]^T$,   $k=1,\ldots,p$.  By H\"{o}lder's inequality,
\begin{align}
  \langle\epsilon_{\bar{\eta}},\bar{X}\hat{\bar{\beta}}-\bar{X}\bar{\beta}^*\rangle=  \sum^p_{k=1}\langle
    \Xi_k\epsilon_{\bar{\eta}},\hat{\bar{\beta}}_k-\bar{\beta}^*_k\rangle
     \leq  \sum^p_{k=1} \|\Xi_k\epsilon_{\bar{\eta}}\|_2\|\hat{\bar{\beta}}_k-\bar{\beta}^*_k\|_2
     \leq  \sqrt{2}\big\|\bar X^T\epsilon_{\bar\eta}\big\|_{\infty} \|\hat{\bar{\beta}}-\bar{\beta}^*\|_{2,1}.
\end{align}
Let $\lambda=A\big\|\bar X^T\epsilon_{\bar\eta}\big\|_{\infty}/\mtbb$ for $ A>1$. Then we have
\begin{align}
    &\,\big\|\bar X^T\epsilon_{\bar\eta}\big\|_{\infty}\big\|\hat{\bar{\beta}}-\bar{\beta}^*\big\|_{2,1}+\lambda\mtbb\big\|\bar\beta^*\big\|_{2,1}-\lambda\mtbb\big\|\hat{\bar\beta}\big\|_{2,1}\nonumber\\\leq &\,\big(\lambda\mtbb+\big\|\bar X^T\epsilon_{\bar\eta}\big\|_{\infty}\big)\big\|\bar\beta^*\big\|_{2,1}+\big(\big\|\bar X^T\epsilon_{\bar\eta}\big\|_{\infty}-\lambda\mtbb\big)\big\|\hat{\bar\beta}\big\|_{2,1}\nonumber\\= &\,(1+A)\big\|\bar X^T\epsilon_{\bar\eta}\big\|_{\infty}\big\|\bar\beta^*\big\|_{2,1}+(1-A)\big\|\bar X^T\epsilon_{\bar\eta}\big\|_{\infty}\big\|\hat{\bar\beta}\big\|_{2,1}\nonumber\\\leq &\,(1+A)\big\|\bar X^T\epsilon_{\bar\eta}\big\|_{\infty}\big\|\bar\beta^*\big\|_{2,1}.
\end{align}

By the  assumption on $\epsilon_{\bar{\eta}}$, the   $\psi$-norm of $| \bar{X}^T_k\epsilon_{\bar{\eta}}|$ is bounded by
\begin{align}
    \big\| \bar{X}^T_k\epsilon_{\bar{\eta}}\big\|_{\psi}=\|\langle\epsilon_{\bar{\eta}},\bar{X}_k\rangle\|_{\psi}\leq \|\epsilon_{\bar{\eta}}\|_{\psi}\|\bar{X}_k\|_2\leq\omega_{\bar\eta}\mtbb.
\end{align}
The following inequality is essentially  Massart's finite class lemma, adapted for our purpose
\begin{align}
    \big\|\max_{1\leq k\leq 2p} |\bar{X}^T_k\epsilon_{\bar{\eta}}|\big\|_{\psi}\leq   2 (1\vee c_1   \vee 2 \psi(c_2))^2 c_0    \psi^{-1}(p)\omega_{\bar\eta}\mtbb .
\end{align}
It can  be obtained by modifying  the proof of  Lemma 2.2.2 in \cite{vanderVaart1996} (details omitted).  By Lemma \ref{psinormconvertlemma},
\begin{align}
   \mathbb{E}\big[ \|\bar X^T\epsilon_{\bar\eta} \|_{\infty}\big]\le \left\| \|\bar X^T\epsilon_{\bar\eta} \|_{\infty}\right\|_{L_2} \le \psi^{-1}(1) \Big\|\big\|\bar X^T\epsilon_{\bar\eta}\big\|_{\infty}\Big\|_{\psi}\leq C c_\psi\psi^{-1}(  p)\omega_{\bar\eta}\mtbb,
\end{align}
where $c_\psi =    (1\vee c_1   \vee 2 \psi(c_2))^2 c_0 \psi^{-1}(1)$.
Therefore,
\begin{align}\label{derive1}
 \mathbb{E}\Big\{ \big\langle\epsilon_{\bar{\eta}},\bar{X}\hat{\bar{\beta}}-\bar{X}\bar{\beta}^*\big\rangle+\lambda\mtbb\big\|\bar\beta^*\big\|_{2,1}-\lambda\mtbb\big\|\hat{\bar\beta}\big\|_{2,1}\Big\}\nonumber\leq &\, (1+A)\mathbb{E}\Big[\big\|\bar X^T\epsilon_{\bar\eta}\big\|_{\infty}\Big]\big\|\bar\beta^*\big\|_{2,1}\nonumber\\\leq&\, C(1+A) c_\psi \psi^{-1}(  p)\omega_{\bar\eta}\mtbb\big\|\bar\beta^*\big\|_{2,1}.
\end{align}

Next, we bound the stochastic term $\langle\epsilon_{\bar m},\hat {\bar m}-\bar m^*\rangle$. For any $a>0$,
\begin{align}
   \langle\epsilon_{\bar m},\hat {\bar m}-\bar m^*\rangle=\langle P_{1_n}\epsilon_{\bar m},\hat {\bar m}-\bar m^*\rangle\leq | u^T\epsilon_{\bar m}|\cdot\|\hat {\bar m}-\bar m^*\|_2\leq a( u^T\epsilon_{\bar m})^2+\frac{1}{2a}\mathrm{\mathbf{D}}_2(\hat {\bar m},\bar m^*),
\end{align}
where $P_{1_n}$ is the orthogonal projection matrix onto the column space of   $1_n$,  or equivalently,  $uu^T$ with $u=(1/\sqrt{n})1_n$.
By   Lemma \ref{psinormconvertlemma},
\begin{align}
    \mathbb{E} \big[( u^T\epsilon_{\bar m})^2\big]\leq \{\psi^{-1}(1)\}^2\|u^T\epsilon_{\bar m}\|^2_{\varphi}\leq  \{\psi^{-1}(1)\}^2\omega^2_{\bar m}.
\end{align}
Likewise, for the stochastic term $\langle\epsilon_{\varsigma},\hat {\varsigma}-\varsigma^*\rangle$  with $\epsilon_{\varsigma}=[\epsilon^T_{\frac{1}{\sigma}1_n},\epsilon^T_{\frac{1}{\nu}1_n}]^T$ and $\hat {\varsigma}-\varsigma^*=[(1/\hat\sigma-1/\sigma^*)1^T_n,(1/\hat\nu-1/\nu^*)1^T_n]^T$, we have
\begin{align}
   &\,\big\langle\epsilon_{\frac{1}{\sigma}1_n},(\frac{1}{\hat\sigma}-\frac{1}{\sigma^*})1_n\big\rangle+\big\langle\epsilon_{\frac{1}{\nu}1_n},(\frac{1}{\hat\nu}-\frac{1}{\nu^*})1_n\big\rangle\nonumber\\=&\,\big\langle P_{1_n}\epsilon_{\frac{1}{\sigma}1_n},(\frac{1}{\hat\sigma}-\frac{1}{\sigma^*})1_n\big\rangle+\big\langle P_{1_n}\epsilon_{\frac{1}{\nu}1_n},(\frac{1}{\hat\nu}-\frac{1}{\nu^*})1_n\big\rangle\nonumber\\
  \leq &\, a(u^T\epsilon_{\frac{1}{\sigma}1_n})^2+a(u^T\epsilon_{\frac{1}{\nu}1_n})^2+\frac{1}{2a}\mathrm{\mathbf{D}}_2(\hat {\varsigma},\varsigma^*),\nonumber
\end{align}
and\begin{align}
    \mathbb{E} \big[( u^T\epsilon_{\frac{1}{\sigma}1_n})^2\big]\leq \{\varphi^{-1}(1)\}^2\|u^T\epsilon_{\frac{1}{\sigma}1_n}\|^2_{\varphi}\leq  \{\varphi^{-1}(1)\}^2\omega^2_{\varsigma},\nonumber
\end{align}
and  $ \mathbb{E} [( u^T\epsilon_{\frac{1}{\nu}1_n})^2]\leq  \{\varphi^{-1}(1)\}^2\omega^2_{\varsigma}$.
To sum up, for any $a>0$,
\begin{align}\label{gammaineq}
   &\, \mathbb{E} \big\{   \langle\epsilon_{\bar m},\hat {\bar m}-\bar m^*\rangle+\langle\epsilon_{\varsigma},\hat {\varsigma}-\varsigma^*\rangle\big\}\nonumber\\ \leq&\, a\{\psi^{-1}(1)\}^2\omega^2_{\bar m}+2a\{\varphi^{-1}(1)\}^2\omega^2_{\varsigma}+\frac{1}{2a}\{\mathrm{\mathbf{D}}_2(\hat {\bar m},\bar m^*)+\mathrm{\mathbf{D}}_2(\hat {\varsigma},\varsigma^*)\}.
\end{align}

Plugging  \eqref{derive1} and \eqref{gammaineq} into   \eqref{basicineql2} to obtain
\begin{align}
     &\,\mathbb{E}\Big\{ \bold{\Delta}_{l}(\hat{\mu},\mu^*)+(\etaa-b)\mathrm{\bold{D}}_2(\hat{\gamma},\gamma^*)\Big\}-\frac{1}{2a}\{\mathrm{\bold{D}}_2(\hat {\bar m},\bar m^*)+\mathrm{\bold{D}}_2(\hat {\varsigma},\varsigma^*)\}\nonumber\\ \leq &\,C(1+A) c_\psi \psi^{-1}(  p)\omega_{\bar\eta}\mtbb\big\|\bar\beta^*\big\|_{2,1}+a\{\psi^{-1}(1)\}^2\omega^2_{\bar m}+2a\{\varphi^{-1}(1)\}^2\omega^2_{\varsigma}+\frac{\etaa^2}{2b}\|\gamma^*\|_{2}^2.
\end{align}
Choosing $b=\etaa/4$, $a=2/\etaa$, we get
\begin{align}
   &\,\mathbb{E}\Big\{ \bold{\Delta}_{l}(\hat{\mu},\mu^*)+\frac{\etaa}{2}\mathrm{\bold{D}}_2(\hat{\gamma},\gamma^*)\Big\}\nonumber\\ \leq&\, C(1+A)c_\psi \psi^{-1}(  p)\omega_{\bar\eta}\mtbb\big\|\bar\beta^*\big\|_{2,1}+\frac{2}{\etaa}\big[\{\psi^{-1}(1)\}^2\omega^2_{\bar m}+2\{\varphi^{-1}\}^2\omega^2_{\varsigma}\big]+2\etaa\|\gamma^*\|_{2}^2\nonumber\\\lesssim&\,c_\psi \psi^{-1}( p)\omega_{\bar\eta}\mtbb\big\|\bar\beta^*\big\|_{2,1}+\frac{1}{\etaa} (\omega^2_{\bar m}+\omega^2_{\varsigma})+\etaa\|\gamma^*\|_{2}^2.
\end{align}
The proof is complete.

\subsection{Proof of Theorem \ref{globalgrouplreconstanttheta} }
\label{appsub:proofofth2}
We prove a  more general result, which includes  Theorem \ref{globalgrouplreconstanttheta} as a special case.
\begin{theorem}
\label{globalgrouplregeneral}
Assume that the effective noises $\epsilon_{\bar\eta}, \epsilon_{\bar m}$, and $\epsilon_\varsigma$ satisfy  $\|\epsilon_{\bar\eta}\|_{\psi}\leq \omega_{\bar\eta},\|\epsilon_{\bar m}\|_{\psi}\leq \omega_{\bar m}$, and $\|\epsilon_\varsigma\|_{\varphi}\leq \omega_\varsigma$.
Let  $\hat\unv$ denote the  optimal solution  for \eqref{l1thmloss} with $\varrho\geq \rhotwoinf$.    Suppose that  there exist some $\vartheta>0$ and some large $K >0$ such that the following condition holds for any $\bar\beta,\gamma$ (recall  $\unv^T=[\bar\beta^T,\gamma^T]$, $\mu^T=[\bar\eta^T,\gamma^T]$ defined based on \eqref{systemnota})
\begin{align}\label{them3regu}
      (1+\vartheta)\lambda\mtbb\big\|(\bar{\beta}-\bar{\beta}^*)_{\mathcal{J}^*}\big\|_{2,1}\leq \bold{\Delta}_{l}(\mu,\mu^*)+\lambda\mtbb\big\|(\bar{\beta}-\bar{\beta}^*)_{\mathcal{J}^{*C}}\big\|_{2,1}+K\lambda^2J^*,
\end{align}
where
\begin{align} \label{lambdachoice-gen}
    \lambda=\big(1\lor\frac{1}{\vartheta}\big)\omega_{\bar\eta}\psi^{-1}\big[p\psi\big\{A\psi^{-1}(p)\big\}\big]
\end{align}
for some large enough $A>0$.
Then for any   $L,L'>0$, 
\begin{align}\label{thm3textbound}
\begin{split}  \bold{\Delta}_{l}(\hat{\mu},\mu^*)+\frac{\etaa}{2}\mathrm{\bold{D}}_2(\gamma,\gamma^*)\, \lesssim \, &\, (\frac{1}{\vartheta}\lor\frac{1}{\vartheta^3})K\omega^2_{\bar\eta}\Big(\psi^{-1}\Big[p\psi\big\{A\psi^{-1}(p)\big\}\Big]\Big)^2J^* \\&\,+(1\lor\frac{1}{\vartheta}) \frac{1}{\etaa}\big(L^2\omega^2_{\bar m}+L^{'2}\omega^2_{\varsigma}\big)+(1\lor\frac{1}{\vartheta})\etaa\|\gamma^*\|_{2}^2\end{split}
\end{align}
 holds with probability at least $1-C/\psi\{A\psi^{-1}(p)\}-1/\psi(cL)-C/\varphi(cL')$, where  $C,c$ are  positive constants.
\end{theorem}

Theorem \ref{globalgrouplregeneral} implies   Theorem \ref{globalgrouplreconstanttheta}. In fact,    \eqref{them3regu},  assuming $\vartheta$ is a constant, is just  \eqref{l1estimationerrorrate}, and   letting  $\psi=\psi_q$ for some $q> 0$, \eqref{lambdachoice-gen} yields  $\lambda=A\omega_{\bar\eta}(\log p)^{\frac{1}{q}}$. 
Then,     choosing $L=C\psi^{-1}(p^{A^q}), L'=C\varphi^{-1}(p^{A^q})$, we can obtain the result in Theorem \ref{globalgrouplreconstanttheta}:    \begin{align*}
    &\, \bold{\Delta}_{l}(\hat{\mu},\mu^*)+\frac{\etaa}{2}\mathrm{\bold{D}}_2(\hat\gamma,\gamma^*)\nonumber\\\lesssim  &\,
   KA^2\omega^2_{\bar\eta}J^*(\log p)^{\frac{2}{q}}+ \frac{1}{\etaa}\{\psi_q^{-1}(p^{A^q})\}^2\omega^2_{\bar m}+\frac{1}{\etaa}\{\varphi^{-1}(p^{A^q})\}^{2}\omega^2_{\varsigma}+\etaa\|\gamma^*\|_{2}^2
\end{align*}
 with probability at least $1-Cp^{-cA^q}$.

\begin{proof}
From $l(\hat{\mu})+\etaa\|\hat\gamma\|^2_2+\lambda\mtbb\|\hat{\bar{\beta}}\|_{2,1}\leq l(\mu^*)+\etaa\|\gamma^*\|^2_2+\lambda\mtbb\|\bar{\beta}^*\|_{2,1}$,  we obtain
\begin{align}\label{basicineqgoupl2}
     &\,\bold{\Delta}_{l}(\hat{\mu},\mu^*)+\frac{\etaa}{2}\|\hat{\gamma}-\gamma^*\|^2_2\nonumber\\\leq &\, \big\langle\epsilon_{\bar{\eta}},\bar{X}\hat{\bar{\beta}}-\bar{X}\bar{\beta}^*\big\rangle+\langle\epsilon_{\bar m},\hat {\bar m}-\bar m\rangle+\langle\epsilon_{\varsigma},\hat {\varsigma}-\varsigma\rangle+\langle-\etaa\gamma^*,\hat{\gamma}-\gamma^*\rangle\nonumber\\&\,+\lambda\mtbb\|\bar{\beta}^*\|_{2,1}-\lambda\mtbb\|\hat{\bar{\beta}}\|_{2,1}\nonumber\\\leq &\,\big\langle\epsilon_{\bar{\eta}},\bar{X}\hat{\bar{\beta}}-\bar{X}\bar{\beta}^*\big\rangle+\langle\epsilon_{\bar m},\hat {\bar m}-\bar m\rangle+\langle\epsilon_{\varsigma},\hat {\varsigma}-\varsigma\rangle+\frac{b}{2}\|\hat\gamma-\gamma^*\|^2_2+\frac{\etaa^2}{2b}\|\gamma^*\|_{2}^2\nonumber\\&\,+\lambda\mtbb\|\bar{\beta}^*\|_{2,1}-\lambda\mtbb\|\hat{\bar{\beta}}\|_{2,1}
\end{align}
for any $b>0$.

To bound $\big\langle\epsilon_{\bar{\eta}},\bar{X}\hat{\bar{\beta}}-\bar{X}\bar{\beta}^*\big\rangle$,
 we use the same notations in Appendix \ref{appsub:proofofth1} and recall
\begin{align}\label{holderineq}
  \langle\epsilon_{\bar{\eta}},\bar{X}\hat{\bar{\beta}}-\bar{X}\bar{\beta}^*\rangle\leq c_0\max_{1\leq k\leq 2p} |\bar{X}^T_k\epsilon_{\bar{\eta}}| \|\hat{\bar{\beta}}-\bar{\beta}^*\|_{2,1},
\end{align}
where the constant $c_0\geq \sqrt{2}$,
and\begin{align}
    \big\| \bar{X}^T_k\epsilon_{\bar{\eta}}\big\|_{\psi}\leq\omega_{\bar\eta}\mtbb, \  \forall k.
\end{align}
Let
 \begin{align}
     \lambda_0=\omega_{\bar\eta}\psi^{-1}\Big[p\psi\big\{A\psi^{-1}(p)\big\}\Big]
 \end{align}
 for some large enough $A>0$. By the union bound and Markov's inequality, the event
 $\underset{1\leq k\leq 2p}{\max}|\bar{X}^T_k\epsilon_{\bar{\eta}}|\geq \lambda_0\mtbb$  occurs with probability
\begin{align}\label{probboundl1}
    \mathbb{P}(\max_{1\leq k\leq 2p} |\bar{X}^T_k\epsilon_{\bar{\eta}}|\geq \lambda_0\mtbb)\leq  & \frac{2p}{\psi\Big(\frac{\omega_{\bar\eta}\mtbb\psi^{-1}\big[p\psi\{A\psi^{-1}(p)\}\big]}{ \omega_{\bar\eta}\mtbb}\Big)}\le \frac{2p}{ {  p\psi\{A\psi^{-1}(p)\}\big]}  }\leq    \frac{2}{\psi\{A\psi^{-1}(p)\}}.
\end{align}
Let $\lambda= \lambda_0/ b' $ with $b'>0$. By the subadditivity of the $\ell_1$-penalty, we get
\begin{align}\label{split}
    &\, c_0\lambda_0\mtbb\big\|\hat{\bar{\beta}}-\bar{\beta}^*\big\|_{2,1}+\lambda\mtbb\big\|\bar\beta^*\big\|_{2,1}-\lambda\mtbb\big\|\hat{\bar\beta}\big\|_{2,1}\nonumber\\\leq &\, c_0\lambda_0\mtbb\big\|(\hat{\bar{\beta}}-\bar{\beta}^*)_{\mathcal{J}^*}\big\|_{2,1}+ c_0\lambda_0\mtbb\big\|(\hat{\bar{\beta}}-\bar{\beta}^*)_{_{\mathcal{J}^{*C}}}\big\|_{2,1}+\lambda\mtbb\big\|\bar\beta^*_{\mathcal{J}^*}\big\|_{2,1}\nonumber\\&\,-\lambda\mtbb\big\|\hat{\bar\beta}_{_{\mathcal{J}^{*}}}\big\|_{2,1}-\lambda\mtbb\big\|(\hat{\bar{\beta}}-\bar{\beta}^*)_{_{\mathcal{J}^{*C}}}\big\|_{2,1}\nonumber\\
    \leq &\,\big( c_0\lambda_0+\lambda\big)\mtbb\big\|(\hat{\bar{\beta}}-\bar{\beta}^*)_{\mathcal{J}^*}\big\|_{2,1}-\big(\lambda- c_0\lambda_0\big)\mtbb\big\|(\hat{\bar{\beta}}-\bar{\beta}^*)_{_{\mathcal{J}^{*C}}}\big\|_{2,1}\nonumber\\=&\,\big( 1+c_0 b'\big)\lambda\mtbb\big\|(\hat{\bar{\beta}}-\bar{\beta}^*)_{\mathcal{J}^*}\big\|_{2,1}-\big(1- c_0 b'\big)\lambda\mtbb\big\|(\hat{\bar{\beta}}-\bar{\beta}^*)_{_{\mathcal{J}^{*C}}}\big\|_{2,1},
\end{align}
where $\mathcal J^{*C}$ is the complement of $\mathcal J^*$.

For  the stochastic term $\langle\epsilon_{\bar m},\hat {\bar m}-\bar m^*\rangle$,   for any $a>0$,
\begin{align}
   \langle\epsilon_{\bar m},\hat {\bar m}-\bar m^*\rangle=\langle P_{1_n}\epsilon_{\bar m},\hat {\bar m}-\bar m^*\rangle\leq | u^T\epsilon_{\bar m}|\cdot\|\hat {\bar m}-\bar m^*\|_2\leq a( u^T\epsilon_{\bar m})^2+\frac{1}{2a}\mathrm{\bold{D}}_2(\hat {\bar m},\bar m^*),
\end{align}
where $P_{1_n} = uu^T$ with $u=(1/\sqrt{n})1_n$.
By the assumption on $\epsilon_{\bar m}$, $\|u^T\epsilon_{\bar m}\|_\psi\leq \omega_{\bar m}$. By Markov's inequality,
\begin{align}
    \mathbb{P}\Big(( u^T\epsilon_{\bar m})^2\geq cL^2\omega^2_{\bar m}\Big)\leq \frac{1}{\psi(\frac{cL\omega_{\bar m}}{\omega_{\bar m}})}=\frac{1}{\psi(cL)},
\end{align}
from which it follow that
\begin{align}\label{thm2mbound}
   \langle\epsilon_{\bar m},\hat {\bar m}-\bar m^*\rangle\leq  acL^2\omega^2_{\bar m}+\frac{1}{2a}\mathrm{\bold{D}}_2(\hat {\bar m},\bar m^*)
\end{align}
occurs with probability at least $1-1/\psi(cL)$.
For the stochastic term $\langle\epsilon_{\varsigma},\hat {\varsigma}-\varsigma\rangle$ with  $\epsilon_{\varsigma}=[\epsilon^T_{\frac{1}{\sigma}1_n},\epsilon^T_{\frac{1}{\nu}1_n}]^T$ and $\hat {\varsigma}-\varsigma^*=[(1/\hat\sigma-1/\sigma^*)1^T_n,(1/\hat\nu-1/\nu^*)1^T_n]^T$, similarly, \begin{align}
   &\,\big\langle\epsilon_{\frac{1}{\sigma}1_n},(\frac{1}{\hat\sigma}-\frac{1}{\sigma^*})1_n\big\rangle+\big\langle\epsilon_{\frac{1}{\nu}1_n},(\frac{1}{\hat\nu}-\frac{1}{\nu^*})1_n\big\rangle\nonumber\\=&\,\big\langle P_{1_n}\epsilon_{\frac{1}{\sigma}1_n},(\frac{1}{\hat\sigma}-\frac{1}{\sigma^*})1_n\big\rangle+\big\langle P_{1_n}\epsilon_{\frac{1}{\nu}1_n},(\frac{1}{\hat\nu}-\frac{1}{\nu^*})1_n\big\rangle\nonumber\\
  \leq &\, a(u^T\epsilon_{\frac{1}{\sigma}1_n})^2+a(u^T\epsilon_{\frac{1}{\nu}1_n})^2+\frac{1}{2a}\mathrm{\bold{D}}_2(\hat {\varsigma},\varsigma^*).\nonumber
\end{align}
By the assumption on $\epsilon_\varsigma$, $\|u^T\epsilon_{\frac{1}{\sigma}1_n}\|_\varphi\leq \omega_\varsigma, \|u^T\epsilon_{\frac{1}{\nu}1_n}\|_\varphi\leq \omega_\varsigma$. It follows from \begin{align}
     \mathbb{P}\Big(( u^T\epsilon_{\frac{1}{\sigma}1_n})^2\geq cL^{'2}\omega^2_{\varsigma}\Big)\leq \frac{1}{\psi(\frac{cL\omega_\varsigma}{\omega_\varsigma})}=\frac{1}{\psi(cL')}
\end{align}
and $ \mathbb{P}\Big(( u^T\epsilon_{\frac{1}{\nu}1_n})^2\geq cL^{'2}\omega^2_{\varsigma}\Big)\leq 1/\psi(cL')$ that \begin{align}\label{thm2varsigmabound}
    \langle\epsilon_{\varsigma},\hat {\varsigma}-\varsigma\rangle\leq 2acL^{'2}\omega^2_{\varsigma} +\frac{1}{2a}\mathrm{\bold{D}}_2(\hat {\varsigma},\varsigma^*),
\end{align}
with probability at least $1-2/\varphi(cL')$.

Plugging \eqref{probboundl1}, \eqref{split}, \eqref{thm2mbound}, and \eqref{thm2varsigmabound}   into   \eqref{basicineqgoupl2}   results in
\begin{align}\label{groupl1basicineq}
 &\,\bold{\Delta}_{l}(\hat{\mu},\mu^*)+(\etaa-b)\mathrm{\bold{D}}_2(\hat\gamma,\gamma^*)-\frac{1}{2a}\mathrm{\bold{D}}_2(\hat{\bar m},\bar m^*)-\frac{1}{2a}\mathrm{\bold{D}}_2(\hat\varsigma,\varsigma^*)\nonumber\\ \leq  &\,\big( 1+c_0b'\big)\lambda\mtbb\big\|(\hat{\bar{\beta}}-\bar{\beta}^*)_{\mathcal{J}^*}\big\|_{2,1}-\big(1- c_0b'\big)\lambda\mtbb\big\|(\hat{\bar{\beta}}-\bar{\beta}^*)_{_{\mathcal{J}^{*C}}}\big\|_{2,1}+\frac{\etaa^2}{2b}\|\gamma^*\|_{2}^2\nonumber\\&\,+ acL^2\omega^2_{\bar m}+2acL^{'2}\omega^2_{\varsigma}.
\end{align}
with probability at least $1-C/\psi\{A\psi^{-1}(p)\}-1/\psi(cL)-C/\varphi(cL')$, where  $C,c$ are  positive constants.

The regularity condition \eqref{them3regu} implies  \begin{align}\label{thm3reguused}
    &\,(1+\frac{\vartheta}{2+\vartheta})\lambda\mtbb\big\|(\hat{\bar{\beta}}-\bar{\beta}^*)_{\mathcal{J}^*}\big\|_{2,1}\nonumber\\\leq&\, \frac{2}{2+\vartheta}\bold{\Delta}_{l}(\hat{\mu},\mu^*)+(1-\frac{\vartheta}{2+\vartheta})\lambda\mtbb\big\|(\hat{\bar{\beta}}-\bar{\beta}^*)_{\mathcal{J}^{*C}}\big\|_{2,1}+\frac{2}{2+\vartheta}K\lambda^2J^*.
\end{align}
Set $b'=\vartheta/\{(2+\vartheta)c_0\}$ and
add    \eqref{groupl1basicineq} and  \eqref{thm3reguused} to get \begin{align}
 &\,  \frac{\vartheta}{2+\vartheta}\bold{\Delta}_{l}(\hat{\mu},\mu^*)+ (\etaa-b)\mathrm{\bold{D}}_2(\hat\gamma,\gamma^*)-\frac{1}{2a}\mathrm{\bold{D}}_2(\hat{\bar m},\bar m^*)-\frac{1}{2a}\mathrm{\bold{D}}_2(\hat\varsigma,\varsigma^*)\nonumber\\\leq &\,\frac{2}{2+\vartheta}K\lambda^2J^*+ acL^2\omega^2_{\bar m}+2acL^{'2}\omega^2_{\varsigma}+\frac{\etaa^2}{2b}\|\gamma^*\|_{2}^2.
\end{align}
Taking $b=\etaa/4$ and $  a=2/\etaa$ leads to
\begin{align}
    \frac{\vartheta}{2+\vartheta}\bold{\Delta}_{l}(\hat{\mu},\mu^*)+\frac{\etaa}{2}\mathrm{\bold{D}}_2(\hat\gamma,\gamma^*)\leq \frac{2}{2+\vartheta}K\lambda^2J^*+ \frac{2}{\etaa}cL^2\omega^2_{\bar m}+\frac{4}{\etaa}cL^{'2}\omega^2_{\varsigma}+2\etaa\|\gamma^*\|_{2}^2,
\end{align}
or equivalently
\begin{align}\label{thmmresult}
    \bold{\Delta}_{l}(\hat{\mu},\mu^*)+\frac{(2+\vartheta)\etaa}{2\vartheta}\mathrm{\bold{D}}_2(\hat\gamma,\gamma^*)\leq \frac{2}{\vartheta}K\lambda^2J^*+\frac{2+\vartheta}{\vartheta}\frac{2}{\etaa}\big(cL^2\omega^2_{\bar m}+2cL^{'2}\omega^2_{\varsigma}\big)+\frac{2+\vartheta}{\vartheta}2\etaa\|\gamma^*\|_{2}^2.
\end{align}
Note that $\{\etaa(2+\vartheta)\}/(2\vartheta)\geq \etaa/2$. With $\lambda= \lambda_0/b'=\{(2+\vartheta)/\vartheta\}c_0\lambda_0$, we can derive from
\eqref{thmmresult} that
 \begin{align}\label{thm3finalbound}
   &\,\bold{\Delta}_{l}(\hat{\mu},\mu^*)+\frac{\etaa}{2}\mathrm{\bold{D}}_2(\hat\gamma,\gamma^*)\nonumber\\\leq&\, Kc^2_0\frac{2(2+\vartheta)^2}{\vartheta^3}\lambda^2_0J^*+\frac{2+\vartheta}{\vartheta}\frac{2}{\etaa}\big(cL^2\omega^2_{\bar m}+2cL^{'2}\omega^2_{\varsigma}\big)+\frac{2+\vartheta}{\vartheta}2\etaa\|\gamma^*\|_{2}^2\nonumber\\\lesssim &\,
   (\frac{1}{\vartheta}\lor\frac{1}{\vartheta^3})K\omega^2_{\bar\eta}\bigg(\psi^{-1}\Big[p\psi\big\{A\psi^{-1}(p)\big\}\Big]\bigg)^2J^*+(1\lor\frac{1}{\vartheta}) \frac{1}{\etaa}\big(L^2\omega^2_{\bar m}+L^{'2}\omega^2_{\varsigma}\big)\nonumber\\&\,+(1\lor\frac{1}{\vartheta})\etaa\|\gamma^*\|_{2}^2,
\end{align}
with probability at least $1-C/\psi\{A\psi^{-1}(p)\}-1/\psi(cL)-C/\varphi(cL')$, where  $C,c$ are  positive constants.
The proof is complete.

\end{proof}

\subsection{An Elementwise  Estimation Error Bound}\label{inftyboundsubsect}
\begin{theorem}\label{inftyboundprop}
Assume that the effective noises $\epsilon_{\bar\eta},\epsilon_{\bar m}$, and $\epsilon_\varsigma$ satisfy    $\|\epsilon_{\bar\eta}\|_{\psi_q}\leq \omega_{\bar\eta},\|\epsilon_{\bar m}\|_{\psi_q}\leq \omega_{\bar m}$, and $\|\epsilon_\varsigma\|_{\varphi}\leq \omega_\varsigma$  for some  $q>0$.    Consider  $\hat\unv$ as the  optimal solution to \eqref{l1thmloss} with $\tau=0$, $\varrho\geq \rhotwoinf$, and
 \begin{align}
    \lambda=A\omega_{\bar\eta}(\log p)^{\frac{1}{q}}
 \end{align}
for some large enough $A>0$.
Suppose that  there exist some $\vartheta,\alpha>0$ and some large $K >0$ such that  for any $\unv^T=[\bar\beta^T,\gamma^T]$
\begin{align}\label{l1infre}
     &\,\alpha  nJ^*\|\bar{\beta}-\bar{\beta}^*\|^2_{2,\infty}+\alpha\mathrm{\bold{D}}_2(\gamma,\gamma^*) +  (1+\vartheta)\lambda\mtbb\big\|(\bar{\beta}-\bar{\beta}^*)_{\mathcal{J}^*}\big\|_{2,1}\nonumber\\\leq &\, \bold{\Delta}_{l}(\mu,\mu^*)+(1-\vartheta)\lambda\mtbb\big\|(\bar{\beta}-\bar{\beta}^*)_{\mathcal{J}^{*C}}\big\|_{2,1}+K\lambda^2J^*.
\end{align}
Then 
\begin{align}\label{l1infbound}
    \|\hat{\bar{\beta}}-\bar{\beta}^*\|_{2,\infty}
    \leq C\frac{\sqrt{K\alpha}\lor\vartheta}{\alpha\sqrt{n}}A\Big\{\omega_{\bar\eta}+\frac{1}{\sqrt{J^*}}(\omega_{\bar m}+\omega_{\varsigma})\Big\}(\log p)^{\frac{1}{q}}
\end{align}
with probability at least $1-Cp^{-c(A \vartheta)^q}-C/\varphi(cA \vartheta (\log p)^{\frac{1}{q}}) $, where  $C,c$ are positive constants.
\end{theorem}

The element-wise error bound \eqref{l1infbound} together with a signal strength condition guarantees faithful variable
selection  with high probability; see  Remark \ref{rem-moreres}.

\begin{proof}

By the  optimality of $\hat{\mu}$: $l(\hat{\mu})+\lambda\mtbb\|\hat{\bar{\beta}}\|_{2,1}\leq l(\mu^*)+\lambda\mtbb\|\bar{\beta}^*\|_{2,1}$,  we obtain the basic inequality
\begin{align}\label{thmvarselbasicineq}
  \bold{\Delta}_{l}(\hat{\mu},\mu^*)\leq  \big\langle\epsilon_{\bar{\eta}},\bar{X}\hat{\bar{\beta}}-\bar{X}\bar{\beta}^*\big\rangle+\langle\epsilon_{\bar m},\hat {\bar m}-\bar m\rangle+\langle\epsilon_{\varsigma},\hat {\varsigma}-\varsigma\rangle+\lambda\mtbb\|\bar{\beta}^*\|_{2,1}-\lambda\mtbb\|\hat{\bar{\beta}}\|_{2,1}.
\end{align}

We follow the same lines as   in Section \ref{appsub:proofofth2} to bound the  stochastic terms on the right-hand side of \eqref{thmvarselbasicineq}.
Let
 \begin{align}
 \lambda_0=A\omega_{\bar\eta}(\log p)^{\frac{1}{q}}
 \end{align}
 for some large enough $A>0$ and $\lambda=\lambda_0/b'$ for some $b'>0$. Similar to the previous analysis, we have  with probability at least $1-Cp^{-A^q}$, \begin{align}\label{thmvarseleta}
    &\,\big\langle\epsilon_{\bar{\eta}},\bar{X}\hat{\bar{\beta}}-\bar{X}\bar{\beta}^*\big\rangle+\lambda\mtbb\|\bar{\beta}^*\|_{2,1}-\lambda\mtbb\|\hat{\bar{\beta}}\|_{2,1}\nonumber\\\leq &\, \big( 1+c_0 b'\big)\lambda\mtbb\big\|(\hat{\bar{\beta}}-\bar{\beta}^*)_{\mathcal{J}^*}\big\|_{2,1}-\big(1- c_0 b'\big)\lambda\mtbb\big\|(\hat{\bar{\beta}}-\bar{\beta}^*)_{_{\mathcal{J}^{*C}}}\big\|_{2,1},
\end{align}
where the constant $c_0\geq \sqrt{2}$. Moreover, for any $a>0$, we get
\begin{align}\label{thmvarselmsigma}
     \langle\epsilon_{\bar m},\hat {\bar m}-\bar m^*\rangle+  \langle\epsilon_{\varsigma},\hat {\varsigma}-\varsigma\rangle\leq  acL^2\omega^2_{\bar m}+2acL^{'2}\omega^2_{\varsigma}+\frac{1}{2a}\big\{\mathrm{\bold{D}}_2(\hat {\bar m},\bar m^*)+\mathrm{\bold{D}}_2(\hat {\varsigma},\varsigma^*)\big\},
\end{align}
with probability at least $1-1/\psi_q(cL)-C/\varphi(cL')$, where $L,L'>0$ can be customized.

Plugging the bounds of \eqref{thmvarseleta} and \eqref{thmvarselmsigma} into  \eqref{thmvarselbasicineq},  we get
\begin{align*}\label{thmvarselmedineq}
    &\, \bold{\Delta}_{l}(\hat{\mu},\mu^*)-\frac{1}{2a}\big\{\mathrm{\bold{D}}_2(\hat {\bar m},\bar m^*)+\mathrm{\bold{D}}_2(\hat {\varsigma},\varsigma^*)\big\} \\
\leq &\, \big( 1+c_0 b'\big)\lambda\mtbb\big\|(\hat{\bar{\beta}}-\bar{\beta}^*)_{\mathcal{J}^*}\big\|_{2,1}-\big(1- c_0 b'\big)\lambda\mtbb\big\|(\hat{\bar{\beta}}-\bar{\beta}^*)_{_{\mathcal{J}^{*C}}}\big\|_{2,1}+acL^2\omega^2_{\bar m}+2acL^{'2}\omega^2_{\varsigma},
\end{align*}
with probability at least
$1-Cp^{-A^q}-1/\psi_q(cL)-C/\varphi(cL')$, where  $C,c$ are  positive constants.

With  the regularity condition and and $b'=\vartheta/c_0$, we obtain
\begin{align}\label{prop2ineq}
   &\, \alpha \big(nJ^*\|\bar{\beta}-\bar{\beta}^*\|^2_{2,\infty}+\mathrm{\bold{D}}_2(\gamma,\gamma^*)\big)-\frac{1}{2a}\big\{\mathrm{\bold{D}}_2(\hat {\bar m},\bar m^*)+\mathrm{\bold{D}}_2(\hat {\varsigma},\varsigma^*)\big\}\nonumber\\\leq &\, K\lambda^2J^*+acL^2\omega^2_{\bar m}+2acL^{'2}\omega^2_{\varsigma}.
\end{align}

Setting  $a=1/\alpha$ and using $\lambda=\lambda_0/b'=(c_0/\vartheta)\lambda_0$, we obtain
\begin{align}
    \|\bar{\beta}-\bar{\beta}^*\|^2_{2,\infty}\leq &\, \frac{Kc_0^2}{n\alpha\vartheta^2}A^2\omega^2_{\bar\eta}(\log p)^{\frac{2}{q}}+\frac{cL^2\omega^2_{\bar m}}{n\alpha^2J^*}+\frac{2cL^{'2}\omega^2_{\varsigma}}{n\alpha^2J^*}
\end{align}
with probability at least $1-Cp^{-A^q}-1/\psi_q(cL)-C/\varphi(cL')$.
With $L=L'=A (\log p)^{\frac{1}{q}}$,
\begin{align}
    \|\bar{\beta}-\bar{\beta}^*\|^2_{2,\infty}\leq &\, \frac{Kc_0^2}{n\alpha\vartheta^2}A^2\omega^2_{\bar\eta}(\log p)^{\frac{2}{q}}+\frac{c}{n\alpha^2J^*}A^2(\log p)^{\frac{2}{q}}(\omega^2_{\bar m}+\omega^2_{\varsigma})\nonumber\\\leq&\, C\frac{A^2}{n\alpha^2\vartheta^2}(\log p)^{\frac{2}{q}}\big\{K\alpha\omega^2_{\bar\eta}+\frac{\vartheta^2}{J^*}(\omega^2_{\bar m}+\omega^2_\varsigma)\big\},
\end{align}
which implies
\begin{align}
      \|\bar{\beta}-\bar{\beta}^*\|_{2,\infty}\leq C \frac{\sqrt{K\alpha}\lor\vartheta}{\sqrt{n}\alpha\vartheta}A(\log p)^{\frac{1}{q}}\big\{\omega_{\bar\eta}+\frac{1}{\sqrt{J^*}}(\omega_{\bar m}+\omega_{\varsigma})\big\},
\end{align}
with probability at least $1-Cp^{-cA^q}-C/\varphi(cA (\log p)^{ {1}/{q}})$.
 Hence by taking $A'=A/\vartheta$ and $\lambda=A'\omega_{\bar \eta}(\log p)^{\frac{1}{q}}$,  we have
\begin{align}
      \|\bar{\beta}-\bar{\beta}^*\|_{2,\infty}\leq C \frac{\sqrt{K\alpha}\lor\vartheta}{\sqrt{n}\alpha}A'(\log p)^{\frac{1}{q}}\big\{\omega_{\bar\eta}+\frac{1}{\sqrt{J^*}}(\omega_{\bar m}+\omega_{\varsigma})\big\},
\end{align}
with probability at least $1-Cp^{-c(A'\vartheta)^q}-C/\varphi(cA'\vartheta (\log p)^{\frac{1}{q}})$. The proof is complete.
\end{proof}

\subsection{Analysis of a General Penalty}\label{generalpenaltysubsect}
Consider  the following problem associated with
 a general sparsity-inducing penalty $P$
\begin{align}\label{quadraticpenaltylossap}
   l(\mu)+\|\mtbb\bar{\beta}\|_{2,P}+\frac{\etaa}{2}\|\mu\|^2_2,
\end{align}
where   $\|\cdot\|_{2,P}$ is  short  for $\|\cdot\|_{2,P(\cdot;\lambda)}$  and $\|\bar{\beta}\|_{2,P(\cdot;\lambda)}:=\sum^p_{k=1}P(\|\bar{\beta}_k\|_2;\lambda)$. 

Since a sparsity-inducing penalty necessarily possesses thresholding power, we assume,  without loss of generality,  that $P(\cdot;\lambda)\geq P_{H}(\cdot;\lambda)$ throughout the subsection, where the ``hard penalty'' $ P_{H}(t;\lambda):=(-t^2/2+\lambda |t|)1_{|t|<\lambda}+(\lambda^2/2)1_{|t|\geq\lambda}$ is induced by the hard-thresholding, and   
 $P_{2,H}(\bar{\beta};\lambda):=\sum^p_{k=1}P_{H}(\|\bar{\beta}_k\|_2;\lambda)$ (see \cite{she2016finite} for more details).   Furthermore, a penalty is referred to as subadditive if it satisfies  $P(t+s;\lambda)\leq P(t;\lambda)+P(s;\lambda)$.
In fact, when $P(t)$ is concave on $\mathbb R_+$ and $P(0)=0$,  $P(|t|)$ is necessarily subadditive. Well-known examples   include the widely used $\ell_1$-penalty, $\ell_0$-penalty, SCAD, MCP, and bridge $\ell_r$ $(0< r<1)$  in the literature. 

\begin{theorem}\label{globalthmsubgaussian}

Assume that the effective noises $\epsilon_{\bar\eta},\epsilon_{\bar m}$, and $\epsilon_\varsigma$ satisfy    $\|\epsilon_{\bar\eta}\|_{\psi_2}\leq \omega_{\bar\eta},\|\epsilon_{\bar m}\|_{\psi_2}\leq \omega_{\bar m}$, and $\|\epsilon_\varsigma\|_{\varphi}\leq \omega_\varsigma$ where  $\{\varphi^{-1}(\cdot)\}^2$  concave or $\{\varphi^{-1}(t)\}^2\lesssim t$ on $\mathbb{R}_+$ (for example, $\varphi$ can be an $L_q$-norm  with $q\geq 2$ or a $\psi_q$-norm with $q>0$).
Let  $\hat\unv$ denote the  optimal solution  of \eqref{quadraticpenaltylossap} with $\varrho\geq \rhotwo: = \|\bar{X}\|_2$.

 (a) Let $\lambda=A\omega_{\bar\eta}\sqrt{\log(ep)}/\sqrt{\etaa\land 1}$ with A a sufficiently large constant.  Then the following bound always holds
 \begin{align}
       \mathbb{E}\bigg\{ \bold{\Delta}_{l}(\hat{\mu},\mu^*)+\frac{\etaa}{2}\mathrm{\bold{D}}_2(\hat{\mu},\mu^*)\bigg\}\lesssim&\,\big\| \varrho\bar{\beta}^*\big\|_{2,P(\cdot;\lambda)}+\frac{1}{\etaa\land 1}\omega^2_{\bar\eta}+\frac{1}{\etaa}\big(\omega^2_{\bar m}+\omega_\varsigma^2\big)+\etaa\|\mu^*\|_{2}^2.
 \end{align}

 (b)  Let $P$ be a  sub-additive penalty.
Assume that there exist some $\alpha\geq0$, $\vartheta>0$, and some large $K>0$  such that
\begin{align}\label{alterregugeneral}
  \alpha  \mathrm{\bold{D}}_2(\mu,\mu^*)+ (1+\vartheta)\|\varrho(\bar{\beta}-\bar{\beta}^*)_{\mathcal{J}^*}\|_{2,P}\leq \bold{\Delta}_{l}(\mu,\mu^*)+(1-\vartheta)\|\varrho(\bar{\beta}-\bar{\beta}^*)_{\mathcal{J}^{*C}}\|_{2,P}+K\lambda^2J^*,
\end{align}
 for all $\unv^T=[\bar\beta^T,\bar m^T,\varsigma^T]$,
where $\lambda=A\omega_{\bar\eta}\sqrt{\log(ep)}/\sqrt{\{(\etaa+\alpha)\land \vartheta\}\vartheta}$ with A a sufficiently large constant.
Then
\begin{align}
  \mathbb{E}\,\mathrm{\bold{ D}}_2(\hat\mu,\mu^*)\lesssim \ &\,\frac{\omega_{\bar\eta}^2}{(\etaa+\alpha)\{(\etaa+\alpha)\land \vartheta\}\vartheta} \, \big\{KA^2J^*\log(ep)+\vartheta\big\}\nonumber\\ &+\frac{\omega^2_{\bar m}+\{\varphi^{-1}(1)\}^2\omega^2_{\varsigma}}{(\etaa+\alpha)^2}  +\big(1\land \frac{\etaa}{\alpha}\big)^2\|\mu^*\|_{2}^2.
\end{align}

\end{theorem}


 According to  the proof below, \eqref{alterregugeneral} can be relaxed to
$\vartheta P_{2,H}\big(\varrho(\bar{\beta}-\bar{\beta}^*)_{\mathcal{J}^*};\lambda\big)+\alpha  \mathrm{\bold{D}}_2(\mu,\mu^*)+\|\varrho(\bar{\beta}-\bar{\beta}^*)_{\mathcal{J}^*}\|_{2,P}\nonumber \leq    \bold{\Delta}_{l}(\mu,\mu^*)+\|\varrho(\bar{\beta}-\bar{\beta}^*)_{\mathcal{J}^{*C}}\|_{2,P}-\vartheta P_{2,H}\big(\varrho(\bar{\beta}-\bar{\beta}^*)_{\mathcal{J}^{*C}};\lambda\big)+K\lambda^2J^*
$, or $
 \vartheta P_{2,H}\big(\varrho(\bar{\beta}-\bar{\beta}^*);\lambda\big)+\alpha  \mathrm{\bold{D}}_2(\mu,\mu^*)+\| \varrho\bar{\beta}^*\|_{2,P}\leq  \bold{\Delta}_{l}(\mu,\mu^*)+\|\varrho\bar{\beta}\|_{2,P}+K\lambda^2J^*
$ if $P$ is not subadditive.

\begin{proof}

  By definition,   $l(\hat{\zeta})+\frac{\etaa}{2}\|\hat\mu\|^2_2+\big\|\mtbb\hat {\bar\beta}\big\|_{2,P}\leq l(\zeta^*)+\frac{\etaa}{2}\|\mu^*\|^2_2+\big\| \mtbb\bar{\beta}^*\big\|_{2,P}$, which means
 \begin{align}\label{basiceq}
    &\,\bold{\Delta}_{l}(\hat{\mu},\mu^*)+\frac{\etaa}{2}\|\hat{\mu}-\mu^*\|^2_2\nonumber\\\leq &\, \big\langle\epsilon_{\bar{\eta}},\bar{X}\hat{\bar{\beta}}-\bar{X}\bar{\beta}^*\big\rangle+\langle\epsilon_{\bar m},\hat {\bar m}-\bar m^*\rangle+\langle\epsilon_{\varsigma},\hat {\varsigma}-\varsigma^*\rangle+\langle-\etaa\mu^*,\hat{\mu}-\mu^*\rangle\nonumber\\&\,+\big\| \mtbb\bar{\beta}^*\big\|_{2,P}-\big\|\mtbb\hat {\bar\beta}\big\|_{2,P}\nonumber\\\leq &\,\big\langle\epsilon_{\bar{\eta}},\bar{X}\hat{\bar{\beta}}-\bar{X}\bar{\beta}^*\big\rangle+\langle\epsilon_{\bar m},\hat {\bar m}-\bar m^*\rangle+\langle\epsilon_{\varsigma},\hat {\varsigma}-\varsigma^*\rangle+\frac{b}{2}\|\hat\mu-\mu^*\|^2_2+\frac{\etaa^2}{2b}\|\mu^*\|_{2}^2\nonumber\\&\,+\big\| \mtbb\bar{\beta}^*\big\|_{2,P}-\big\|\mtbb\hat {\bar\beta}\big\|_{2,P},
\end{align}
where $b$ can be any positive number.

To bound the stochastic term $\langle\epsilon_{\bar\eta},\bar{X}(\hat{\bar{\beta}}-\bar{\beta}^*)\rangle$,
define  $\lambda_0=\omega_{\bar\eta}\sqrt{\log(ep)}$  and
\begin{align}\label{Rdef}
    R=&\,\sup_{\bar{\beta}\in \mathbb{R}^{2p}}\Big\{\langle\epsilon_{\bar\eta},\bar{
    X}(\bar{\beta}-\bar{\beta}^*)\rangle-\frac{1}{2a'}\big\|\bar{
    X}(\bar{\beta}-\bar{\beta}^*)\big\|^2_2-\frac{1}{2b'}P_{2,H}\big(\varrho(\bar{\beta}-\bar{\beta}^*);\sqrt{a'b'}A\lambda_0\big)\Big\}.
\end{align}
By Lemma \ref{basiclemma}, for any $a'\geq 2b'>0$ and a sufficiently large constant $A$,  $\mathbb{P}(R\geq a'\omega^2_{\bar\eta} t)\leq C\exp(-ct)p^{{-cA^2}}$. Therefore, $\mathbb{E}[R]\lesssim a' \omega^2_{\bar\eta}$.

Next, we bound      $\langle\epsilon_{\bar m},\hat {\bar m}-\bar m^*\rangle$ by \begin{align}
   \langle\epsilon_{\bar m},\hat {\bar m}-\bar m^*\rangle=\langle P_{1_n}\epsilon_{\bar m},\hat {\bar m}-\bar m^*\rangle\leq | u^T\epsilon_{\bar m}|\cdot\|\hat {\bar m}-\bar m^*\|_2\leq a(u^T\epsilon_{\bar m})^2+\frac{1}{2a}\mathrm{\bold{D}}_2(\hat {\bar m},\bar m^*)
\end{align}
 for any $a>0$,
where $P_{1_n}=uu^T$ with $u=(1/\sqrt{n})1_n$. By the assumption on $\epsilon_{\bar m}$, $\|(u^T\epsilon_{\bar m})^2\|_{\psi_1}\lesssim\|u^T\epsilon_{\bar m}\|^2_{\psi_2}\leq \omega^2_{\bar m}$, from which it follows that $ \mathbb{E}[ (u^T\epsilon_{\bar m})^2]\leq C\omega^2_{\bar m}$.
Similarly, for the stochastic term $\langle\epsilon_{\varsigma},\hat {\varsigma}-\varsigma^*\rangle$ with $\epsilon_{\varsigma}=[\epsilon^T_{\frac{1}{\sigma}1_n},\epsilon^T_{\frac{1}{\nu}1_n}]^T$ and $\hat {\varsigma}-\varsigma^*=[(1/\hat\sigma-1/\sigma^*)1^T_n,(1/\hat\nu-1/\nu^*)1^T_n]^T$, we have
\begin{align}
   &\,\big\langle\epsilon_{\frac{1}{\sigma}1_n},(\frac{1}{\hat\sigma}-\frac{1}{\sigma^*})1_n\big\rangle+\big\langle\epsilon_{\frac{1}{\nu}1_n},(\frac{1}{\hat\nu}-\frac{1}{\nu^*})1_n\big\rangle\nonumber\\=&\,\big\langle P_{1_n}\epsilon_{\frac{1}{\sigma}1_n},(\frac{1}{\hat\sigma}-\frac{1}{\sigma^*})1_n\big\rangle+\big\langle P_{1_n}\epsilon_{\frac{1}{\nu}1_n},(\frac{1}{\hat\nu}-\frac{1}{\nu^*})1_n\big\rangle\nonumber\\
  \leq &\, a(u^T\epsilon_{\frac{1}{\sigma}1_n})^2+a(u^T\epsilon_{\frac{1}{\nu}1_n})^2+\frac{1}{2a}\mathrm{\bold{D}}_2(\hat {\varsigma},\varsigma^*).\nonumber
\end{align}
From Lemma \ref{psinormconvertlemma}, we get a bound in $L_2$-norm:
\begin{align}
    \mathbb{E} \big[( u^T\epsilon_{\frac{1}{\sigma}1_n})^2\big]\leq \{\varphi^{-1}(1)\}^2\|u^T\epsilon_{\frac{1}{\sigma}1_n}\|^2_{\varphi}\leq  \{\varphi^{-1}(1)\}^2\omega^2_{\varsigma},\nonumber
\end{align}
and  $ \mathbb{E} [( u^T\epsilon_{\frac{1}{\nu}1_n})^2]\leq  \{\varphi^{-1}(1)\}^2\omega^2_{\varsigma}$.
 To sum up, for any $a>0$,
\begin{align}\label{gammaineq1}
   \mathbb{E} \big\{   \langle\epsilon_{\bar m},\hat {\bar m}-\bar m^*\rangle+\langle\epsilon_{\varsigma},\hat {\varsigma}-\varsigma^*\rangle\big\}\leq aC\omega^2_{\bar m}+2a\{\varphi^{-1}(1)\}^2\omega^2_{\varsigma}+\frac{1}{2a}\{\mathrm{\bold{D}}_2(\hat {\bar m},\bar m^*)+\mathrm{\bold{D}}_2(\hat {\varsigma},\varsigma^*)\}.
\end{align}

Now, plugging the bounds  \eqref{Rdef} and  \eqref{gammaineq1} into \eqref{basiceq} yields
\begin{align}\label{middlebound1}
&\,\mathbb{E}\Big\{\bold{\Delta}_{l}(\hat{\mu},\mu^*)+\frac{\etaa-b}{2}\|\hat{\mu}-\mu^*\|^2_2\Big\}-\frac{1}{a'}\mathrm{\bold{D}}_2(\hat{\bar\eta},\bar{\eta}^*)-\frac{1}{2a}\{\mathrm{\bold{D}}_2(\hat {\bar m},\bar m^*)+\mathrm{\bold{D}}_2(\hat {\varsigma},\varsigma^*)\} \nonumber\\\leq &\, \frac{1}{2b'}P_{2,H}\big(\varrho(\hat{\bar{\beta}}-\bar{\beta}^*);\sqrt{a'b'}A\lambda_0\big)+\big\| \varrho\bar{\beta}^*\big\|_{2,P(\cdot;\lambda)}-\big\|\varrho\hat {\bar\beta}\big\|_{2,P(\cdot;\lambda)}\nonumber\\&+aC\omega^2_{\bar m}+2a\{\varphi^{-1}(1)\}^2\omega^2_{\varsigma}+\frac{\etaa^2}{2b}\|\mu^*\|_{2}^2+Ca'\omega^2_{\bar\eta}.
\end{align}

To prove part  (a), we use the subadditivity of $P_H$:
\begin{align}
    &\,\mathbb{E}\Big\{\bold{\Delta}_{l}(\hat{\mu},\mu^*)+\frac{\etaa-b}{2}\|\hat{\mu}-\mu^*\|^2_2\Big\}-\frac{1}{a'}\mathrm{\bold{D}}_2(\hat{\bar\eta},\bar{\eta}^*)-\frac{1}{2a}\{\mathrm{\bold{D}}_2(\hat {\bar m},\bar m^*)+\mathrm{\bold{D}}_2(\hat {\varsigma},\varsigma^*)\}\nonumber\\\leq &\, \frac{1}{2b'}P_{2,H}\big(\varrho\hat{\bar{\beta}};\sqrt{a'b'}A\lambda_0\big)+\frac{1}{2b'}P_{2,H}\big(\bar{\beta}^*;\sqrt{a'b'}A\lambda_0\big)+\big\| \varrho\bar{\beta}^*\big\|_{2,P(\cdot;\lambda)}-\big\|\varrho\hat {\bar\beta}\big\|_{2,P(\cdot;\lambda)} \nonumber\\&+aC\omega^2_{\bar m}+2a\{\varphi^{-1}(1)\}^2\omega^2_{\varsigma}+\frac{\etaa^2}{2b}\|\mu^*\|_{2}^2+Ca'\omega^2_{\bar\eta}.
\end{align}
Because $P(\cdot;\lambda)\geq P_{H}(\cdot;\lambda)$, taking $b=\etaa/4,b'=1/2,a=2/\etaa,a'=4/(\etaa\land 1)$ gives
\begin{align}
    &\,\mathbb{E}\Big\{\bold{\Delta}_{l}(\hat{\mu},\mu^*)+\frac{\etaa}{2}\mathrm{\bold{ D}}_2(\hat{\mu},\mu^*)\Big\}\nonumber\\\leq &\,2\big\| \varrho\bar{\beta}^*\big\|_{2,P(\cdot;\lambda)}+2\etaa\|\mu^*\|_{2}^2+\frac{2C}{\etaa}\omega^2_{\bar m}+\frac{4}{\etaa}\{\varphi^{-1}(1)\}^2\omega^2_{\varsigma}+\frac{4C}{\etaa\land 1}\omega^2_{\bar\eta}\nonumber\\\lesssim &\,\big\| \varrho\bar{\beta}^*\big\|_{2,P(\cdot;\lambda)}+\etaa\|\mu^*\|_{2}^2+\frac{1}{\etaa}(\omega^2_{\bar m}+\omega^2_{\varsigma})+\frac{1}{\etaa\land 1}\omega^2_{\bar\eta}.
\end{align}

To prove part (b),
we use  $\big\| \varrho\bar{\beta}^*_{\mathcal{J^{*C}}}\big\|_{2,P(\cdot;\lambda)}=\|\varrho(\bar{\beta}-\bar{\beta}^*)_{\mathcal{J^{*C}}}\big\|_{2,P(\cdot;\lambda)}$ and $\big\| \varrho\bar{\beta}^*_{\mathcal{J^*}}\big\|_{2,P(\cdot;\lambda)}-\big\|\varrho \bar\beta_{\mathcal{J^*}}\big\|_{2,P(\cdot;\lambda)}\leq \|\varrho(\bar{\beta}-\bar{\beta}^*)_{\mathcal{J^*}}\big\|_{2,P(\cdot;\lambda)}$ and rewrite    \eqref{middlebound1} as
\begin{align}\label{middlebound}
&\,\mathbb{E}\Big\{\bold{\Delta}_{l}(\hat{\mu},\mu^*)+\frac{\etaa-b}{2}\|\hat{\mu}-\mu^*\|^2_2\Big\}-\frac{1}{a'}\mathrm{\bold{D}}_2(\hat{\bar\eta},\bar{\eta}^*)-\frac{1}{2a}\{\mathrm{\bold{D}}_2(\hat {\bar m},\bar m^*)+\mathrm{\bold{D}}_2(\hat {\varsigma},\varsigma^*)\}\nonumber\\\leq &\, \frac{1}{2b'}P_{2,H}\big(\varrho(\hat{\bar{\beta}}-\bar{\beta}^*)_{\mathcal{J}^*};\sqrt{a'b'}A\lambda_0\big)+\frac{1}{2b'}P_{2,H}\big(\varrho(\hat{\bar{\beta}}-\bar{\beta}^*)_{\mathcal{J}^{*C}};\sqrt{a'b'}A\lambda_0\big)\nonumber\\&\, +\|\varrho(\bar{\beta}-\bar{\beta}^*)_{\mathcal{J^*}}\big\|_{2,P(\cdot;\lambda)}-\|\varrho(\bar{\beta}-\bar{\beta}^*)_{\mathcal{J^{*C}}}\big\|_{2,P(\cdot;\lambda)}+aC\omega^2_{\bar m}+2a\{\varphi^{-1}(1)\}^2\omega^2_{\varsigma}\nonumber\\&+\frac{\etaa^2}{2b}\|\mu^*\|_{2}^2+Ca'\omega^2_{\bar\eta}.
\end{align}
The condition \eqref{alterregugeneral} implies
\begin{align}\label{realtt1}
    &\alpha  \mathrm{\bold{D}}_2(\mu,\mu^*)+\vartheta P_{2,H}\big(\varrho(\bar{\beta}-\bar{\beta}^*)_{\mathcal{J}^*};\lambda\big)+\|\varrho(\bar{\beta}-\bar{\beta}^*)_{\mathcal{J}^*}\|_{2,P(\cdot;\lambda)}\nonumber\\\leq&\,  \bold{\Delta}_{l}(\mu,\mu^*)+\|\varrho(\bar{\beta}-\bar{\beta}^*)_{\mathcal{J}^{*C}}\|_{2,P(\cdot;\lambda)}-\vartheta P_{2,H}\big(\varrho(\bar{\beta}-\bar{\beta}^*)_{\mathcal{J}^{*C}};\lambda\big)+K\lambda^2J^*.
\end{align}
With   $b'=1/(2\vartheta)$, we add
 \eqref{middlebound} and  \eqref{realtt1} to get
\begin{align}
 &\,\mathbb{E}\Big\{(\etaa+\alpha-b)\mathrm{\bold{D}}_2(\hat\mu,\mu^*)\Big\}-\frac{1}{a'}\mathrm{\bold{D}}_2(\hat{\bar\eta},\bar{\eta}^*)-\frac{1}{2a}\{\mathrm{\bold{D}}_2(\hat {\bar m},\bar m^*)+\mathrm{\bold{D}}_2(\hat {\varsigma},\varsigma^*)\}\nonumber\\\leq&\,  K\lambda^2J^*+aC\omega^2_{\bar m}+2a\{\varphi^{-1}(1)\}^2\omega^2_{\varsigma}+\frac{\etaa^2}{2b}\|\mu^*\|_{2}^2+Ca'\omega^2_{\bar\eta},
\end{align}
where $\lambda=\sqrt{a'b'}A\lambda_0$ and $A\geq A_0$ with $A_0$ given in the Lemma \ref{basiclemma}. Now, setting $b=(\etaa+\alpha)/4, a'=4/\{(\etaa+\alpha)\land \vartheta\}, a=2/(\etaa+\alpha)$ gives
\begin{align}
   &\,\mathbb{E}\big\{\frac{\etaa+\alpha}{2}\mathrm{\bold{ D}}_2(\hat\mu,\mu^*)\big\}\nonumber\\\leq&\,  \frac{2KA^2}{\{(\etaa+\alpha)\land \vartheta\}\vartheta} {\omega_{\bar\eta}}^2 J^*\log(ep)+\frac{2C\omega^2_{\bar m}}{(\etaa+\alpha)}+\frac{4\{\varphi^{-1}(1)\}^2\omega^2_{\varsigma}}{(\etaa+\alpha)}+\frac{2\etaa^2\|\mu^*\|_{2}^2}{\etaa+\alpha}+\frac{4C\omega^2_{\bar\eta}}{(\etaa+\alpha)\land \vartheta}.\nonumber
\end{align}
The proof is complete.
\end{proof}

\begin{lemma}\label{basiclemma}
Let  $\epsilon=[\epsilon_i]$ be an $n$-dimensional random vector (not necessarily mean centered or having independent components) satisfying  $\|\epsilon\|_{\psi_2}\leq \omega$. Suppose that $\bar{X}\in \mathbb{R}^{n\times 2p}$ satisfies $\|\bar{X}\|_2\leq \varrho$. Let $\lambda_0=\omega\sqrt{\log(ep)}$. Then there exist universal constants $A_0,C,c>0$ such that for any $a\geq 2b>0$, $A_1\geq A_0$, the following event \begin{align}\label{basicboundlemma}
    \sup_{\bar{\beta} \in \mathbb{R}^{2p}} 2\langle \epsilon,\bar{X}\bar{\beta}\rangle -\frac{1}{a}\|\bar{X}\bar{\beta}\|^2_2-\frac{1}{b}P_{2,H}(\varrho\bar{\beta};\sqrt{ab}A_1\lambda_0)\geq a\omega^2 t
\end{align}
occurs with probability at most $C\exp(-ct)p^{-cA^2_1}$.
\end{lemma}

\begin{proof}
Let $P_0 = (\lambda^2/2)1_{t\ne 0}$. Define $l_H(\bar{\beta})=2\langle\epsilon,\bar{X}\bar{\beta}\rangle-\frac{1}{a}\|\bar{X}\bar{\beta}\|^2_2-\frac{1}{b}P_{2,H}(\bar{\beta};\sqrt{ab}A_0\lambda_0),l_0(\bar{\beta})=2\langle\epsilon,\bar{X}\bar{\beta}\rangle-\frac{1}{a}\|\bar{X}\bar{\beta}\|^2_2-\frac{1}{b}P_{2,0}(\bar{\beta};\sqrt{ab}A_0\lambda_0)$. Introduce two events $\varepsilon_H=\{\sup_{\bar{\beta}\in \Gamma}l_H(\bar{\beta})\geq at\omega^2\}$, and $\varepsilon_0=\{\sup_{\bar{\beta}\in \Gamma}l_0(\bar{\beta})\geq at\omega^2\}$

First, we use an optimization technique to prove that $\varepsilon_H=\varepsilon_0$. Since $P_0\geq P_H$,  $P_{2,0}\geq P_{2,H}$ and thus $\varepsilon_0\subset \varepsilon_H$. The occurrence of $\varepsilon_H$ implies that $l_H(\bar{\beta}^o)\geq a t\omega^2$ for any $\bar{\beta}^o$ defined by
\begin{align}
    \bar{\beta}^o\in \argmin_{\bar{\beta}\in \mathbb{R}^{2p}} \frac{1}{a}\|\bar{X}\bar{\beta}\|^2_2-2\langle\epsilon,\bar{X}\bar{\beta}\rangle+\frac{1}{b}P_{2,H}(\bar{\beta};\sqrt{ab}A_0\lambda_0).
\end{align}
From    Lemma 5 in \cite{she2016finite}, under
  $\|\bar{X}\|_2\leq 1$,  there exists a globally optimal solution $\bar{\beta}^0$ to the problem
\begin{align}
    \min_{\bar{\beta}\in \mathbb{R}^{2p}}\frac{1}{2}\|y-\bar{X}\bar{\beta}\|^2_2+P_{2,H}(\bar{\beta};\lambda),
\end{align}
such that for any $j=1,\ldots,p$, either $\bar{\beta}^0_j=0$ or $\|\bar{\beta}^0_j\|_2\geq\lambda$. Therefore, with $a\geq 2b>0$, there exists at least one global minimizer $\bar{\beta}^{oo}$ satisfying $P_{2,H}(\bar{\beta}^{oo};\sqrt{ab}A_1\lambda_0)=P_{2,0}(\bar{\beta}^{oo};\sqrt{ab}A_1\lambda_0)$ and thus $l_H(\bar{\beta}^{oo})=l_0(\bar{\beta}^{oo})$. This means $\sup_{\bar{\beta}\in \mathbb{R}^{2p}} l_0(\bar{\beta})\geq l_0(\bar{\beta}^{oo})=l_{H}(\bar{\beta}^{oo})\geq a t\omega^2$, and so  $\varepsilon_H\subset\varepsilon_0$. It suffices to prove $\mathbb{P}(\varepsilon_0)\leq C\exp(-ct)p^{-cA^2_1}$.

  Next, we   use Lemma \ref{supboundlm2} to bound the tail probability of $R$ defined by
 \begin{align}\label{rdefappend}
     R= \sup_{1\leq J\leq p} \sup_{\bar\beta\in \Gamma_J}\big\{\langle\epsilon,\bar{X}\bar{\beta}\rangle-\frac{1}{2a}\|\bar{X}\bar{\beta}\|^2_2-\frac{1}{2b}P_{2,0}(\bar\beta;\sqrt{ab}A_1\lambda_0)\big\},
 \end{align}
where $P_{2,0}(\bar\beta;\lambda_0)=P_{0}(J;\lambda_0)=(1/2)J\lambda^2_0$ for $\bar\beta \in \Gamma_J$ and $\Gamma_J=\{\bar{\beta}\in\mathbb{R}^{2p}:J(\bar{\beta})=J\}$  (in the trivial case of $J=0$, the quantity inside the braces is 0).

 By  a scaling argument, for any $a>0$,
\begin{align}\label{ineqdr2}
  \langle\epsilon,\bar{X}\bar{\beta}/\|\bar{X}\bar{\beta}\|_2\rangle\|\bar{X}\bar{\beta}\|_2\leq \frac{a}{2}\sup_{\alpha\in \Gamma_J}(\langle\epsilon,\alpha\rangle)^2+\frac{1}{2a}\|\bar{X}\bar{\beta}\|^2_2.
\end{align}
  Applying Lemma \ref{supboundlm2} with $G=2,q=2$ results in
 \begin{align}
   \mathbb{P}\big(  \sup_{\bar\beta\in \Gamma_J}\langle\epsilon,\bar{X}\bar{\beta}\rangle-\frac{1}{2a}\|\bar{X}\bar{\beta}\|^2_2-\frac{a}{2}LJ\log (ep)\omega^2\geq \frac{a}{2}\omega^2t\big)\leq C\exp(-ct),
\end{align}
or  \begin{align}
   \mathbb{P}\big(  \sup_{\bar\beta\in \Gamma_J}\langle\epsilon,\bar{X}\bar{\beta}\rangle- \frac{1}{2a}\|\bar{X}\bar{\beta}\|^2_2- aLP_{2,0}(\bar\beta;\lambda_0)\geq a\omega^2t\big)\leq C\exp(-ct).
\end{align}
 Set $A_1\geq \sqrt{2L}$. Noticing   that (i) $(A^2_1/2)P_0(J;\lambda_0)\geq LP_0(J;\lambda_0)+cA^2_1  P_0(J;\lambda_0)$ for some $c>0$, and (ii) $J\log(ep)\geq \log p+J$ for any $J\geq 1$, we get
\begin{align}\label{unionboundJ}
   &\, \mathbb{P}\big(R\geq a\omega^2t\big)\nonumber\\\leq &\,\sum^p_{J=1}\mathbb{P}\big(  \sup_{\bar\beta\in \Gamma_J}\langle\epsilon,\bar{X}\bar{\beta}\rangle-\frac{1}{2a}\|\bar{X}\bar{\beta}\|^2_2-\frac{1}{2b}P_{2,0}(\bar\beta;\sqrt{ab}A_1\lambda_0)\geq a\omega^2t\big)\nonumber\\= &\,\sum^p_{J=1}\mathbb{P}\big(  \sup_{\bar\beta\in \Gamma_J}\langle\epsilon,\bar{X}\bar{\beta}\rangle-\frac{1}{2a}\|\bar{X}\bar{\beta}\|^2_2-\frac{1}{4}aA_1^2J\log (ep)\omega^2_{\bar\eta}\geq a\omega^2t\big)\nonumber\\\leq &\,\sum^p_{J=1}\mathbb{P}\big(  \sup_{\bar\beta\in \Gamma_J}\langle\epsilon,\bar{X}\bar{\beta}\rangle-\frac{1}{2a}\|\bar{X}\bar{\beta}\|^2_2-\frac{a}{2}LJ\log (ep)\omega^2\geq a\omega^2t+caA^2_1 J\log (ep)\omega^2\big)\nonumber\\\leq &\,\sum^p_{J=1}C\exp(-ct)\sum^p_{J=1}\exp\{-cA_1^2(J+\log p)\}\nonumber\\\leq &\,C\exp(-ct)p^{-cA^2_1},
\end{align}
where the last inequality is due to the sum of geometric sequence.
\end{proof}

\begin{lemma}\label{supboundlm2}

Given a  matrix $X\in \mathbb{R}^{n\times pG}$ with   a block form of $X=[X_1,\ldots,X_p]$ with $X_j\in\mathbb{R}^{n\times G},1\leq j\leq p$,
 let $X_{\mathcal{J}}$ denote the submatrix formed by the column blocks of $X$ indexed by $\mathcal{J}$. Define $\Gamma_{\mathcal{J}}=\{\alpha\in \mathbb{R}^{pG}: \|\alpha\|_2\leq 1, \alpha\in \mathcal{R}(X_{\mathcal{J}})\}$ and  $\Gamma_{J}=\bigcup_{|\mathcal{J}|=J}\Gamma_{\mathcal{J}}$, where $1\leq J\leq p$.
 Let $\epsilon=[\epsilon_i]$ be an $n$-dimensional random vector. (i) Assume $\epsilon$ satisfies $\|\epsilon\|_{\psi_q}\leq \omega$ for some $q\geq 1$. Then for any $t>0$,
\begin{align}
    &  \mathbb{P}\big( \sup_{\alpha\in \Gamma_{\mathcal{J}}}(\langle\epsilon,\alpha\rangle)^2-LJG\omega^2-L(JG)^{\frac{2}{q}}\omega^2\geq \omega^2t\big)\leq C\exp\big(-ct^{\frac{q}{2}}\big), \label{lemmaass1medb}\\
    &  \mathbb{P}\big( \sup_{\alpha\in \Gamma_J}(\langle\epsilon,\alpha\rangle)^2-L\big[JG+(JG)^{\frac{2}{q}}+(2J)^{\frac{2}{q}}\{\log (ep)\}^{\frac{2}{q}}\big]\omega^2\geq \omega^2t\big)\leq C\exp\big(-ct^{\frac{q}{2}}\big),\label{lemmaass1finalb}
\end{align}
where   $C$ is a universal constant  and  $c,L$ are constants   depending on $q$ only.
(ii) Assume that  $\epsilon_1,\ldots,\epsilon_n$  are independent, centered,  and  $\|\epsilon_i\|_{\psi_q}\leq \omega$ for some $q\in(0,2]$. Then \eqref{lemmaass1medb} and \eqref{lemmaass1finalb} are replaced by  \begin{align}
&     \mathbb{P}\Big( \sup_{\alpha\in \Gamma_{\mathcal{J}}}(\langle\epsilon,\alpha\rangle)^2-LJG\omega^2\geq \omega^2t\Big)\leq C\exp\big(-ct^{\frac{q}{2}}\big), \label{lemmaunionbound1}
\\
 &   \mathbb{P}\Big( \sup_{\alpha\in \Gamma_J}(\langle\epsilon,\alpha\rangle)^2-L\big[JG+(2J)^{\frac{2}{q}}\{\log(ep)\}^\frac{2}{q}\big]\omega^2\geq \omega^2t\Big)\leq C\exp\big(-ct^{\frac{q}{2}}\big).\label{lemmameanzerobound}
\end{align}
\end{lemma}

\begin{proof}

First, we show  \eqref{lemmaass1medb} under the assumption that $\|\epsilon\|_{\psi_q}\leq \omega$ for some $q\geq 1$.
Define a  centered random vector $\epsilon_c=\epsilon-M$ with  $M=\mathbb{E}[\epsilon]$, then
\begin{align}
   \sup_{\alpha\in \Gamma_{\mathcal{J}}} \langle\epsilon,\alpha\rangle\leq\sup_{\alpha\in \Gamma_{\mathcal{J}}}\langle \epsilon_c,\alpha\rangle+\sup_{\alpha\in \Gamma_{\mathcal{J}}}\langle P_{ X_{\mathcal{J}}}M,\alpha\rangle\leq \sup_{\alpha\in \Gamma_\mathcal{J}}\langle\epsilon_c,\alpha\rangle+\|U^TM\|_2,
\end{align}
 where  $P_{ X_{\mathcal{J}}}=UU^T$ is the orthogonal projection matrix onto $\mathcal{R}(X_{\mathcal{J}})$ with $\{U_1,\ldots,U_{JG}\}$ as  an orthonormal basis and $\|P_{ X_{\mathcal{J}}}\|_2=1$. We claim that $\|U^TM\|^2_2=\sum^{JG}_{i=1}(\mathbb{E}[U^T_i\epsilon])^2\leq CJG\omega^2$. In fact,
because $\|U^T_i\epsilon\|_{\psi_q}\leq \omega$,
\begin{align}\label{probtomean}
\mathbb{P}\big(|U^T_i\epsilon|>t\omega\big)\leq \int^{+\infty}_{0}\frac{1}{\psi_q(\frac{t\omega}{\omega})}\rd t=\int^{+\infty}_{0}\exp(-t^q)\rd t<\infty,
\end{align}
and so $\mathbb{E}|U^T_i\epsilon|\leq C\omega,\,1\leq i\leq J$. 

By definition, $\{\langle\epsilon_c,\alpha\rangle:\alpha\in \Gamma_{\mathcal{J}}\}$ is a (centered)      $\psi_q$-process. The induced metric on $\Gamma_{\mathcal{J}}$ is: $d(\alpha,\alpha')=\omega\|\alpha-\alpha'\|_2$.
To bound the metric entropy $\log (\mathcal{N}(\varepsilon,\Gamma_{\mathcal{J}},d)))$, where the $\mathcal{N}(\varepsilon,\Gamma_{\mathcal{J}},d))$ is the smallest cardinality of an $\varepsilon$-net that covers $\Gamma_{\mathcal{J}}$ under the metric $d$, we apply a standard volume argument to get
\begin{align}
  \log (\mathcal{N}(\varepsilon,\Gamma_{\mathcal{J}},d)))\leq \log  (\frac{C\omega}{\varepsilon})^{JG}=JG\log(C\omega/\varepsilon).
\end{align}
By Theorem 5.36 in \cite{wainwright2019high},  we get
\begin{align}
         \mathbb{P}\Big( \sup_{\alpha,\alpha'\in \Gamma_{\mathcal{J}}}|\langle\epsilon_c,\alpha-\alpha'\rangle|-L\int^{D}_{0}\psi^{-1}_q(\mathcal{N}(\varepsilon,\Gamma_{\mathcal{J}},d)) \rd \varepsilon \geq Lt\Big)\leq 2\exp\big(-\frac{t^{q}}{D^q}\big),
\end{align}
for any $t>0$, where $D=\sup_{\alpha,\alpha'\in\Gamma_{\mathcal{J}}}d(\alpha,\alpha')=2\omega$. Equivalently,
\begin{align}
     \mathbb{P}\Big( \sup_{\alpha,\alpha'\in \Gamma_{\mathcal{J}}}|\langle\epsilon_c,\alpha-\alpha'\rangle|-L\int^{2\omega}_{0}\psi^{-1}_q(\mathcal{N}(\varepsilon,\Gamma_{\mathcal{J}},d)) \rd \varepsilon \geq \omega t\Big)\leq C\exp\big(-ct^{q}\big),
\end{align}
which implies
\begin{align}\label{lemma3meanzeromedinq}
   \mathbb{P}\big( \sup_{\alpha\in \Gamma_\mathcal{J}}\langle\epsilon_c,\alpha\rangle-L\int^{2\omega}_{0}\psi^{-1}_q(\mathcal{N}(\varepsilon,\Gamma_{\mathcal{J}},d)) \rd \varepsilon \geq \omega t\Big)\leq C\exp\big(-ct^{q}\big).
\end{align}

Based on \eqref{probtomean} and \eqref{lemma3meanzeromedinq}, we obtain
\begin{align}
    \mathbb{P}\big( \sup_{\alpha\in \Gamma_\mathcal{J}}\langle\epsilon,\alpha\rangle-L\sqrt{JG}\omega-L\int^{2\omega}_{0}\psi^{-1}_q(\mathcal{N}(\varepsilon,\Gamma_{\mathcal{J}},d)) \rd \varepsilon\geq \omega t\big)\leq C\exp\big(-ct^{q}\big),
\end{align}
or
\begin{align}
    \mathbb{P}\big( \sup_{\alpha\in \Gamma_\mathcal{J}}\langle\epsilon,\alpha\rangle-L\sqrt{JG}\omega-L(JG)^{\frac{1}{q}}\omega\geq \omega t\big)\leq C\exp\big(-ct^{q}\big).
\end{align}
Therefore, \begin{align}
    \mathbb{P}\Big( \sup_{\alpha\in \Gamma_{\mathcal{J}}}(\langle\epsilon,\alpha\rangle)^2-LJG\omega^2-L(JG)^{\frac{2}{q}}\omega^2\geq \omega^2t^2\Big)\leq C\exp\big(-ct^{q}\big),
\end{align}
and letting $s=t^2$ gives   \eqref{lemmaass1medb}. With a union bound, we also obtain
\begin{align}
        \mathbb{P}\big( \sup_{\alpha\in \Gamma_J}(\langle\epsilon,\alpha\rangle)^2-LJG\omega^2-L(JG)^{\frac{2}{q}}\omega^2\geq \omega^2t\big)\leq &\, {p\choose J}C\exp\big(-ct^{\frac{q}{2}}\big)\nonumber\\\leq &\, C\exp\big\{-ct^{\frac{q}{2}}+J\log(ep)\big\},
\end{align}
from which it follows
\begin{align}
    \mathbb{P}\Big( \sup_{\alpha\in \Gamma_J}(\langle\epsilon,\alpha\rangle)^2-LJG\omega^2-L(JG)^{\frac{2}{q}}\omega^2-L2^{\frac{2}{q}-1}\{J\log (ep)\}^{\frac{2}{q}}\omega^2\geq \omega^2t\Big)\leq C\exp\big(-ct^{\frac{q}{2}}\big).
\end{align}


Next, we prove \eqref{lemmaunionbound1} assuming  $\epsilon_1,\ldots,\epsilon_n$  are independent, centered,  and $\|\epsilon_i\|_{\psi_q}\leq \omega$ for some $0<q\le 2$.
 By H\"{o}lder's  inequality,
\begin{align}
   \sup_{\alpha\in \Gamma_{\mathcal{J}}} \langle\epsilon,\alpha\rangle=\sup_{\alpha\in \Gamma_{\mathcal{J}}}\langle P_{ X_{\mathcal{J}}}\epsilon,\alpha\rangle\leq \|U^T\epsilon\|_2,
\end{align}
where $P_{ X_{\mathcal{J}}}=UU^T$ is the orthogonal projection matrix onto $\mathcal{R}(X_{\mathcal{J}})$ with $\{U_1,\ldots,U_{JG}\}$ as  an orthonormal basis and $\|P_{ X_{\mathcal{J}}}\|_2=1$.
It remains to get a tail bound for $\|U^T\epsilon\|^2_2$.

 Noticing that    (i) $\mathbb{E}[\epsilon]=0$, (ii)   $\|\epsilon_i\|_{\psi_q}\leq \omega$ implies    $\mathbb{E}[\epsilon^2_i]\leq C\omega^2$, (iii) $\|UU^T\|_2=1$, and (iv) $Tr(UU^T)=JG$, we apply a generalized Hanson-Wright inequality (\cite{sambale2020some}, Theorem 2.1) to obtain\begin{align}\label{lemma3mean0medianineq1}
    \mathbb{P}\big( \|U^T\epsilon\|^2_2-LJG\omega^2\geq \omega^2t\big)\leq C\exp\big(-ct^{\frac{q}{2}}\big),
\end{align}
where $C$ can be $2$  and  $c,L$ are constants   depending on $q$ only. This results in \begin{align}\label{lemmamidleboundresult}
    \mathbb{P}\Big( \sup_{\alpha\in \Gamma_{\mathcal{J}}}(\langle\epsilon,\alpha\rangle)^2-LJG\omega^2\geq \omega^2t\Big)\leq C\exp\big(-ct^{\frac{q}{2}}\big).
\end{align}
Finally, we prove \eqref{lemmameanzerobound}. Applying the union bound on  \eqref{lemmamidleboundresult} yields
\begin{align}\label{unionbondfirst}
     \mathbb{P}\big( \sup_{\alpha\in \Gamma_J}(\langle\epsilon,\alpha\rangle)^2-LJG\omega^2\geq \omega^2t\big)\leq {p\choose J}C\exp\big(-ct^{\frac{q}{2}}\big)\leq C\exp\big\{-ct^{\frac{q}{2}}+J\log(ep)\big\}.
\end{align}
Let $t'=t^{\frac{q}{2}}-(1/c)J\log(ep)$. Then \eqref{unionbondfirst} can be rewritten as
\begin{align}\label{supineq1}
     \mathbb{P}\big( \sup_{\alpha\in \Gamma_J}(\langle\epsilon,\alpha\rangle)^2-LJG\omega^2\geq \omega^2\{t'+\frac{1}{c}J\log (ep)\}^{\frac{2}{q}}\big)\leq C\exp(-ct').
\end{align}
Given $0<q\leq 2$, by the convexity of $x^\frac{2}{q}$ on $\mathbb{R}_+$, we have  $(a+b)^\frac{2}{q}\leq 2^{\frac{2}{q}-1}(a^\frac{2}{q}+b^\frac{2}{q})$ for any $a,b\geq0$, and thus
\eqref{supineq1} becomes
\begin{align}
     \mathbb{P}\big( \sup_{\alpha\in \Gamma_J}(\langle\epsilon,\alpha\rangle)^2-LJG\omega^2-L2^{\frac{2}{q}-1}\{J\log (ep)\}^{\frac{2}{q}}\omega^2\geq 2^{\frac{2}{q}-1}\omega^2(t')^\frac{2}{q}\big)\leq C\exp(-ct'),
\end{align}
or equivalently,
\begin{align}
     \mathbb{P}\big( \sup_{\alpha\in \Gamma_J}(\langle\epsilon,\alpha\rangle)^2-LJG\omega^2-L2^{\frac{2}{q}-1}\{J\log (ep)\}^{\frac{2}{q}}\omega^2\geq \omega^2s\big)\leq C\exp\big(-cs^{\frac{q}{2}}\big).
\end{align}
The proof is now complete.
\end{proof}

\section{Further Extensions} \label{app:furtherext}
We showcase two {nonparametric} statistical applications of skewed pivot-blend, one based on data ranks and the other based on kernels.

\paragraph{Skewed rank-based estimation.}
Recently, there has been a lot of interest in rank-based  nonparametric estimation methods that minimize    $\sqrt{12}\sum^n_{i=1}r_i \{\mbox{R}(r_i)/(n+1) -\frac{1}{2} \}$ \citep{Jaeckel1972},
where $\mbox{R}(r_i)$ denotes the rank of $r_i$ among $r_1,\ldots,r_n$, or equivalently, the $\ell_1$ loss on the spread of residuals: $
\frac{1}{ n(n-1) }\sum_{i\neq j}|r_i-r_j|
$ \citep{hettmansperger1978statistical,hettmansperger1998robust,Wang2020}. 
Applying the technique of U-statistics (with kernel $h(r_i, r_j)= |r_i-r_j|$) can show that the criterion results in estimators with desired asymptotic properties.

The adoption of the symmetric (and relatively robust) $\ell_1$-loss function heavily depends on the assumption that $r_i$ follows an i.i.d. distribution (when assessed against the statistical truth). In turn, it implies that the distribution of  differences  $r_i - r_j$ is symmetrical, as in the case of a double-exponential.

 Nonetheless,   the assumption may not  hold in real-world applications \textbf{beyond} the i.i.d. setting. Such deviations can  result in extra skewness of $|r_i - r_j|$ that cannot be captured by an exponential (or half-normal) distribution. We introduce a skewed criterion that operates on the \textit{absolute} differences $|r_i - r_j|$ and $m, \sigma, \nu>0$:
\begin{align}\label{skewrloss}
&\,\sum_{i\neq j}\Big\{ \big(\frac{|r_i-r_j|-m}{\sigma} + m\big) \,1_{m(1-\sigma)\le |r_i-r_j|\leq m}+ \big(\frac{|r_i-r_j|-m}{\nu} +m\big)\,1_{|r_i-r_j|>m}\Big\}\nonumber\\&\,\quad\quad+ n(n-1)\log \{ \sigma [1-\exp(-m)  ]+\nu \exp(-m) \}.
\end{align}
\eqref{skewrloss} results from applying  skewed pivot-blend  to the exponential density (instead of   the double-exponential density). A regularization term (such as an $\ell_1$ penalty) can be incorporated to capture structural parsimony. 

\paragraph{Kernel-assisted nonparametric skew estimation.}
Assume  $y-X\beta^*\sim \mbox{SP}^{(\phi^*)}(\sigma^{*},\nu^{*},m^{*})$, where the functional form of the density  $\phi^{*}$ is also unknown.    We can employ  a backward-forward scheme for    {nonparametric} estimation, along with  \textit{explicit} capture of skewness that  aligns with the theme of this paper.

To estimate   $\sigma^{*},\nu^{*},m^{*}$,       kernel may be employed to approximate    $\phi^{*}$, but the data  follow  $\mbox{SP}^{(\phi^*)} $. We can address this by using   backward pivot-blend. Holding $\beta, \sigma, \nu, m$ constant for now, and referring to the definitions of $\tilde{r}_i$ ($1\le i \le n$), $L(m)$, and $R(m)$ in Remark \ref{conversionremark},  introduce a kernel density estimator with appropriate weights to estimate $\phi^*$ based on the transformed residuals $\tilde r_i$:
 \begin{align}\label{rform}
 \begin{split}
      \phi_K(t;h)=\,\frac{1}{h}\sum^{n}_{i=1}&\Big\{\frac{\nu}{L(m)\nu+R(m)\sigma}K_h\big(\frac{r_i-m}{\sigma}+m-t\big)1_{r_i\leq m}\\&+\frac{\sigma}{L(m)\nu+R(m)\sigma}K_h\big(\frac{r_i-m}{\nu}+m-t\big)1_{r_i>m}\Big\}.
 \end{split}
 \end{align}
where   $K_h(t)=K\big(t/h\big)$, the kernel function $K(t)$ can be  any continuous symmetric function with $\int^{\infty}_{-\infty} K(t),dt=1$, and $h>0$ represents the bandwidth parameter.
The kernel approach in conjunction with (forward) skewed pivot-blend gives rise to a nonparametric skew-estimation problem, which warrants further investigation.  

\section{More Experiments}
\label{appendixb}
In this part, we  consider some  densities which may not be unimodal or symmetric (as shown in   Figure \ref{fig:ex3ex4ex5}). The experimental setup is the same as in Section  \ref{subsec:basicsimu}.

\begin{figure}
\begin{subfigure}{0.32\textwidth}
  \includegraphics[width=\textwidth, height=3.6cm]{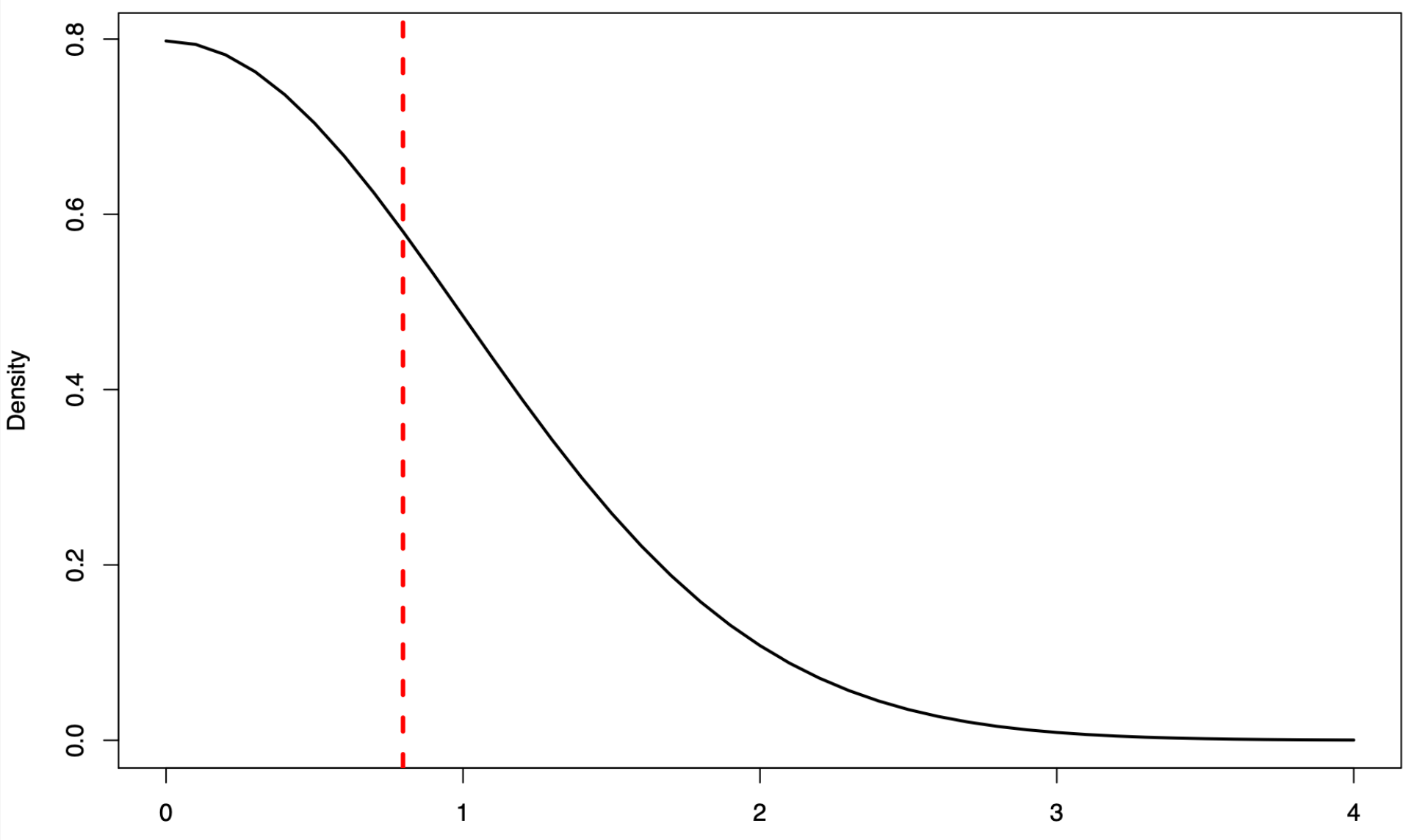}
  \caption{\scriptsize  $\phi$: half-normal,  $m^*$: mean}
  \label{halfnormal}
\end{subfigure}
\hfill
\begin{subfigure}{0.32\textwidth}

  \includegraphics[width=\textwidth, height=3.6cm]{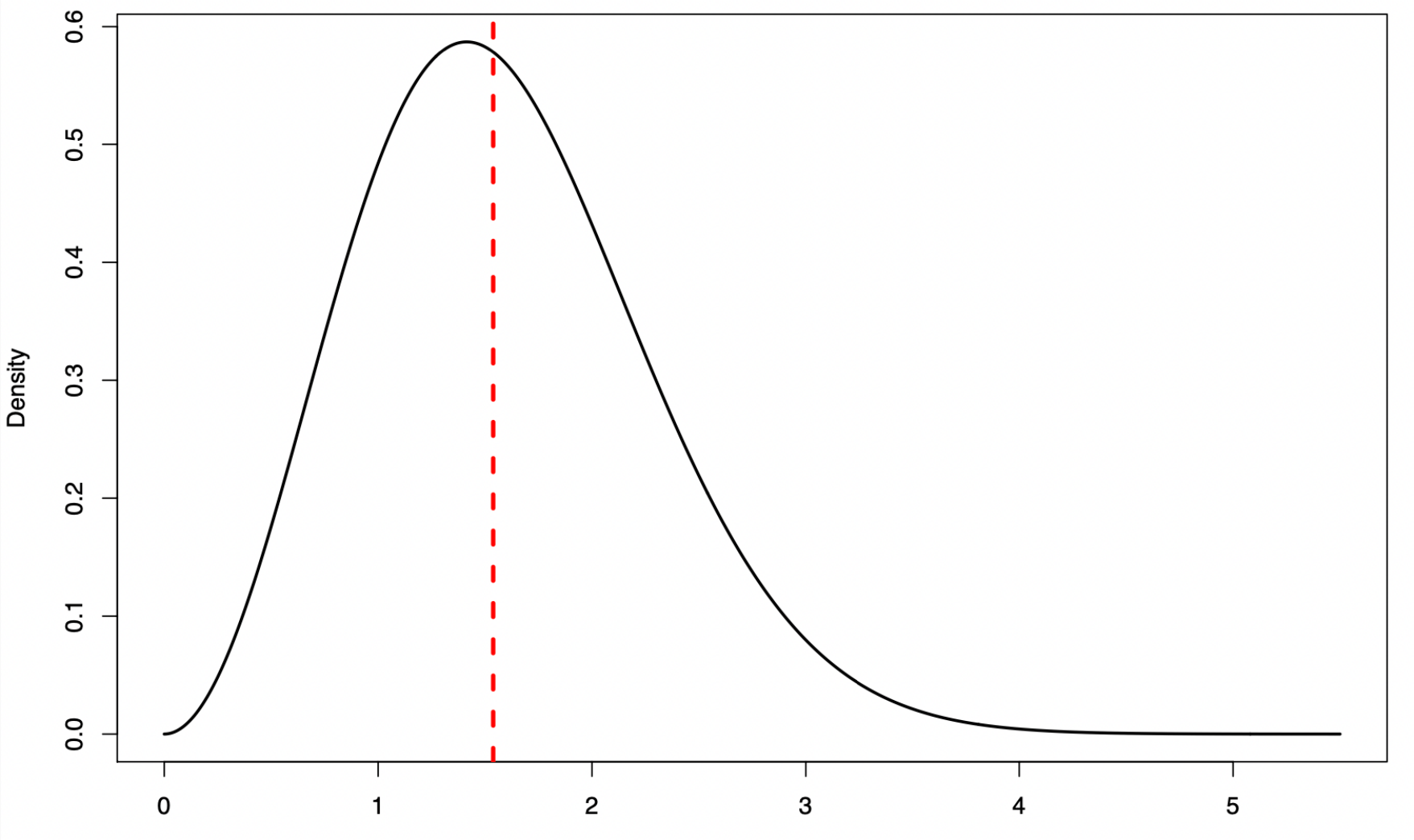}
  \caption{\scriptsize  $\phi$: Maxwell,   $m^*$: median}
  \label{maxwell}
\end{subfigure}
\hfill
\begin{subfigure}{0.32\textwidth}

  \includegraphics[width=\textwidth, height=3.6cm]{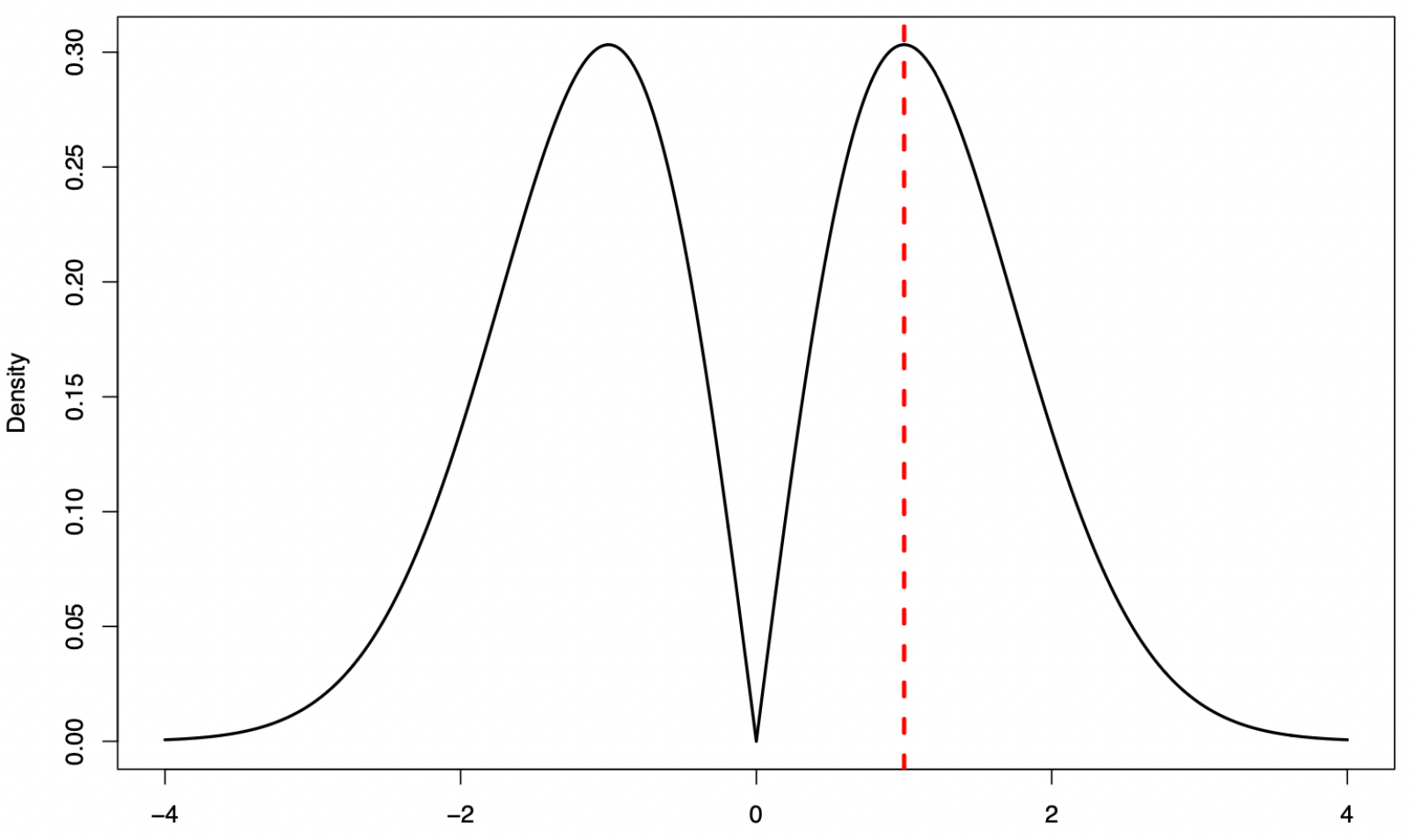}
  \caption{\scriptsize $\phi$: double Rayleigh,   $m^*$:  right mode}
  \label{doublerayleigh}
\end{subfigure}

\caption{\footnotesize Plots of  $\phi$ for \textbf{Ex 3--Ex 5}. The choices  of $m^*$ are indicated in red dashed lines.}
\label{fig:ex3ex4ex5}
\end{figure}

\textbf{Ex 3.} (Half-normal with the   mean as the pivotal point): Let $\phi$ be the half-normal density (with scale 1).  We set $n=700 $,   $\kappa=0.1$, $m^*$ as the mean,   $\nu^*=0.1$ and  $\sigma^*=0.2,0.3,0.4$.

\textbf{Ex 4.} (Maxwell-Boltzmann with the  median as the pivotal point): Here, $\phi$ is the Maxwell-Boltzmann density which is widely used in statistical mechanics  \citep{huang2008statistical}. We set   $n=500$,   $\kappa=0.2$, $m^*$ as the median,  $\nu^*=0.1$ and  $\sigma^*=0.4,0.5,0.6$.

\textbf{Ex 5.} (Double Rayleigh with $m^*$ as the  right mode): Let $\phi$ be the double Rayleigh density (with scale 1),    $n=500$, $\kappa=0.1$,   $m^*$ be the  right mode,     $\nu^*=0.1$ and  $\sigma^*=0.2,0.3,0.4$.

\begin{table}[!ht]
\scriptsize
\centering
\resizebox{\columnwidth}{!}{
\begin{tabular}{ccccccccccccc}
\hline\hline\noalign{\vskip 1mm}
     \multicolumn{13}{c}{Skewed half-normal}\\
      \cline{3-13}
          &       &   \multicolumn{3}{c}{$\sigma^*/\nu^*=2$ }                    &  & \multicolumn{3}{c}{$\sigma^*/\nu^*=3$ }                     &  & \multicolumn{3}{c}{$\sigma^*/\nu^*=4$ }
     \\
     \cline{3-5} \cline{7-9} \cline{11-13}

 &    & Err($\beta$)  & Err($\sigma$) & Err($\nu$) &   & Err($\beta$)  & Err($\sigma$) & Err($\nu$) &    & Err($\beta$)  & Err($\sigma$) & Err($\nu$) \\
 \hline
QR$^*$ &  & 0.80  &---  & --- &  & 0.80 & ---   & --- &          &  0.80 &   ---    &  ---    \\
BQR$^*$ &  &0.80  & 0.49  & 0.62&  & 0.80 & 0.51   & 0.64 &          &  0.80 &   0.53   &  0.65   \\
AME  &  &0.64  & 0.99  & 1.22 &  & 0.56 & 1.0   & 2.02 &          &  0.48 &   1.0    &  2.82   \\
ZQR   &  &0.65  &0.97  & 0.09 &  & 0.59 & 0.95   & 0.37 &          &  0.52 &   0.95    &  0.67  \\
 SPEUS  &  &0.14  &0.11  & 0.1 &  & 0.08 & 0.08   & 0.13 &          &  0.06 &  0.09    &  0.17  \\
  \hline\noalign{\vskip 1.5mm}
 \multicolumn{13}{c}{Skewed Maxwell-Boltzmann}\\   \cline{3-13}
          &       &   \multicolumn{3}{c}{$\sigma^*/\nu^*=4$ }                    &  & \multicolumn{3}{c}{$\sigma^*/\nu^*=5$ }                     &  & \multicolumn{3}{c}{$\sigma^*/\nu^*=6$ }
     \\
     \cline{3-5} \cline{7-9} \cline{11-13}

 &    & Err($\beta$)  & Err($\sigma$) & Err($\nu$) &   & Err($\beta$)  & Err($\sigma$) & Err($\nu$) &    & Err($\beta$)  & Err($\sigma$) & Err($\nu$) \\
  \hline
QR$^*$  &  &1.59  &---  & --- &  & 1.59 & ---   & --- &          &  1.59 &   ---    &  ---    \\
BQR$^*$  &  &1.59  &0.41 & 0.74 &  & 1.59 & 0.43   & 0.75 &          &  1.59 &   0.44    &  0.75   \\
AME   &  &1.47  &0.22  & 0.78 &  & 1.46 & 0.19   & 0.84 &          &  1.46 &  0.17   &  0.91   \\
ZQR  &  &1.49  &0.58  & 0.30 &  &1.48 & 0.58   & 0.32 &          &  1.48 &   0.57    &  0.34   \\
 SPEUS  &  &0.18  &0.14  & 0.20 &  & 0.17 & 0.13  & 0.21 &          &  0.16 &   0.12    & 0.30  \\
 \hline\noalign{\vskip 1.5mm}
 \multicolumn{13}{c}{Skewed double Rayleigh}\\   \cline{3-13}
          &       &   \multicolumn{3}{c}{$\sigma^*/\nu^*=2$ }                    &  & \multicolumn{3}{c}{$\sigma^*/\nu^*=3$ }                     &  & \multicolumn{3}{c}{$\sigma^*/\nu^*=4$ }
     \\
     \cline{3-5} \cline{7-9} \cline{11-13}

 &    & Err($\beta$)  & Err($\sigma$) & Err($\nu$) &   & Err($\beta$)  & Err($\sigma$) & Err($\nu$) &    & Err($\beta$)  & Err($\sigma$) & Err($\nu$) \\
  \hline
QR$^*$  &  &1.0  &---  & --- &  & 1.0 & ---   & --- &          &  1.0 &   ---    &  ---    \\
BQR$^*$   &  &1.0  &0.68  & 0.27 &  & 1.0 & 0.70   & 0.26 &          &  1.0 &  0.71    & 0.26     \\
AME    &  &0.87  &1.78  & 1.76 &  &0.88 & 2.02   & 1.9 &          &  0.89 &   2.15   &  1.92   \\
ZQR  &  &0.67  &0.50  & 1.67 &  & 0.50 & 0.48   & 2.69 &          &  0.32 &   0.45   &  3.78    \\
 SPEUS  &  &0.22  &0.03  & 0.18 &  & 0.15 & 0.03   & 0.24 &          &  0.10 &  0.03  & 0.26 \\
        \hline\hline
\end{tabular}
}
\caption{\footnotesize Performance comparison for skewed half-normal  with pivotal point at the  mean,  skewed Maxwell with pivotal point at the median, and skewed double Rayleigh with pivotal point at the right mode ({Ex 3}--{Ex 5})\label{halfnormalmaxelldoublerayleigh}}
\end{table}

 Table \ref{halfnormalmaxelldoublerayleigh} illustrates significant issues with conventional methods in capturing skewness, even when asymmetry is centered around the median, mean, or mode. Specifically, QR$^*$ and BQR$^*$ failed to accurately recover the true $\beta^*$ despite using the true quantile as input.  Moreover, in all three cases, AME or AME provided at least one misleading estimate of the scales. In contrast, our proposed method demonstrated remarkable performance, with the associated $\beta$-error being at most $1/3$ of that of the other methods. \\


We also conducted  experiments  on the air pollution data of
Leeds from 1994 to 1998
\citep{heffernan2004conditional,southworth2020package}. The dataset contain
578 measurements of the daily maximum levels of  ozone  (O3), nitrogen dioxide (NO2), nitrogen oxide (NO), sulfur dioxide (SO2) and particulate matter  (PM10). We forecasted PM10 levels for the summer months (April to July) using  other air pollutants. We considered the Gumbel model and its skewed pivot-blend enhancement SPEUS, in addition to ZQR, AME, and ESN.
The p-values from the Kolmogorov-Smirnov tests for these models were  8e-2, 0.71,  9e-2,  0.20, and 8e-8,  respectively.  The results may appear surprising, considering the popularity of the skewed Gumbel distribution  for modeling such extreme values \citep{boldi2007mixture}. Figure \ref{airqqplot-ori} presents   the residual Q-Q plots, demonstrating  that our skew-reinforced method significantly improves data fit compared to the classical Gumbel model.

\begin{figure}[htp!]
  \includegraphics[width=.45\textwidth]{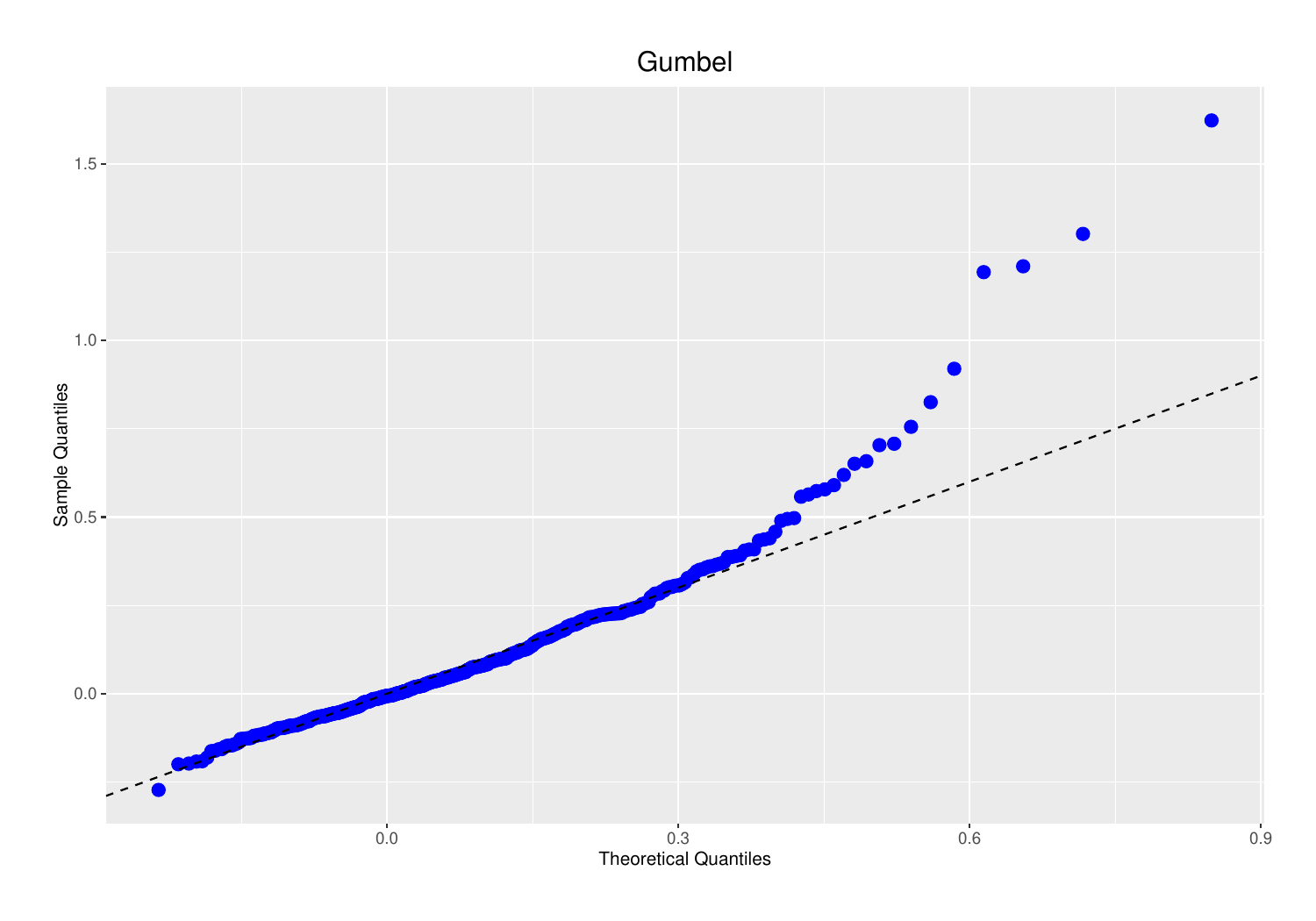}
  \includegraphics[width=.45\textwidth]{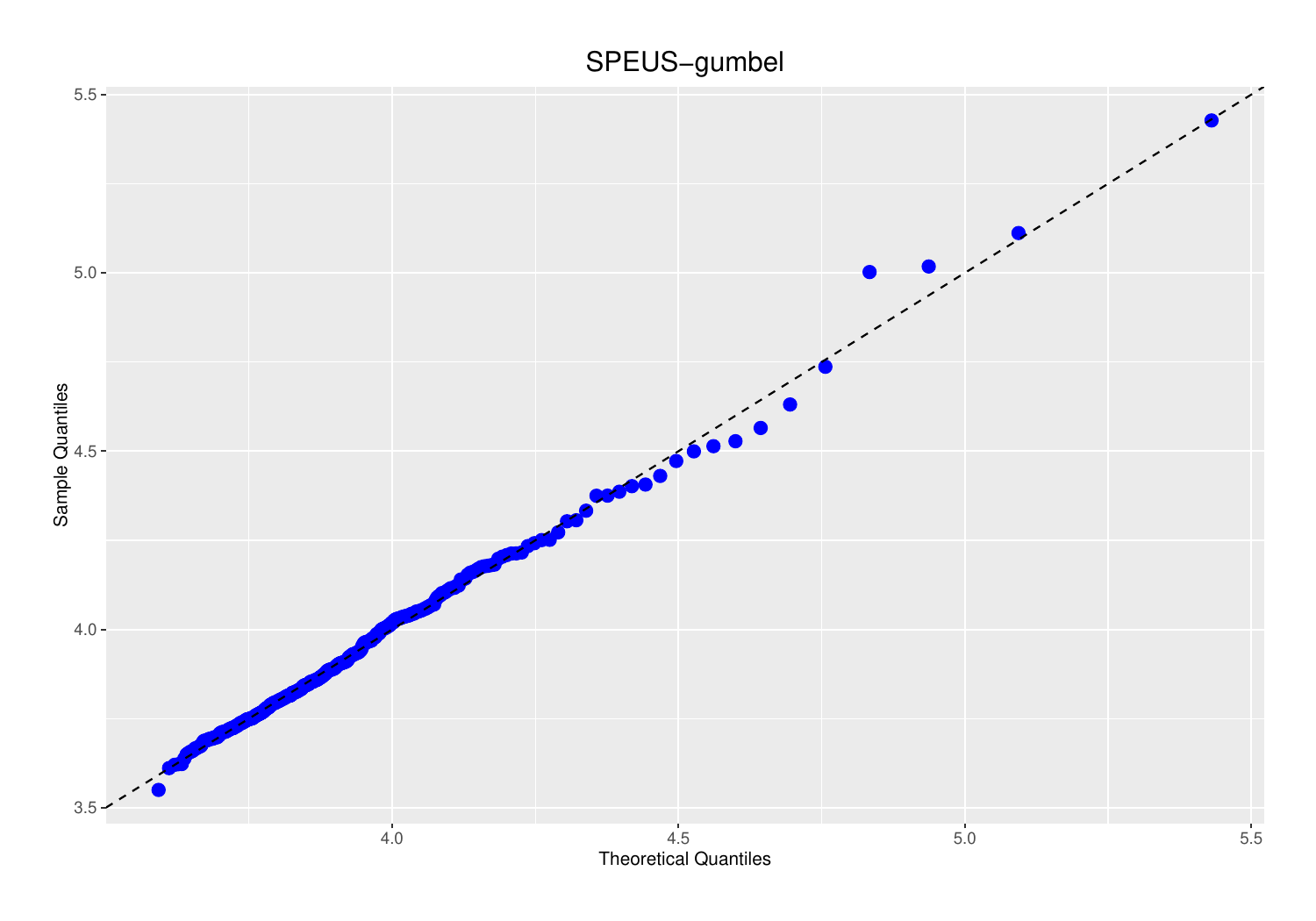}
\caption{\footnotesize Q-Q plots of residuals for Gumbel (left) and SPEUS (right) on air pollution  data.}
\label{airqqplot-ori}
\end{figure}

In real-world applications, data scientists frequently confront challenges related to skewness. However,   fitting a   common distribution that inadequately addresses these distortions can lead to subpar model fits and  misleading inferences. The skewed pivot-blend technique refines skewness management in these distributions, thus enhancing the accuracy and reliability of the models.

{
\bibliographystyle{apalike}
\bibliography{speus}
}

%
%

\end{document}